\documentclass[reqno]{amsart}

\usepackage[english]{babel}
\usepackage{mathtools,amsfonts,amsthm,amssymb,dsfont,mathrsfs,hyperref} 
\usepackage[margin=3cm]{geometry}
\usepackage{csquotes}
\usepackage{verbatim}
\usepackage{enumitem}
\usepackage{xcolor}

\newtheorem{thm}{Theorem}
\newtheorem{lem}[thm]{Lemma}
\newtheorem{prop}[thm]{Proposition}
\newtheorem{cor}[thm]{Corollary}
\newtheorem{df}[thm]{Definition}

\newtheorem{remark}[thm]{Remark}

\renewcommand{\a}{a^{\vphantom{*}}}
\renewcommand{\b}{b^{\vphantom{*}}}
\newcommand{\id}{\mathrm{id}}
\newcommand{\Bog}{\mathrm{Bog}}
\newcommand{\BEC}{\mathrm{BEC}}

\def\Tr{{\mathrm{Tr}}}

\def\d{{\mathrm{d}}}

\begin{document}

\title[Random field approximation and local counting statistics for the Bose Gas]{Random field approximation and local counting statistics for the weakly interacting thermal Bose gas}

\author{Andreas Deuchert}
\address{Virginia Tech}
\email{andreas.deuchert@vt.edu} 

\author{Marcin Napiórkowski}
\address{ University of Warsaw}
\email{marcin.napiorkowski@fuw.edu.pl} 

\author{Błażej Ruba}
\address{ University of Warsaw}
\email{bruba@fuw.edu.pl} 
\date{\today}

\begin{abstract} We study weakly interacting bosons on the three-dimensional torus at temperatures proportional to the critical temperature for Bose--Einstein condensation. We show that, as the expected particle number tends to infinity, the grand canonical Gibbs state is asymptotically described by the coherent state quantization of a random field that converges to a novel Bogoliubov random field. Moreover, the point process associated with the Gibbs state can be approximated by a Cox process. These results are based on and extend the approximation of the Gibbs state recently obtained by the first two authors and Nam. Using these approximations, we compute limiting distributions of the number of particles outside the condensate, the joint occupation statistics of finitely many momentum modes, and particle-number fluctuations in macroscopic subsets of the torus, all governed by the Bogoliubov field. We further prove convergence of the empirical measure of the associated point process to the uniform measure and convergence of the microscopic particle-number statistics to the boson point process. The latter result establishes the universality of particle-number statistics on microscopic length scales.
\end{abstract}

\maketitle

\tableofcontents

\section{Introduction}
The wave function in quantum mechanics is a solution of the deterministic Schrödinger equation. However, predictions derived from the wave function are probabilistic in nature. An additional layer of randomness arises when one considers quantum systems at positive temperature. The Gibbs state of such a system is a mixed state, that is, a statistical mixture of rank-one projections (pure states). Understanding the resulting probability distributions in various  scaling limits leads to rich and typically very challenging mathematical problems. 

In this article, we consider a weakly interacting Bose gas on the unit torus at temperatures comparable to the critical temperature for the Bose--Einstein condensation (BEC) phase transition. Our goal is to determine the asymptotic behavior of several physically relevant probability distributions associated with the system's grand canonical Gibbs state. Our first main result is a trace-norm approximation of the Gibbs state by the coherent-state quantization of a certain random field with two contributions, one arising from the condensate and the other from the thermal cloud. The condensate contribution is described in terms of the Gibbs distribution of a one-mode $\Phi^4$ theory, which depends on the interaction between the particles, whereas the thermal-cloud contribution converges, as the expected particle number $N \to \infty$, to a Gaussian random field that we call the Bogoliubov field. The latter depends on both the condensate fraction and the interaction between the particles, and reduces to the Gaussian free field when either of them vanishes. This approximation is semiclassical in nature and builds on and extends the Gibbs-state approximation established by the first two authors and Nam in \cite[Theorem~1]{DeuNamNap-25}.

Associated with the Gibbs state is a point process describing the statistics of particle numbers in subsets of the torus. Our approximation of the Gibbs state allows us to approximate this point process by a Cox process -- a generalization of a Poisson point process in which the intensity measure is itself random. 

These approximations allow us to derive limiting distributions for random variables associated with field operators, the number of particles outside the condensate, the joint occupation statistics of finitely many momentum modes, and the particle numbers in macroscopic subsets of the torus. All these limiting distributions admit a natural description in terms of the Bogoliubov field and a contribution arising from the BEC. 

We also show that the empirical measure of the point process associated with the Gibbs state converges to the uniform measure on the torus, and that the particle-number statistics on the microscopic length scale $N^{-1/3}$ converge to the boson point process. Since the boson point process is independent of the interaction potential, this establishes the universality of microscopic particle-number statistics. This answers in the affirmative a question raised in the physics literature \cite[Section~6.2.3]{BardanetEtal2022}, namely whether there exists a parameter regime for an interacting quantum many-particle system in which the local particle-number statistics are universal and can be approximated by a boson point process. 

The above results hold for interacting systems with a non-zero interaction between the particles. For the non-interacting system, large fluctuations in the number of condensed particles lead to different behavior. We emphasize that our results provide quantitative control of the corresponding probability laws, including convergence in total variation and Wasserstein distances.

The level of accuracy achieved in our work, in particular the control of particle-number statistics down to the microscopic length scale, should be viewed in comparison with related developments in the statistical mechanics of classical particle systems. We refer to \cite{BauerBourgNikulaYau2017,LeblSerf-2017,Serfaty-2018,LeblSerf2018,BauerBourgNikulaYau2019,Chat2019,ArmstFerf2021,Gau-2021,Lewin2022,Seraty-2024} for results on classical particle systems, including Coulomb and Riesz gases, and to \cite{Johansson1998,RidVirag-2007,Forrester-2010,AmHedenMaka-2011,BorGui2013,BourgErdYau-2014,BekFigGui2015,AmHedenMaka-2015,BekLebSerf-2018,LamLedWebb-2019} for results on $\beta$-ensembles. Related results for non-interacting quantum systems have been obtained for fermions (determinantal point processes) in \cite{Eisler2013,DelLam-2024,DelLam-2024b,DelLam-2025} and for bosons (permanental point processes) in \cite{ShiTaka-2004,TamuraIto2006,TamuraIto2007,TamuraZagrebnov2012}. 

Probabilistic questions for the Bose gas have previously been studied in the context of ground states (the zero temperature case) and the dynamics of systems exhibiting complete BEC (almost all particles in the condensate). We refer to \cite{ArousKirkSchl2013,BuchSaffSchlein2014,RaSchl-2019,Radem2020} for central limit theorems for quantum observables, to \cite{KirkRadeSchlein2021,RaSei-2022,NamRade2025,Radem2025,BehrBreRade2025} for tail bounds for BECs, to \cite{BossPetrat2023} for the justification of an Edgeworth expansion, and to \cite{Radem2025GP} for the computation of the limiting moment generating function of the number of particles outside the condensate in the Gross--Pitaevskii regime. The convergence and fluctuations of empirical measures associated with general quantum observables were recently studied in the Gross--Pitaevskii limit in \cite{PortRadeViro2023}, where the authors also established convergence of the laws of the relevant random variables in a Wasserstein distance.

The results in the previous paragraph are based on earlier works concerned with the justification of Bogoliubov theory \cite{Bogoliubov-47} for the Bose gas in different parameter regimes. Bogoliubov theory is an approximation scheme that replaces the full Hamiltonian by one whose spectrum can be computed explicitly (the Bogoliubov Hamiltonian). This allows one to compute physical observables, such as the low-lying excitation spectrum and the condensate depletion, for weakly interacting and dilute Bose gases with high accuracy. Due to impressive progress over the last two decades, Bogoliubov theory for the ground state is now rather well understood from a mathematical perspective, see \cite{Seiringer-11,GreSei-13,LewNamSerSol-15,DerNap-13,BroSei-2022} for mean-field systems, \cite{BocBreCenSch-20,BocBreCenSch-19,NamNapRicTri-21,NamTri-23,BreSchSch-22,BreSchSch-22a,HaiSchTri-22,CarOlgSAubSchl-2024} for Bose gases in the Gross--Pitaevskii limit, and \cite{YauYin-09,BasCenSch-21,BroOldSAubSchl2026,BaBroCenaOlgSchl2026,FouSol-20,FouSol-23} for dilute Bose gases in the thermodynamic limit. We also refer to \cite{HHNST-23,HHST-24} for a free energy expansion for the latter system at very low temperatures. For a review of Bogoliubov theory in the context of quantum dynamics, see \cite{Napiorkowski-23}.

The problem of understanding bosonic many-particle systems and their statistical properties becomes substantially more difficult if a system at a positive temperature is considered. In this case the condensate depletion is not small and a macroscopic fraction of the particles resides outside the condensate. As suggested by Lee and Yang in \cite{LeeYan-58} a c-number substitution is still possible. However, the temperature dependence of the condensate fraction has to be taken into account, which leads to a~temperature-dependent Bogoliubov Hamiltonian. In the physics literature this approximation is called the Popov approximation \cite{Andersen2004,Popov1983}. As is pointed out in \cite{Andersen2004} it is expected to break down when the temperature approaches the critical point of the BEC phase transition.

In \cite{DeuNamNap-25}, it was shown that the Popov approximation is valid for a system of weakly interacting bosons in the mean-field limit at all temperatures comparable to the critical temperature of the BEC phase transition. This includes temperatures below, above, and at the critical point. More precisely, it was proved that the Gibbs state of the interacting model can be approximated in trace norm by a convex combination of products of coherent states and Gibbs states associated with temperature-dependent Bogoliubov Hamiltonians. The convex combination is expressed as an integral with respect to the Gibbs distribution of a one-mode $\Phi^4$ theory, which depends on the interaction between the particles. This Gibbs distribution accurately describes the interacting condensate and was recently introduced in \cite{BocDeuSto-24,CapDeu-23} in the construction of trial states for the free energy of systems in the Gross--Pitaevskii limit. Because of the convex combination, the state approximating the Gibbs state is not quasi-free. Nevertheless, all its correlation functions can be computed explicitly. As an application of the Gibbs-state approximation, the limiting distribution of the number of particles in the condensate was derived.
 
In this paper we extend the results described above to obtain a description of the Bose gas in terms of the Gaussian Bogoliubov field for all temperatures that are proportial to the critical temperature of its BEC phase transition. Recently, a slightly different scaling regime has been considered in \cite{LewNamRou-21,FroKnoSchSoh-22}, where effective descriptions in terms of non-linear field theories are obtained for a weakly interacting system at the critical point (cf. Remark~\ref{rem:discussionMainResults} (i)). These techniques have been extended to derive $\Phi^4$ theories in two and three dimensions  \cite{FroKnoSchSoh-25,NamZhuZHu2025,CarKnoRanaGies2026,JouRou-2026} (see also \cite{FroKnoSchSoh-19,RoutSoh-2025,NamYangZhu-2026,LuNamZhu-2026,NamZhuZHu2026}). 
 
For other studies of the Bose gas at positive temperature, we refer to the following works, beginning with the weakly interacting regime. In \cite{DeuSei-21}, the dependence of the critical temperature on the interaction is analyzed in a semiclassical mean-field limit. Spontaneous symmetry breaking for a mean-field system is established in \cite{DeuNapNam-25}. The dynamics of a mean-field Bose gas initially prepared in a positive-temperature state is studied in \cite{DeuCapSchlein-24}. In the singular Gross--Pitaevskii limit, the BEC phase transition was established for trapped systems in \cite{DeuSeiYng-19} and for homogeneous systems in \cite{DeuSei-20}. Free-energy expansions for dilute Bose gases in the thermodynamic limit in two and three spatial dimensions were derived in \cite{Seiringer-08,Yin-10,BaBocCenaDeu2025,DeuMaySei-20,MaySei-20}. 

Finally, we refer to \cite{Hepp1974,GinVelo1979,AmmNie-08,AmmCorFalcGau2025} for a selection of works on the semiclassical analysis of bosonic systems.

In the next section, we present the precise mathematical setup, define the point process associated with the Gibbs state, introduce Cox processes, the boson point process, and the Bogoliubov field, and state and discuss our main results.

\section{Setup and Main Results}
\label{sec:setupAndMainResults}

\subsection{Fock space and Hamiltonian}

We consider a system of bosonic particles captured in the three-dimensional flat torus $\Lambda = (\mathbb R / \mathbb Z)^3$. The one-particle Hilbert space is $L^2(\Lambda)$. We are interested in a~gas with a fluctuating particle number, and therefore choose the bosonic Fock space
\begin{equation}
    \mathscr{F}(L^2(\Lambda)) = \bigoplus_{n=0}^{\infty} L^2_{\mathrm{sym}}(\Lambda^n)
    \label{eq:FockSpace}
\end{equation}
as the Hilbert space of the system. Here $L^2_{\mathrm{sym}}(\Lambda^n)$ denotes the closed linear subspace of $L^2(\Lambda^n)$ consisting of those functions $\psi(x_1,...,x_n)$ that are invariant under any permutation of the particle coordinates $x_1, ..., x_n \in \Lambda$. As usual, we set $L^2_{\mathrm{sym}}(\Lambda^0) = \mathbb{C}$. We also let $\mathscr{F}_n = L^2_{\mathrm{sym}}(\Lambda^n)$.

On the bosonic Fock space we denote by $a_p^*$ and $a_p$ the usual creation and annihilation operators of a~particle with momentum $p \in \Lambda^* = 2 \pi \mathbb{Z}^3$, which satisfy the canonical commutation relations
\begin{equation}
	[a_p,a_q^*] = \delta_{p,q}, \quad \quad [a_p,a_q] = [a^*_p,a^*_q]=0.
	\label{eq:CCR}
\end{equation}

The Hamiltonian of the system is given by 
\begin{equation}
    \mathcal H_N = \sum_{p \in \Lambda^*} p^2 a_p^* \a_p + \frac{1}{2N} \sum_{p,u,v \in \Lambda^*} \widehat v(p) a^*_{u-p} a^*_{v+p} \a_u \a_v.
    \label{eq:Hamiltonian}
\end{equation}
Here $\widehat{v}(p) = \int_{\Lambda} e^{-\mathrm{i}p  \cdot x} v(x) \d x$ denotes the Fourier transform of the interaction potential $v : \Lambda \to \mathbb R_+$. We~assume that $v$ is an even real-valued function and that its Fourier transform satisfies $\widehat v(p) \geq 0 $ and $\sum_{p \in \Lambda^*}(1+|p|) \widehat v(p) < \infty$. The parameter $N > 0$ in $\mathcal H_N$ will later be chosen as the expected number of particles in the system, and therefore implements a mean-field scaling. The operator $\mathcal{H}$ acts on a~suitable dense domain in $\mathscr{F}$ and is self-adjoint. Its restriction to $\mathscr{F}_n$ is given by
\begin{equation}
  \left.  \mathcal{H}_N \right|_{\mathscr{F}_n} = \sum_{i=1}^n -\Delta_i + \frac{1}{N} \sum_{1 \leq i < j \leq n} v(x_i - x_j),
    \label{eq:nParticleHamiltonian}
\end{equation}
where $\Delta_i$ is the Laplacian on $\Lambda$ acting on the coordinate $x_i$ of the $i$-th particle. 

\subsection{Grand canonical Gibbs state}
We are interested in the grand canonical Gibbs state
\begin{equation} 
    G_{\beta, N} = \frac{e^{- \beta (\mathcal H_N - \mu_{\beta,N} \mathcal N)}}{\mathrm{Tr}[ e^{- \beta (\mathcal H_N - \mu_{\beta,N} \mathcal N) } ]}
    \label{eq:GibbsState}
\end{equation}
defined on $\mathscr{F}(L^2(\Lambda))$. Here $\mathcal{N} = \sum_{p \in \Lambda^*} a_p^* \a_p$ denotes the particle number operator, $\beta > 0$ is the inverse temperature, and $\mu_{\beta,N} \in \mathbb{R}$ the chemical potential. We prefer to use the expected particle number as the parameter instead of the chemical potential, and therefore we choose $\mu_{\beta,N}$ uniquely determined by the condition $\Tr[\mathcal{N} G_{\beta,N}] = N$. 

The weakly interacting Bose gas, described by $G_{\beta,N}$, exhibits a BEC phase transition. To~state this precisely, we introduce the one-particle density matrix (1-pdm) $\gamma_{\beta,N}$ of $G_{\beta,N}$. The 1-pdm is a~positive operator on $L^2(\Lambda)$ with trace $N$, and can be defined via its integral kernel in Fourier space: $\widehat{\gamma}_{\beta,N}(p,q) = \Tr[a_q^* \a_p G_{\beta,N}]$. We consider the limit $N \to \infty$ such that
\begin{equation}
  \lim_{N \to \infty}  \frac{\beta(N)}{\beta_c(N)} = \kappa \in (0,\infty), \qquad \text{with} \qquad  \beta_c(N) = \frac{1}{4 \pi} \left( \frac{\zeta(3/2)}{N} \right)^{2/3},
    \label{eq:beta_asymp}
\end{equation}
where $\zeta$ is the Riemann zeta function. As has been shown in \cite[Theorem~4]{DeuNamNap-25}, we have
\begin{equation}
    \lim_{N \to \infty} \frac{\sup_{\Vert \psi \Vert = 1} \langle \psi, \gamma_{\beta,N} \psi \rangle}{N} = \lim_{N \to \infty} \frac{ \langle \varphi_0, \gamma_{\beta,N} \varphi_0 \rangle}{N}  = \left[ 1 - \kappa^{-3/2} \right]_+ \eqqcolon g(\kappa).
    \label{eq:condensateFraction}
\end{equation}
Here $\varphi_0(x)=1$ denotes the $L^2$-normalized constant function on $\Lambda$ and $[x]_+ = \max\{0,x\}$. Thus, the largest eigenvalue of $\gamma_{\beta,N}$ grows linearly in $N$
 if and only if $\kappa > 1$. In particular, a BEC phase transition occurs at $\kappa = 1$. The condensate wave function, i.e., the eigenfunction corresponding to the largest eigenvalue of $\gamma_{\beta,N}$, is the constant function. The quantity $g(\kappa)$ is the (limiting) condensate fraction. Note that $\beta_c(N)$ and $g(\kappa)$ do not depend on the interaction potential. In particular, they coincide with the corresponding quantities for the non-interacting system.

\subsection{Point process associated with the Gibbs state} \label{sec:Gibbs_point_process}
To the Gibbs state $G_{\beta,N}$ we associate a~point process $\Theta_{\beta,N}$ on $\Lambda$, which describes the spatial distribution of the particles. A point process is a random variable whose values are locally finite point configurations. Equivalently, it may be viewed as a random locally finite counting measure, assigning to each bounded Borel set the number of particles it contains. Here we will take the latter perspective. For an introduction to point processes we refer to \cite{DaleyVere-Jones2003,DaleyVere-Jones2008}. 

Since $[G_{\beta,N},\mathcal{N}] = 0$, the Gibbs state admits the decomposition $G_{\beta,N} = \sum_{n=0}^{\infty} p_n G_n$, where $G_n$ is a positive operator on $\mathscr{F}_n$ satisfying $\Tr[G_n] = 1$, and $p_n$ are probabilities: $p_n \geq 0$, and $\sum_{n=0}^{\infty} p_n = 1$. Let $M$ be a random variable with distribution $\mathbf P(M=n)=p_n$ for all $n \in \mathbb{N}_0$. For each $n \ge 0$, let 
\begin{equation}
    (X_1^{(n)},\ldots,X_n^{(n)})
    \label{eq:definitionX}
\end{equation}
be a random vector in $\Lambda^n$ whose distribution is given by
\begin{equation}
\mathbf P\bigl( (X_1^{(n)},\ldots,X_n^{(n)})\in B_n \bigr)
= \int_{B_n} G_n(x_1,\ldots,x_n;x_1,\ldots,x_n) \d (x_1, \cdots, x_n),
\end{equation}
where $B_n \subseteq \Lambda^n$ is a Borel set and $G_n(x_1,\ldots,x_n;x_1,\ldots,x_n)$ is the diagonal of the integral kernel of~$G_n$. The \textbf{point process} $\Theta_{\beta,N}$ is then defined by 
\begin{equation}
\mathcal{L}\left( \Theta_{\beta,N} | M = n \right) = \mathcal{L}\left( \sum_{j=1}^{n} \delta_{X_j^{(n)}} \right),
\label{eq:GibbsPointProcess}
\end{equation}
where $ \mathcal{L}$ denotes the law the point process and $\delta_x$ is the Dirac measure at $x\in\Lambda$. Properties of $\Theta_{\beta,N}$ and further details on its relation to the Gibbs state $G_{\beta,N}$ are discussed in Appendix~\ref{app:GibbsPointProcess}. Among other things, we show there that the two objects have the same correlation functions.

\subsection{Cox point processes}
\label{sec:CoxPointProcesses}
In Theorem~\ref{thm:semiclassicalApproximation}~(c) below, we show that the point process $\Theta_{\beta,N}$ in \eqref{eq:GibbsPointProcess} can be approximated by a Cox point process. A \textbf{Cox process directed by a random measure $\lambda$} is a point process $\Gamma$ whose law, conditional on $\lambda$, is that of a Poisson point process with intensity measure $\lambda$. For introductions to Poisson point processes and Cox point processes we refer to \cite[Chapters~2.4~and~6.2]{DaleyVere-Jones2003}.

\subsection{Boson point process}
\label{sec:bosonPointProcess}
In Theorem~\ref{thm:localcounting}~(c) below, we show that at the microscopic length scale $N^{-1/3}$, the point process $\Theta_{\beta,N}$ associated with the Gibbs state converges, as $N \to \infty$, to the boson point process \cite{Macchi1971,Macchi1975,ShiTaka-2003a,TamuraIto2006,TamuraIto2007,TamuraZagrebnov2012,DaleyVere-Jones2003}, which we discuss next.

Let $\beta_{\infty}(\kappa) = \frac{\kappa}{4 \pi} \zeta^{2/3}(3/2)$, and let $H(x)$ be a centered complex-valued Gaussian random field on $\mathbb{R}^3$ with covariance
\begin{equation}
          \mathbf{E}[ \overline{H(x)} H(y) ] =  \int_{\mathbb{R}^3} \frac{e^{\mathrm{i} p \cdot (x-y)}}{e^{\beta_{\infty}(\kappa)p^2}-1} \frac{\d p}{(2\pi)^3},
\label{eq:covarianceH}
\end{equation}
and with vanishing pseudo-covariance: $\mathbf{E}[ H(x) H(y) ]=0$. Since the kernel in \eqref{eq:covarianceH} is smooth, the random field $H$ satisfies $H \in \mathcal{C}^{\infty}(\mathbb{R}^3)$ almost surely, see e.g.~\cite[Theorem~1.4.2]{Adler-2007}.

We also define the random measure
\begin{equation}
    I(D) = \int_{D} | \sqrt{g(\kappa)} + H(x) |^2 \d x,
\end{equation}
where $g(\kappa)$ is the condensate fraction, defined in \eqref{eq:condensateFraction}. The \textbf{boson point process} $\Xi$ is a Cox process directed by the measure $I$. Since $I$ is given in terms of the non-centered Gaussian random field $\sqrt{g(\kappa)} + H(x)$, the boson point process is also a non-centered permanental point process. 

\subsection{Bogoliubov field}
\label{sec:BogoliubovFreeField}
Theorem~\ref{thm:semiclassicalApproximation} below establishes a relation between the grand canonical Gibbs state $G_{\beta,N}$ and a~complex-valued Gaussian random field, which we now define. We choose a subset $A \subset \Lambda^*_+ = \Lambda^* \backslash \{ 0 \}$ such that $A \cup A^{\mathrm{c}} = \Lambda_+^*$ and $p \in A$ implies $-p \notin A$. That is, $A$ contains one element from each pair $p,-p$. Let $\{ (\widehat \Psi(p), \widehat \Psi(-p)) \}_{p \in A}$ be a family of independent, centered complex Gaussian random vectors with covariance and pseudo-covariance matrices 
\begin{align}
	 & \mathbf{E} \begin{pmatrix}
		|\widehat \Psi(p)|^2 & \widehat \Psi(p) \overline{\widehat \Psi(-p)} \\ \overline{\widehat \Psi(p)} \widehat \Psi(-p) & |\widehat \Psi(-p)|^2 		\end{pmatrix} = \frac{p^2 + g(\kappa) \widehat v(p)}{p^2 (p^2 + 2 g(\kappa) \widehat v(p))}  \begin{pmatrix}
	1 & 0 \\ 0 & 1 
	\end{pmatrix},  \nonumber \\
	 & \mathbf{E} \begin{pmatrix}
	\widehat \Psi(p)^2 & \widehat \Psi(p) \widehat \Psi(-p) \\ \widehat \Psi(p) \widehat \Psi(-p) & \widehat \Psi(-p)^2 		
	\end{pmatrix} = -\frac{g(\kappa)\widehat v(p)}{p^2 (p^2 + 2g(\kappa) \widehat v(p))}  \begin{pmatrix}
	0 & 1 \\ 1 & 0 
	\end{pmatrix},
\label{eq:familyOfGaussianRandomVariables}
\end{align}
where $g(\kappa)$ is defined in \eqref{eq:condensateFraction}.

We define the \textbf{Bogoliubov field} as the random Fourier series (note the absence of the $p=0$ mode)
\begin{equation}
    \Psi(x) = \sum_{p \in \Lambda_+^*} \widehat \Psi(p) e^{\mathrm{i} p \cdot x}.
    \label{eq:BogoliubovFreeField}
\end{equation}
As in the case of the Gaussian free field, a typical realization of the Bogoliubov field is not a function but only a distribution. More precisely, for every $s < -1/2$ we have $\Psi \in H^{s}(\Lambda)$ almost surely, where $H^{s}(\Lambda)$ denotes the Sobolev space of distributions $f$ such that $(1-\Delta)^{-\frac{s}{2}} f \in L^2(\Lambda)$. 

The Bogoliubov field $\Psi(x)$ is a natural generalization of the complex-valued Gaussian free field, to~which it reduces if either $g(\kappa)=0$ (non-condensed phase) or $v = 0$ (non-interacting system). As its name and the formulas above suggest, it is closely related to Bogoliubov theory for the Bose gas. 

One of our results is formulated in terms of the \textbf{Wick product} $: | \Psi(x)|^2 :$ of the fields $\overline{\Psi(x)}$ and~$\Psi(x)$, which is a random field defined as follows. For $q>0$, let $\Psi^{< q}(x) = \sum_{p \in \Lambda_+} \mathds{1}(|p| < q) \widehat \Psi(p) e^{\mathrm{i} p \cdot x}$, and~define
\begin{equation}
    : | \Psi(x)|^2 : = \lim_{q \to \infty} \left( | \Psi^{< q}(x) |^2 - \mathbf{E}[ | \Psi^{< q}(x) |^2 ] \right).
    \label{eq:WickProduct}
\end{equation}
The right-hand side tested against any $f \in H^{-\frac12}(\Lambda)$ converges as $q \to \infty$. Wick products and Wick ordering play an important role in constructive quantum field theory, stochastic PDEs, and Gaussian multiplicative chaos \cite{GlimmJaffe1987,Hairer2014,Janson2009,RhodesVargas2014}.

\subsection{Main results part I: the condensed phase}
\label{sec:mainResults}
Our first main result is an approximation of the Gibbs state $G_{\beta,N}$ in terms of the coherent state quantization of an $N$-dependent random field that converges to the Bogoliubov field in the limit $N \to \infty$. Before giving the precise statement, we introduce the following notation.

We define the probability distribution on $\mathbb C$ with density
\begin{equation}
    g^{\mathrm{BEC}}(z) = \frac{\exp\left( - \beta \left( \frac{\widehat{v}(0)}{2N} |z|^4 - \mu^{\mathrm{BEC}} |z|^2 \right) \right)}{\int_{\mathbb{C}} \exp\left( - \beta \left( \frac{\widehat{v}(0)}{2N} |w|^4 - \mu^{\mathrm{BEC}} |w|^2 \right) \right) \d w}.
    \label{eq:gBEC}
\end{equation}
Here $\d w = \d x \d y/\pi$, where $w = x + \mathrm{i}y$, denotes two-dimensional Lebesgue measure divided by $\pi$, and $\mu^{\mathrm{BEC}} \in \mathbb{R}$. The function $g^{\mathrm{BEC}}$ can be interpreted as the Gibbs distribution of a one-mode $\Phi^4$ theory. It has been used in \cite{BocDeuSto-24,CapDeu-23, DeuNamNap-25} to approximate interacting BECs in the Gross--Pitaevskii and mean-field limits. 

To every $\phi \in L^2(\Lambda)$ we associate the coherent state
\begin{equation}
    \Omega_{\phi} =  \exp\left(  \sum_{p \in \Lambda^*} \left( \widehat \phi(p)  a_p^* - \overline{\widehat \phi(p)} \a_p \right) \right) \Omega, 
    \label{eq:coherentState}
\end{equation}
where $\Omega \in \mathscr{F}$ denotes the Fock vacuum. Moreover, we set 
\begin{equation}
    L^2_+(\Lambda) = \mathds{1}(-\Delta \neq 0) L^2(\Lambda),
    \label{eq:L2plus}
\end{equation}
where $\mathds{1}(-\Delta \neq 0)$ denotes the spectral projection of the Laplacian onto the non-zero modes. 

For simplicity, we first state all theorems in the condensed phase ($\kappa>1$), where our results take their most interesting form. Results for $\kappa\in(0,1]$ are presented in Section~\ref{sec:mainNonCondensed}.

\begin{thm}[\textbf{Semiclassical approximation of Gibbs state and point process}]
\label{thm:semiclassicalApproximation}
    Let the interaction potential $v $ be a~nonnegative, even, $1$-periodic function in $L^1(\Lambda)$ whose Fourier coefficients $\widehat{v}(p)$ are nonnegative and satisfy $\sum_{p \in \Lambda^*} (1+|p|) \widehat{v}(p) <  \infty$ and $\hat{v}(0) > 0$. We consider the limit $N \to \infty$, $\beta/ \beta_{\mathrm{c}}\to \kappa \in (1,\infty)$, see \eqref{eq:beta_asymp}. 
    \begin{enumerate}[label=(\alph*)]
    \item \textbf{Approximation of Gibbs state.} There exists a Gaussian random field $\Phi_N(x)$ with law $\mathbf{P}_{\Phi_N}$, explicitly defined in \eqref{eq:product_measure} below,  such that $\Phi_N \in L^2_+(\Lambda)$ almost surely and 
    \begin{equation}
        \left\Vert G_{\beta,N} - \int | \Omega_{\phi_z} \rangle \langle \Omega_{\phi_z} | g^{\mathrm{BEC}}(z) \d \mathbf{P}_{\Phi_N}(\phi) \d z \right\Vert_1 \lesssim N^{-\frac{1}{48}},
        \label{eq:approximate_Gibbs}
    \end{equation}
    where $\Vert \cdot \Vert_1$ denotes the trace norm. Here $\phi_z(x) = z + \frac{z}{|z|} \phi(x)$, and 
    $\mu^{\mathrm{BEC}}$ in the definition \eqref{eq:gBEC} of $g^{\mathrm{BEC}}$ is chosen such that 
    \begin{equation}
        \int_{\mathbb{C}} |z|^2 g^{\mathrm{BEC}}(z) \d z + \int \Vert \phi \Vert_{L^2(\Lambda)}^2 \mathrm{d} \mathbf{P}_{\Phi_N}(\phi) = N .
    \end{equation}
    \item \textbf{Convergence to the Bogoliubov field.} The random field $\Phi_N$ and the Bogoliubov field $\Psi$ can be realized on a common probability space such that, for all $s < -\frac12$ and $p \in [1,\infty)$, 
    \begin{equation}
        \lim_{N \to \infty} \mathbf{E}\!\left[
            \left\Vert \beta^{1/2} \Phi_N - \Psi \right\Vert_{H^{s}(\Lambda)}^p \right] = 0.
    \end{equation}
    Explicit convergence rates and other modes of convergence are given in Section \ref{sec:semiclassicalApproximation}. 
    \item \textbf{Approximation of point process.} Let $\widetilde{\Theta}_{\beta,N}$ be the \textbf{Cox point process} on $\Lambda$ directed by the random measure 
    \begin{equation}
        \d \lambda(x) = \left| \ |Z| + \Phi_{N}(x) \right|^2 \d x,
    \end{equation}
    where the random variable $Z$ has distribution $g^{\mathrm{BEC}}$ defined in \eqref{eq:gBEC}, and $\Phi_N$ is the random field introduced in parts~(a),~(b). Then, for all bounded Borel-measurable functions $F: \mathbb{R} \to \mathbb{R}$ and $f:\Lambda \to \mathbb{R}$, we have
    \begin{equation}
        \left| \mathbf{E}[ F( \Theta_{\beta,N}(f) ) ] - \mathbf{E}[ F( \widetilde{\Theta}_{\beta,N}(f) ) ] \right| \lesssim N^{-1/48} \Vert F \Vert_{\infty},
        \label{eq:point_process_comparison}
    \end{equation}
    where $\Theta_{\beta,N}(f) = \int_{\Lambda} f(x) \mathrm{d} \Theta_{\beta,N}(x)$.
    \end{enumerate}
\end{thm}

The above theorem establishes a connection between the Gibbs state $G_{\beta,N}$ and the random field $Z + (Z/|Z|) \Phi_N(x)$, where $\Phi_N$ is related to the Bogoliubov field $\Psi$ and the distribution of $Z$ is given by $g^{\mathrm{BEC}}$ in \eqref{eq:gBEC}. Moreover, the point process associated with the Gibbs state can be approximated by a Cox process whose random intensity measure is determined by the same random field. This connection allows us to derive limiting distributions for several physically relevant observables. In Theorem~\ref{thm:thermaldistr}, we study the limiting distributions for observables related to different non-zero momentum modes, while in Theorem~\ref{thm:localcounting}, we determine the limiting distributions of the number of particles contained in microscopic and macroscopic subsets of the torus. 

To make this precise, we introduce the following notation. Let $B_1,\dots,B_n \subset \mathbb{R}$ be Borel sets and let $f_1,\dots,f_n \in L^2_+(\Lambda)$ be an orthonormal family. We define the random vector $(A_{f_1},\dots,A_{f_n})$ associated with the field operators $a(f_j) + a^*(f_j)$, $j=1,\dots,n$, and the Gibbs state $G_{\beta,N}$ by
\begin{equation}
    \mathbf{P}(A_{f_j} \in B_j \text{ for } j=1,\dots,n)
    = \Tr\!\left[ \prod_{j=1}^n \mathds{1}\big( a(f_j) + a^*(f_j) \in B_j \big)\, G_{\beta,N} \right].
\end{equation}
This is well defined because $[a(f_i),a^*(f_j)] = \langle f_i, f_j \rangle = 0$ for $i \neq j$. For $p_1,\dots,p_n \in \Lambda^*_+$, we similarly define the random vector $(N_{p_1},\dots,N_{p_n})$ by
\begin{equation}
    \mathbf{P}( N_{p_j} \in B_j \text{ for } j=1,\dots,n )
    = \Tr\!\left[ \prod_{j=1}^n \mathds{1}\big(a_{p_j}^* \a_{p_j} \in B_j\big)\, G_{\beta,N} \right].
\end{equation}
Finally, the random variable $N_+$ (the number of particles outside the BEC), associated with the operator $\mathcal{N}_+ = \sum_{p\in\Lambda_+} a_p^* \a_p$, is defined by
\begin{equation}
    \mathbf{P}( N_+ \in B )
    = \Tr\!\left[ \mathds{1}(\mathcal{N}_+ \in B)\, G_{\beta,N} \right],
\end{equation}
for any Borel set $B \subset \mathbb{R}$. 

\begin{thm}[\textbf{Limiting distributions for the thermal cloud}] \label{thm:thermaldistr}
    Let the interaction potential $v$ satisfy the assumptions of Theorem~\ref{thm:semiclassicalApproximation}. We consider the limit $N \to \infty$, $\beta/ \beta_{\mathrm{c}}\to \kappa \in (1,\infty)$, see \eqref{eq:beta_asymp}.
    \begin{enumerate}[label=(\alph*)]
    \item \textbf{Distribution of the quantum field.} Let $\theta$ be a random variable, uniformly distributed on $[-\pi,\pi]$ and independent of 
    $\Psi$.
    Let $f_1,...,f_n \in L^2_+(\Lambda)$ be an orthonormal family. As $N \to \infty$, we have 
    \begin{align}
         ( \beta^{\frac12}  A_{f_i})_{i=1}^n \xrightarrow{\mathrm{TV}} \left( 2 \mathrm{Re} \left[ e^{\mathrm{i}\theta}  \langle f_i, \Psi \rangle_{L^2(\Lambda)}  \right] \right)_{i=1}^n.
    \end{align}
    Here $\Psi$ is the Bogoliubov field in \eqref{eq:BogoliubovFreeField} and $\xrightarrow{\mathrm{TV}}$ stands for convergence of the laws of the random variables in total variation norm.
    \item \textbf{Joint distribution of occupation numbers.} As $N \to \infty$, we have convergence of the laws in Wasserstein distance $W_q$ with $q \in [1,4)$:
    \begin{equation}
         ( \beta N_{p_i})_{i=1}^n \xrightarrow{W_q} ( |\widehat \Psi(p_i)|^2)_{i=1}^n,
        \label{eq:NpN-p}
    \end{equation} 
    where $\widehat \Psi(p)$ are the Fourier modes of $\Psi$, see \eqref{eq:familyOfGaussianRandomVariables}. 
    
    Observe that $|\widehat \Psi(p_i)|^2$ and $|\widehat \Psi(p_j)|^2$ 
    are independent unless $p_j = \pm p_i$. 
    The distribution of $( |\widehat \Psi(p)|^2, |\widehat \Psi(-p)|^2)$ has the following density 
    on $[0,\infty)^2$:
    \begin{equation}
        f(x,y) = \left( p^4 + 2 p^2 g(\kappa) \widehat{v}(p) \right) e^{- (p^2 + g(\kappa) \widehat{v}(p))(x+y)} I_0(2 g(\kappa) \widehat{v}(p) \sqrt{xy}),
        \label{eq:Bessel_density}
    \end{equation}
    where $I_0$ is the modified Bessel function of the first kind and $g(\kappa)$ is defined in \eqref{eq:condensateFraction}.
    \item \textbf{Distribution of the number of excited particles.} As $N \to \infty$, we have 
    \begin{equation}
        \beta ( N_+ - \mathbf E[N_+] )  \xrightarrow{W_2} \int_{\Lambda} : |\Psi(x)|^2 : \d x, 
        \label{eq:limitOfN+}
    \end{equation}
    where $: |\Psi(x)|^2 :$ is the Wick square of $\Psi$, see \eqref{eq:WickProduct}. Moreover, the distribution of 
    the right-hand side of \eqref{eq:limitOfN+} coincides with that of
    \begin{equation}
        R = \sum_{p \in A} \left[ \frac{X_p-1}{p^2} + \frac{Y_p-1}{p^2 + 2 g(\kappa) \widehat v(p)} \right],
        \label{eq:W_def} 
    \end{equation}
    where $\{ X_p , Y_p \}_{p \in A}$ are independent exponential random variables with rate parameter one, and the set $A \subset \Lambda^*_+$ is defined above \eqref{eq:familyOfGaussianRandomVariables}.
    \end{enumerate}
\end{thm}

The above theorem shows that the limiting distribution of the quantum field is described by the Bogoliubov field $\Psi$. The random variable $e^{\mathrm{i}\theta}$ appears because the observable $a(f)+a^*(f)$ is not gauge invariant. It restores particle number conservation. This is necessary since, in contrast to the Gibbs state $G_{\beta,N}$ in \eqref{eq:GibbsState}, the field $\Psi$ is not gauge invariant. Rescaled versions of $N_+$ and $N_p$ are likewise described in terms of $\Psi$. The limiting distribution of the condensate occupation number $N_0$ was established in \cite[Theorem~7]{DeuNamNap-25}. Moreover, $N_0$ and $N_p$, $p \in \Lambda_+^*$, become independent as $N \to \infty$.

Next, we study the point process $\Theta_{\beta,N}$ in \eqref{eq:GibbsPointProcess} associated with the Gibbs state $G_{\beta,N}$. Let $D \subset \Lambda$ be Borel measurable and let $N_{D}$ denote the number of particles contained in $D$. Then
\begin{equation}
    N_{D} = \Theta_{\beta,N}(\mathbf{1}_{D}) = \sum_{m=1}^M \mathbf{1}_{D}(X_m^{(M)}),
\end{equation}
where $\mathbf{1}_{D}$ denotes the indicator function of $D$. The point process $\Theta_{\beta,N}$ is defined in \eqref{eq:GibbsPointProcess} and the random variables $X_j^{(M)}$, $j=1,\ldots,M$ and $M$ are defined in and above \eqref{eq:definitionX}, respectively.

\begin{thm}[\textbf{Local counting statistics}] 
\label{thm:localcounting}
    Let the interaction potential $v$ satisfy the assumption of Theorem~\ref{thm:semiclassicalApproximation}. We consider the limit $N \to \infty$, $\beta/ \beta_{\mathrm{c}}\to \kappa \in (1,\infty)$, see \eqref{eq:beta_asymp}. 
    \begin{enumerate}[label=(\alph*)]
        \item \textbf{Convergence of linear statistics on macroscopic length scale.} Let $f \in H^{\frac{1}{5}}(\Lambda)$ be real-valued. 
        As $N \to \infty$, the normalized \textbf{empirical measure} $N^{-1} \Theta_{\beta,N}$ satisfies
        \begin{equation}
         \mathbf E \left[ \left| \frac{1}{N} \Theta_{\beta,N}(f) - \int_\Lambda f(x) \d x \right|^2 \right] \lesssim_f N^{- \frac13}.
            \label{eq:convergenceOfEmprircalMeasure}
        \end{equation}
        If $f \in L^2(\Lambda)$, the statement is still true with the right-hand side replaced by $o(1)$. That is, $\frac{1}{N} \Theta_{\beta,N}$ converges to the uniform probability measure on $\Lambda$. 
        \item \textbf{Particle number fluctuations in macroscopic sets.}
        Let $X$ be a standard Gaussian random variable, independent of 
        $\Psi$,
        and define 
        \begin{equation}
            V_j = \frac{|D_j|}{\sqrt{\widehat{v}(0)}} X + 2 \sqrt{g(\kappa)} \mathrm{Re} \langle \Psi, \mathbf{1}_{D_j} \rangle_{L^2(\Lambda)}, \qquad j=1,...,n.
            \label{eq:defVj}
        \end{equation}
        Here $|D_j|$ and $\mathbf{1}_{D_j}$ are the Lebesgue measure and the indicator function of $D_j \subset \Lambda$, respectively. As $N \to \infty$, we have
        \begin{equation}
            \Big(  \sqrt{\beta/N} (N_{D_i} - N |D_i|)\Big)_{i=1}^n  \xrightarrow{d} \left( V_1, ..., V_n \right),
            \label{eq:macroscaleStats}
        \end{equation}
        where $\xrightarrow{d}$ denotes convergence in distribution. If the sets $D_i$ have smooth boundary, convergence in distribution can be replaced with convergence of the laws of the random variables in Wasserstein distance $W_2$. 
        
        Moreover, $(V_1,...,V_n)$ is a centered real Gaussian random vector with the covariance matrix 
        \begin{equation}
            \mathbf{E}[V_i V_j] = 2 g(\kappa) \sum_{p \in \Lambda^*} \frac{\overline{\widehat{\mathbf{1}}_{D_i}(p)} \widehat{\mathbf{1}}_{D_j}(p)}{p^2 + 2 g(\kappa) \widehat{v}(p)}. 
        \end{equation}
        Here $\widehat{\mathbf{1}}_D$ denotes the Fourier transform of $\mathbf{1}_D$.
        \item \textbf{Particle number fluctuations in microscopic sets.} 
        On the microscopic length scale $N^{-1/3}$, the point process $\Theta_{\beta,N}$ defined in \eqref{eq:GibbsPointProcess} converges weakly to the \textbf{boson point process} $\Xi$ defined in Section~\ref{sec:bosonPointProcess}. More precisely, for $f \in \mathcal{C}_{\mathrm{c}}(\mathbb{R}^3)$, where $\mathcal{C}_{\mathrm{c}}$ denotes continuous functions with compact support, let $f_N(x) = f(N^{\frac13}x)$. For all sufficiently large $N$, $f_N$ may be regarded as a function on $\Lambda$, and 
        \begin{equation}
            \lim_{N \to \infty} \mathbf E[e^{i \Theta_{\beta,N}(f_N)}]= \mathbf E[e^{i \Xi (f)}]. 
        \end{equation}
    \end{enumerate}
\end{thm}

An introduction to convergence of point processes can be found in \cite[Chaper~11]{DaleyVere-Jones2008}. The above theorem provides information about the distribution of particles in $\Lambda$. The empirical measure converges to the uniform measure, as expected from the translation invariance of the Gibbs state. On macroscopic sets, particle-number fluctuations are governed by two independent contributions: one arising from fluctuations of the condensate and the other depending on the condensate fraction and the Bogoliubov field. In contrast, for sets on the microscopic length scale $N^{-1/3} \sim \beta^{1/2}$ (the thermal wavelength), the~particle-number fluctuations are independent of the interaction potential and are described by the boson point process. The interaction becomes negligible because the spatial scale $N^{-1/3}$ corresponds to momentum scales $|p| \sim N^{1/3}$, for which $\hat{v}(p)$ is small. This, in particular, proves universality of the local particle number statistics. 

We emphasize that this result holds under the assumption $\hat{v}(0) > 0$. In the case of the ideal gas ($v=0$), the large condensate fluctuations of order $N$ lead to a different result. Repulsive interactions suppress these large condensate fluctuations, reducing their order to $N^{5/6}$. These fluctuations are accurately described by $g^{\mathrm{BEC}}$ in \eqref{eq:gBEC}. The derivations of the boson point process for the ideal gas in \cite{TamuraIto2006,TamuraIto2007,TamuraZagrebnov2012} were carried out in the canonical ensemble, where condensate fluctuations are of the order $N^{2/3}$ and this problem does not arise.

Regarding the above theorems, we have the following remarks. 

\begin{remark}
\label{rem:discussionMainResults}
    \begin{enumerate}[label=(\alph*)]
        \item For simplicity, we state Theorems~\ref{thm:semiclassicalApproximation}--\ref{thm:localcounting} in the condensed phase, i.e., for $\kappa>1$. Results for the case $\kappa\in(0,1]$ are discussed in Section~\ref{sec:mainNonCondensed} below.
        \item Theorem~\ref{thm:semiclassicalApproximation} provides an approximation of the Gibbs state $G_{\beta,N}$ in \eqref{eq:GibbsState} that differs from the one in \cite[Theorem~1]{DeuNamNap-25}, on which its proof relies. The formulation in Theorem~\ref{thm:semiclassicalApproximation} is needed for the proofs of Theorems~\ref{thm:thermaldistr} and~\ref{thm:localcounting}. Moreover, it offers a different perspective, as it is semiclassical in nature and establishes a connection to the Bogoliubov field $\Psi$ in \eqref{eq:BogoliubovFreeField}. This perspective also allows us to prove that the point process $\Theta_{\beta,N}$ in \eqref{eq:GibbsPointProcess} associated with the Gibbs state can be approximated by a Cox process. 
        \item To the best of our knowledge, the Bogoliubov field $\Psi$ in \eqref{eq:BogoliubovFreeField} has not previously appeared in either the mathematics or the physics literature on the Bose gas. The closest related work we found is \cite{SinatraCastin2009}, where a version of the Gaussian field $\Phi_N$ from Theorem~\ref{thm:semiclassicalApproximation} is used as part of a stochastic numerical method that implements Bogoliubov theory for Bose gases at a positive temperature.
        \item Part~(c) of Theorem~\ref{thm:localcounting} establishes boson bunching for a bosonic many-particle system with repulsive interactions, as this property is inherited from the boson point process. Boson bunching refers to the enhanced probability of finding one boson close to another compared with the probability in the absence of the other particle. In our system, this behavior survives despite the repulsive interactions between the particles because of a separation of length scales: the range of the interaction is much larger than both the thermal wavelength and the typical interparticle distance, which are of order $N^{-1/3}$. This phenomenon may be contrasted with the Pauli exclusion principle for fermions. 
        
        In the physics literature, boson bunching is also known as the Hanbury--Brown--Twiss effect, see \cite{HanburyBrownTwiss1956,YasShi1996,Schelletal2005,Jeltesetal2007}. The question of whether this effect can be proved to occur in an interacting bosonic many-particle system was raised in \cite[Section~6.2.3]{BardanetEtal2022}. 
        \item Some of the statements in Theorems~\ref{thm:thermaldistr} and \ref{thm:localcounting} establish convergence in total variation norm and Wasserstein distances. Their proofs rely on bounds for the second moments of quantum observables given by second quantizations of one-particle operators in the Gibbs state, obtained in \cite{DeuNamNap-25}. Simpler proofs are possible if these stronger notions of convergence are replaced by convergence in distribution.
        \item Theorems~\ref{thm:thermaldistr} and~\ref{thm:localcounting} are partly motivated by recent experiments on cold alkali gases \cite{SakKase2016,TenartEtAl2021,HerceEtAl2023,LamiraultEtAl2025}. In \cite{SakKase2016}, the authors demonstrate that current experiments are sensitive to the discrepancy between the empirical measure and the average density, thereby revealing finite-size effects. In \cite{TenartEtAl2021} and \cite{LamiraultEtAl2025}, correlations between occupation numbers of different momentum modes are measured. Finally, \cite{HerceEtAl2023} reports measurements of the distribution of particle numbers in spatial subsets of the configuration space for a gas of lattice bosons.
        \item Using Theorems~4 and~5 in \cite{DeuNamNap-25}, one readily obtains 
        \begin{align}
            \lim_{N \to \infty} \beta \Tr[ a_p^* \a_q G_{\beta,N} ]
            &= \mathbf{E}[ \overline{\widehat \Psi(p)} \widehat \Psi (q) ], \nonumber \\
            \lim_{N \to \infty} \beta^2 \Tr[ a_{p_1}^* a_{p_2}^* \a_{q_1} \a_{q_2} G_{\beta,N} ]
            &= \mathbf{E}[ \overline{\widehat \Psi (p_1)} \overline{\widehat \Psi (p_2)} \widehat \Psi (q_1) \widehat \Psi (q_2) ]
        \end{align}
        for non-zero momenta $p,q,p_1,p_2,q_1,q_2 \in \Lambda_+^*$, where $\widehat \Psi(p)$ denotes the Fourier coefficients of the Bogoliubov field defined in \eqref{eq:familyOfGaussianRandomVariables}. We expect analogous identities to hold for higher-order correlation functions. However, the analysis in \cite{DeuNamNap-25} provides control only over the quantities above and does not extend to higher orders. 
        \item In Theorem~\ref{thm:localcounting}, we discuss the particle-number statistics on the macroscopic length scale $1$ and the microscopic length scale $N^{-1/3}$. On mesoscopic scales $N^{-1/3} \ll \delta_N \ll 1$, one can obtain a result comparable to part~(b) of Theorem~\ref{thm:localcounting}. In this regime, the Bogoliubov field reduces to the Gaussian free field, and the limiting mesoscopic statistics are also universal. 
        \item The scaling factor $1/N$ in front of the interaction in \eqref{eq:Hamiltonian} is chosen so that Bogoliubov theory is visible in the condensed phase. This should be contrasted with the recent works \cite{LewNamRou-21,FroKnoSchSoh-22}, where the authors consider a mean-field Bose gas with an interaction strength of order $N^{-2/3}$ at the critical point ($\kappa =1$). Their scaling is designed to obtain the convergence of the partition and correlation functions, after suitably subtracting a contribution from Hartree theory, to those of a classical $\Phi^4$ theory with a nonlocal interaction. 
        
        With this interaction scaling, one formally obtains the Bogoliubov dispersion relation
        \begin{equation}
            \epsilon(p) = \sqrt{p^4 + 2 g(\kappa) p^2 N^{1/3} \widehat{v}(p)}
        \label{eq:dispersionRelationInOtherScaling}
        \end{equation}
        in the condensed phase. For momenta $|p|$ of order one as $N\to\infty$, it is dominated by the interaction term, while its behavior at large $|p|$ depends sensitively on the regularity of $v$. Therefore, our scaling, for which the $N^{1/3}$ factor in \eqref{eq:dispersionRelationInOtherScaling} is absent, is more natural when considering the full temperature range $\kappa\in(0,\infty)$ rather than focusing only on the critical point $\kappa = 1$.
        \item The random variable $R$ in \eqref{eq:W_def} describes the limiting fluctuations of the number of particles outside the BEC and depends explicitly on the interaction between the particles. In the special cases $v=0$ (non-interacting system) or $g(\kappa)=0$ (non-condensed phase), it reduces, up to a~multiplicative constant, to a random variable in \cite[Eq.~(3)]{ChatDiac2014}. In this work a very general class of one-particle Hamiltonians $h$ is considered, including 
        the Laplacian on the three-dimensional 
        torus (corresponding to 
        $\alpha = 3/2$ in \cite{{ChatDiac2014}}). The main result of \cite{ChatDiac2014} is that $R$ describes the limiting distribution of the number of particles in the BEC for a system of non-interacting bosons. The fact that the same random variable describes the fluctuations of the thermal cloud in our case and the fluctuations of the condensate in \cite{ChatDiac2014} is related to the use of the canonical ensemble in \cite{ChatDiac2014}. Since in this ensemble the total particle number is fixed, the particle numbers in the condensate and in the thermal cloud exhibit the same fluctuations with opposite signs. We conclude that $R$ in \eqref{eq:W_def} can be viewed as a generalization of \cite[Eq.~(3)]{ChatDiac2014} to systems of interacting bosons. 
        
        As Theorem~\ref{thm:thermaldistr} shows, $R$ is related to the Bogoliubov field $\Psi$ in \eqref{eq:BogoliubovFreeField}. Similarly, the random variable in \cite[Eq.~(3)]{ChatDiac2014} is related to the Gaussian free field $F$, defined by the series
        \begin{equation}
            F(x) = \sum_{j=1}^{\infty} \frac{F_j}{E_j - E_0} f_j(x),
        \end{equation}
        where $\{E_j\}_{j=0}^{\infty}$ and $\{ f_j \}_{j=0}^{\infty}$ denote the eigenvalues and eigenfunctions of the one particle Hamiltonian $h$, respectively, and $\{ F_j \}_{j=0}^{\infty}$ is a family of independent, centered, complex Gaussian random variables with covariance $\mathds{1}$ and zero pseudo-covariance. 

        Large-deviation estimates for the number of condensed particles in the canonical ideal gas have recently been established in \cite{DeuLi-26}.
    \end{enumerate}
\end{remark}
\subsection{Main results part II: the non-condensed phase}
\label{sec:mainNonCondensed}
In this section, we consider particle-number statistics in subsets of the torus for $\kappa \in (0,1]$. We first state a theorem for $\kappa \in (0,1)$, corresponding to the non-condensed phase, and then briefly discuss the critical case $\kappa = 1$.
 
\begin{thm} \label{thm:high_temperature}
Let the interaction potential $v$ satisfy the assumption of Theorem~\ref{thm:semiclassicalApproximation}. We consider the limit $N \to \infty$, $\beta/ \beta_{\mathrm{c}}\to \kappa \in (0,1)$, see \eqref{eq:beta_asymp}. 
\begin{enumerate}[label=(\alph*)]
    \item \textbf{Particle number fluctuations in macroscopic sets.} Let
    \begin{align}
        z_\infty = \lim_{N \to \infty} e^{\beta \mu^\id(\beta,N)}, \qquad \sigma_{\beta} = (4 \pi \beta)^{-\frac32} \mathrm{Li}_{\frac12}(z_\infty),
    \end{align}
    where $\mathrm{Li}_{s}(z)$ denotes the polylogarithm. Let $D_1,\ldots,D_n \subseteq \Lambda$ be measurable sets whose boundaries have Lebesgue measure zero and whose closures are disjoint. Define
    \begin{equation}
        V_j = \langle W , \mathbf{1}_{D_j} \rangle_{L^2(\Lambda)}, \qquad  j=1,\ldots,n,
    \end{equation}
    where $W$ denotes real-valued Gaussian white noise on $\Lambda$. Then
    \begin{equation} \left( \frac{N_{D_i} - N |D_i|}{\sqrt{\sigma_{\beta}}} \right)_{i=1}^n \xrightarrow{d} \left( V_1, ..., V_n \right) 
    \end{equation}
    as $N \to \infty$. In particular, $V_1,\dots,V_n$ are independent centered Gaussian random variables with $\mathbf{E}[V_j^2]= |D_j|$.
    \item \textbf{Particle number fluctuations in microscopic sets.} Let $\mu^{\mathrm{id}}(\beta,N) < 0$ be the unique solution to the equation
    \begin{equation}
        N = \sum_{p \in \Lambda^*} \frac{1}{e^{\beta(p^2-\mu^{\mathrm{id}}(\beta,N))}-1},
        \label{eq:id_number_of_part}
    \end{equation}
    and define the effective mass by
    \begin{equation}
        m_{\mathrm{eff}}(\kappa) = \lim_{N \to \infty} \sqrt{-N^{-2/3} \mu^{\mathrm{id}}(\beta,N)}.
    \end{equation}
    We also define the complex Gaussian random field $H_m$ on $\mathbb R^3$ with covariance
    \begin{equation}
        \mathbf E[\overline {H_m(x)} H_m(y) ] = \int_{\mathbb R^3} \frac{e^{ip \cdot (x-y)}}{e^{ \beta_{\infty}(\kappa) ( p^2 + m_{\mathrm{eff}}^2(\kappa) ) } -1} \frac{\d p}{(2 \pi)^3}, 
    \end{equation}
    where $\beta_{\infty}(\kappa) = \frac{\kappa}{4 \pi} \zeta^{2/3}(3/2)$, and vanishing pseudo-covariance. On the microscopic length scale $N^{-1/3}$, the point process $\Theta_{\beta,N}$ defined in \eqref{eq:GibbsPointProcess} converges, as $N \to \infty$, weakly to the boson point process associated with $H_m$ and $g(\kappa) = 0$. More precisely, the limiting process is the Cox process with random intensity measure
    \begin{equation}
        I(D) = \int_{D} \left| H_m(x) \right|^2 \mathrm{d}x.
    \end{equation}
\end{enumerate}
\end{thm}

As stated in part~(a) of Theorem~\ref{thm:high_temperature}, the fluctuations of $N_D$ are Gaussian for $\kappa \in (0,1)$. Heuristically, this can be understood as a central-limit-type result: we may think of the set $D$ as being partitioned into an increasing number of subsets $D_j$, $j=1,\ldots,m$, with $m \gg 1$, such that the random variables $N_{D_j}$, $j=1,\ldots,m$, are essentially uncorrelated. Concerning part~(b), we remark that $m_{\mathrm{eff}}=0$ whenever $\kappa \geq 1$. This explains the absence of a mass term in part~(c) of Theorem~\ref{thm:localcounting}.

Let us now briefly discuss the case $\kappa = 1$ (critical point). In this case, the result is very sensitive to the relation between two different length scales, which we introduce next. We recall the definition of the chemical potential $\mu^{\mathrm{id}}(\beta,N)$ of the ideal gas in \eqref{eq:id_number_of_part}. It is related to the expected number of particles $N_0^{\mathrm{id}}(\beta,N)$ in the $p=0$ mode through
\begin{equation}
\mu^{\mathrm{id}}(\beta,N)
= - \frac{1}{\beta} \ln\left(1+\frac{1}{N_0^{\mathrm{id}}(\beta,N)}\right).
\end{equation}
For $\kappa>1$ (the condensed phase), we have $-\mu^{\mathrm{id}}(\beta,N) \sim N^{-1/3}$, whereas for $\kappa < 1$, $-\mu^{\mathrm{id}}(\beta,N) \sim N^{2/3}$. At the critical point $\kappa=1$, the intermediate regime
\begin{equation}
N^{-1/3} \ll -\mu^{\mathrm{id}}(\beta,N) \ll N^{2/3}
\end{equation}
occurs. The precise behavior is determined by how $\beta/\beta_{\mathrm{c}}$ approaches $\kappa$. Associated with $\mu^{\mathrm{id}}(\beta,N)$ is the length scale
\begin{equation}
\ell_{\mathrm{cor}}(\beta,N)
= \frac{1}{\sqrt{-\mu^{\mathrm{id}}(\beta,N)}}.
\end{equation}
We emphasize that $\ell_{\mathrm{cor}}(\beta,N) \gg \beta^{1/2} \sim N^{-1/3}$ for $\kappa = 1$.

The second length scale is associated with the set under consideration. More precisely, for a measurable set $D \subset \Lambda$ and a sequence $\delta_N$ satisfying $N^{-1/3} \leq \delta_N \leq 1$, we define $D_N = \delta_N D$.

We have the following two general remarks.
\begin{enumerate}[label=(\alph*)]
\item If $\ell_{\mathrm{cor}}(\beta,N) \sim \delta_N \gg N^{-1/3}$, the system admits an effective-field description with effective mass
\begin{equation}
m_{\mathrm{eff}}
= \lim_{N \to \infty}
\frac{\delta_N}{\ell_{\mathrm{cor}}(\beta,N)}.
\label{eq:meff}
\end{equation}
If the system is probed on length scales $\delta_N \gg \ell_{\mathrm{cor}}(\beta,N)$, the effective mass becomes infinite and the effective field reduces to Gaussian white noise. Consequently, the fluctuations of $N_{D_N}$ are Gaussian in the limit $N \to \infty$ and correlations between disjoint sets vanish. In particular, the same heuristic as in part~(a) of Theorem~\ref{thm:high_temperature} applies. On the other hand, if $\delta_N \ll \ell_{\mathrm{cor}}(\beta,N)$, the limiting field is massless.
\item Assume that $\delta_N \sim N^{-1/3}$. Since $\ell_{\mathrm{cor}}(\beta,N) \gg N^{-1/3}$ for $\kappa = 1$, the limiting point process is, by the same heuristic as in part~(b), always a boson point process whose intensity measure is defined in terms of a massless random field. Since the condensate fraction satisfies $g(1) = 0$, no separate contribution from the $p=0$ mode appears.
\end{enumerate}
Using our techniques, one can treat all of the above parameter regimes and compute the corresponding limiting distributions. For simplicity, we restrict our attention to the cases covered by the theorems above, which we believe to be the most interesting ones.
\subsection{Organization of the article} The article is organized as follows. In Section \ref{sec:effective} we discuss the effective theories that will play a central role throughout the work. Section \ref{sec:ideal} is devoted to the discussion of the  ideal Bose gas. In particular, Proposition \ref{prop:limits} provides a precise description of the relations between the temperature, chemical potential and condensate density in the non-interacting system. Section \ref{sec:effectiveBogo}  discusses Bogoliubov theory at positive temperatures. There, we introduce a~slightly modified approximate Gibbs state, compared with the one considered in \cite{DeuNamNap-25}. Several estimates needed to establish the rate at which this modified state approximates the Gibbs state of the interacting system are collected in Section \ref{app:cutoff}. In Section \ref{sec:gwasserstein}, we study properties of the condensate-mode random variable $Z$ introduced in Theorem \ref{thm:semiclassicalApproximation}. 

Section \ref{sec:semiclassicalApproximation} is devoted to the proof of Theorem \ref{thm:semiclassicalApproximation}. We first determine the semiclassical symbols of the approximate Gibbs state. This leads to the construction of the random field $\Phi_N (x)$ and the proof of part (a) of Theorem \ref{thm:semiclassicalApproximation}. We then construct the Bogoliubov field $\Psi (x)$ and establish its regularity properties which are used in the proof of Theorem \ref{thm:semiclassicalApproximation} (b). We conclude the section with a proof of Theorem \ref{thm:semiclassicalApproximation} (c) which relies on the previous statements and some properties of the point processes associated to the Gibbs state discussed in Appendix \ref{app:GibbsPointProcess}.

Section \ref{sec:thermaldistr} is devoted to the proof of Theorem \ref{thm:thermaldistr} and is rather technical. In order to get the quantitative estimates in total variation and Wasserstein metrics, we use both the field theoretic representations derived in the preceding section and the correlation inequalities established in \cite{DeuNamNap-25} for the full many-body problem. The same applies to the proof of Theorem \ref{thm:localcounting} given in Section \ref{sec:prooflocalcounting}. The proof of Theorem~\ref{thm:high_temperature}, presented in Section \ref{sec:noncond}, follows a somewhat different route. There, convergence in distribution is established by analyzing the asymptotic behavior of regularized Fredholm determinants. Basic properties of these determinants are presented in  Appendix \ref{app:fredholm}. On the way we use some technical lemmas presented in Appendix \ref{app:lemmas}. 

\section{Effective theories}\label{sec:effective}
\subsection{Ideal gas} 
\label{sec:ideal}
In this section we consider non-interacting bosons in the grand canonical ensemble, with inverse temperature $\beta >0$ and chemical potential $\mu^\id<0$. The expected number of particles is given by \eqref{eq:id_number_of_part}. This equation can be solved for $\mu^\id$ as a function of $\beta$ and~$N$. To describe this relation, it is convenient to use the polylogarithm function:
\begin{equation}
\mathrm{Li}_s(z) = \sum_{m=1}^\infty m^{-s} z^m.
\label{eq:polylog_def}
\end{equation}
We will also consider the condensate fraction, defined as the ratio of the expected number of particles with momentum zero and the total expected number of particles:
\begin{equation}
    g^{\id} = \frac{1}{N} \frac{1}{e^{- \beta \mu^\id}-1}.
    \label{eq:gid_def}
\end{equation}

\begin{prop} \label{prop:limits}
Consider the limit $N \to \infty$ with $N$-dependent $\beta$ as in \eqref{eq:beta_asymp}, and $\mu^\id$ determined by the relation \eqref{eq:id_number_of_part}. Let $g^\id$ be defined as in \eqref{eq:gid_def}.
\begin{enumerate}
    \item If $\kappa >1$, then $ g^\id \to 1 - \kappa^{-\frac32} $ and $\lim_{N \to \infty} N^{\frac13} \mu^\id$ exists in $(-\infty,0)$.    
    \item If $\kappa < 1$, then $e^{\beta \mu^\id}  \to \mathrm{Li}_{\frac{3}{2}}^{-1}(\kappa^{\frac32} \zeta(\tfrac32))$, where $\mathrm{Li}_{\frac32}^{-1} : [0, \zeta(\tfrac32) ] \to [0,1]$ is the the inverse of $\mathrm{Li}_{\frac32} $ defined in \eqref{eq:polylog_def}. The limits $\lim_{N \to \infty} N g^\id$ and $-\lim_{N \to \infty} N^{-\frac23} \mu^\id$ exist in $(0,\infty)$. 
    
    \item If $\kappa =1$, then
    \begin{equation}
        \lim_{N\to \infty} N g^\id = \infty, \qquad  \lim_{N\to \infty}g^\id = 0, \qquad \lim_{N\to \infty} N^{-\frac23} \mu^\id = 0, \qquad  \lim_{N\to \infty} N^{\frac13} \mu^\id =-\infty. 
    \end{equation}
    If $\mu_* : (0,\infty) \to (-\infty,0)$ is a function such that $ N^{-\frac23} \mu_*(N) \to 0$ and $N^{\frac13} \mu_*(N) \to -\infty$, then for some $\beta(N)$ satisfying \eqref{eq:beta_asymp} with $\kappa=1$ we have $\mu^\id = \mu_*(N)$ for sufficiently large $N$.
    \item There exists $c>0$ such that for all $\kappa \geq 1$ we have
\begin{align}
    1 - g^\id - \left( \tfrac{\beta_c}{\beta} \right)^{\frac32} = O(N^{-\frac13} e^{- c \sqrt{- \mu^{\id}}}) +  O((N g^\id)^{-\frac12} ),  \qquad 
    Ng^\id =(-\beta \mu^\id)^{-1} (1+o(1)). \label{eq:particle_partition}
\end{align}
In particular, $\mu^{\mathrm{id}}$ has a limit in $(-\infty,0)$ if and only if $N^{\frac{1}{3}} g^{\id}$ has a limit in $(0,\infty)$. 
    \item For every $\kappa \in (0, \infty)$ we have $\liminf_{N \to \infty} e^{\beta \mu^\id} >0$. Equivalently, $\beta \mu^\id$ is bounded. 
\end{enumerate}
\end{prop}
\begin{proof}
A simple manipulation of \eqref{eq:n_function} in Appendix \ref{app:lemmas} gives
\begin{equation}
    g^\id= 1 - \left( \frac{\beta_c}{\beta} \right)^{\frac32} \frac{\mathrm{Li}_{\frac32}(e^{\beta \mu^\id})}{\zeta(\tfrac32)} + O(N^{-\frac13} e^{- c \sqrt{- \mu^{\id}}}) =  1 - \kappa^{-\frac32} \frac{\mathrm{Li}_{\frac32}(e^{\beta \mu^\id})}{\zeta(\tfrac32)} + o(1).
    \label{eq:gbetamu}
\end{equation}

(1) Let $\kappa>1$. Since $\mathrm{Li}_{\frac32}(z) \leq \zeta(\tfrac32)$ for $z \in [0,1]$, we obtain $    \liminf_{N \to \infty} g^\id \geq 1 - \kappa^{-\frac32} > 0$.
By~inspection of \eqref{eq:gid_def}, this implies that $e^{\beta \mu^{\id}} \to 1$. Considering \eqref{eq:gbetamu} again, we find that $g^\id \to 1 - \kappa^{-\frac32}$. The asymptotic behavior of $ \mu^\id$ follows easily.

(2) Suppose that $\kappa<1$. The non-negativity of the right hand side of \eqref{eq:gbetamu} for large $N$ implies that
\begin{equation}
    \limsup_{N \to \infty} \mathrm{Li}_{\frac32}(e^{\beta\mu^\id}) \leq \kappa^{\frac32} \zeta(\tfrac32) < \zeta(\tfrac32).
\end{equation}
By the strict monotonicity of $\mathrm{Li}_{\frac32}$, this implies $\limsup_{N \to \infty} e^{\beta\mu^\id} < 1$, and hence $g^\id \to 0$. Considering \eqref{eq:gbetamu} again we find $  \mathrm{Li}_{\frac32}(e^{\beta\mu^\id}) \to \kappa^{\frac32} \zeta(\tfrac32)$. Now note that
$\mathrm{Li}_{\frac32} : [0,1] \to [0, \zeta(\tfrac32)]$ is a homeomorphism.

(3) Let $\kappa =1$. If along some subsequence $e^{\beta \mu^\id} \to z \in [0,1)$, then \eqref{eq:gbetamu} gives
\begin{equation}
    g^\id \to 1 - \frac{\mathrm{Li}_{\frac{3}{2}}(z)}{\zeta(\tfrac32)} >0.
\end{equation}
This implies that $e^{\beta \mu^\id} \to 1$, in contradiction with the choice of the subsequence. Hence $e^{\beta \mu^\id} \to 1$, which implies that $N g^\id \to \infty$ and (by \eqref{eq:gbetamu}) that $g^\id \to 0$. 

Now fix the function $\mu_*(N)$ and choose $c >1$. By the preceding, for sufficiently large $N$ we have:
\begin{align}
    \mu^\id & > \mu_*(N) \qquad \text{for } \beta = c \beta_c, \\
    \mu^\id & < \mu_*(N) \qquad \text{for } \beta = c^{-1} \beta_c. \nonumber
\end{align}
By continuity, there exists an intermediate value $\beta \in [c^{-1} \beta_c, c \beta_c]$ for which $\mu^\id = \mu_*(N)$. We claim that for $\beta$ chosen this way we have $\frac{\beta}{\beta_c} \to 1$. Since $\frac{\beta}{\beta_c}$ is valued in the compact interval $[c^{-1},c]$, it is enough to show that $\frac{\beta}{\beta_c}$ does not have subsequences which converge to a limit different than $1$. By the previous points, along every such subsequence we would have either $N^{-\frac23} \mu_*(N) \to c'$ for some $c' \neq 0$ or $N^{\frac13} \mu_*(N) \to c'$ for some $c' \neq 0$. This is impossible by the choice of $\mu_*(N)$. 

(4) We have seen that for all $\kappa \geq 1$ we have $e^{\beta \mu^\id} \to 1$. Hence, by Lemma \ref{lem:Li_Holder} in Appendix \ref{app:lemmas},
\begin{equation}
    \mathrm{Li}_{\frac32} (e^{\beta \mu^\id}) = \zeta(\tfrac32) - O (\sqrt{1-e^{\beta \mu^\id}}) = \zeta(\tfrac32) - O((Ng^\id)^{-\frac12}). 
\end{equation}
Combining this with \eqref{eq:gbetamu} we find \eqref{eq:particle_partition}.

(5) An immediate consequence of the previous points. 
\end{proof}

As we will see later, the quantitative behavior of many observables in the limit $N \to \infty$ depends on
\begin{equation}
  m^2 = - \lim_{N \to \infty} \mu^{\id}.
\label{eq:mu_infty_def}
\end{equation}
We have $m^2 = 0$ if $\kappa>1$ and $m^2=\infty$ if $\kappa <1$. Any value of $m^2$ in $[0,\infty]$ is possible in the critical case $\kappa=1$. We will always choose $\beta(N)$ in such a way that the limit \eqref{eq:mu_infty_def} exists in $[0,\infty]$. 

The quantity $m^2$ defined for the ideal gas turns out to be a useful parameter in the analysis of the interacting system with the same $\beta, N$. We emphasize that the corresponding chemical potentials $\mu$ and $\mu^\id$ are different, and we are not claiming that $-\mu \to m^2$ as $N \to \infty$. 

Note that, by Proposition \ref{prop:limits}.(4), the convergence $\mu^{\id}\to -\infty$ implies that $N g^{\id} \ll N^{2/3}$. As proven in \cite[Theorem~1]{DeuNamNap-25}, in this parameter regime the Gibbs state of the interacting system is well approximated by the one of the ideal gas. More precisely, we have 
\begin{equation}\label{eq:norm_approx_ideal}
 \|G_{\beta,N}-G^{\id}\|_1 \lesssim N^{-1/6}, \quad \text{ where } \quad G^{\id}=\frac{e^{- \beta (\mathcal H^{\id}_N - \mu^{\id} \mathcal N)}}{\mathrm{Tr}[ e^{- \beta (\mathcal H^{\id}_N - \mu^{\id} \mathcal N) } ]}
\end{equation}
is the Gibbs state of the ideal gas (here $\mathcal{H}^{\id}_N$ denotes the Hamiltonian of the ideal gas, i.e.\ the operator $\mathcal{H}_N$ as in \eqref{eq:Hamiltonian}, with $v\equiv 0$). Let us note that the exponent $-1/6$ in \eqref{eq:norm_approx_ideal} does not appear directly in the original formulation in \cite[Theorem~1]{DeuNamNap-25}, as the bound given there is stated uniformly across all parameter regimes; cf.\ \eqref{eq:norm-approximation-intro}.

\subsection{Bogoliubov approximation} \label{sec:effectiveBogo}
As mentioned in the introduction, the fraction of particles occupying the zero-momentum mode depends on the temperature. To obtain a good approximation to the Gibbs state, we account for this effect by using a temperature-dependent Bogoliubov Hamiltonian and an associated reference state. We will now introduce this construction.

By the exponential property of Fock spaces, we have
$$\mathscr F(L^2(\Lambda))\cong \mathscr F_0 \otimes \mathscr F_+$$
where $\mathscr F_0 = \mathscr F(\text{span}\{ 1 \})$ is the Fock space of the condensate, and $\mathscr F_+=\mathscr F( L_+^2(\Lambda) ) $ is the Fock space of excitations (cf. \eqref{eq:L2plus}). 
In the space $\mathscr F_0$, for $z\in \mathbb{C}$,  we introduce the coherent states 
\begin{equation}
	|z\rangle= \exp( z a_0^* - \overline{z} \a_0 )\Omega_0,  
	\label{eq:coherentstate}
\end{equation}
 with the vacuum vector $\Omega_0 \in \mathscr F_0$. Heuristically speaking, the coherent state $|z\rangle$ describes a BEC in the constant function $z/|z| \in L^2(\Lambda)$, with an expected number of $|z|^2$ particles. It is a normalized eigenvector of the annihilation operator $a_0$ to the eigenvalue $z$, i.e. 
\begin{equation}
    \a_0 |z\rangle = z |z\rangle.
    \label{eq:coherentstateEigenfunction}
\end{equation}
In particular, $\langle z| a_0^* \a_0 z \rangle=|z|^2$ and $\langle z| \a_0 \a_0 z \rangle=z^2$. This allows to implement the Bogoliubov c-number substitution by modeling the zero momentum particles with a coherent state. Doing this and truncating  higher order terms in the resulting Hamiltonian leads to the description of thermally excited particles by the Bogoliubov Hamiltonian depending on the temperature and $z$:
\begin{equation}
  \sum_{p \in \Lambda^*_+} (p^2-\mu^\id) a_p^* \a_p + \frac{|z|^2}{2N} \sum_{p \in \Lambda_+^*} \hat{v}(p) \left( 2 a_p^* \a_p  + \frac{z^2}{|z|^2} a_p^* a_{-p}^* + \frac{\overline{z}^2}{|z|^2} \a_p \a_{-p} \right) .
    \label{BogoliubovHamiltonianAndi}
\end{equation}
At zero temperature one would then set $|z|^2\approx N$. This is not the case at positive temperatures. Here, the condensate is described by a one-mode $\Phi^4$-theory with the Gibbs distribution
\begin{equation}\label{eq:GibbsDistributionDiscrete}
	g^{\mathrm{BEC}}(z) = \frac{\exp\left( -\beta\left( \frac{\hat{v}(0)}{2N} |z|^4 - \mu^{\mathrm{BEC}} |z|^2 \right) \right)}{\int_{\mathbb{C}} \exp\left( -\beta\left( \frac{\hat{v}(0)}{2N} |w|^4 - \mu^{\mathrm{BEC}} |w|^2 \right) \right) \d w} ,
\end{equation}
where  $\d w = \d x \d y/\pi$ with $w=x+iy$, and  $\mu^{\mathrm{BEC}} \in \mathbb{R}$ is an appropriately chosen chemical potential of the condensate. The function $g^{\mathrm{BEC}}$ appeared in a similar context in \cite{BocDeuSto-24} (see also \cite{CapDeu-23}), and is related to the classical field theory in \cite{LewNamRou-21,FroKnoSchSoh-22}. The induced probability distribution of $|z|^2$ is strongly peaked around its mean, which can be approximated to leading order by $N g^\id$.  Therefore, we replace $\frac{|z|^2}{N}$ in \eqref{BogoliubovHamiltonianAndi} with $g^{\id}$. This yields the Hamiltonian:
\begin{equation}
 \mathcal H^{\mathrm{Bog}}(z)= \sum_{p \in \Lambda^*_+} (p^2-\mu^\id) a_p^* \a_p + \frac{g^{\mathrm{id}}}{2} \sum_{p \in \Lambda_+^*} \hat{v}(p) \left( 2 a_p^* \a_p  + \frac{z^2}{|z|^2} a_p^* a_{-p}^* + \frac{\overline{z}^2}{|z|^2} \a_p \a_{-p} \right),
    \label{BogoliubovHamiltonianNamed}
\end{equation}
which depends on $\frac{z}{|z|}$, but not on $|z|$. Hamiltonians with different $z$ are unitarily equivalent:
\begin{equation} \label{eq:z=1notation}
     \mathcal H^{\mathrm{Bog}}(z) = \left( \frac{z}{|z|} \right)^{\mathcal N_+} \mathcal H^{\mathrm{Bog}} \left( \frac{z}{|z|} \right)^{-\mathcal N_+}, \qquad \text{where} \qquad \mathcal H^{\mathrm{Bog}} \coloneq \mathcal H^{\mathrm{Bog}}(1).
\end{equation}

The Gibbs state of the Bogoliubov Hamiltonian in \eqref{BogoliubovHamiltonianNamed} will be denoted by 
\begin{equation}
	G^{\mathrm{Bog}}(z) = \frac{\exp\left( -\beta \mathcal{H}^{\mathrm{Bog}}(z)  \right)}{\mathrm{Tr}_{\mathscr F_+} \left[ \exp\left( -\beta \mathcal{H}^{\mathrm{Bog}}(z) \right) \right]}, \qquad G^{\mathrm{Bog}} = G^{\mathrm{Bog}}(1). \label{eq:BogoliubovGibbsState}
\end{equation}
  
The main result of \cite[Theorem~1]{DeuNamNap-25} states that for $Ng^{\id}\geq CN^{2/3}$, the Gibbs state $G_{\beta,N}$ in \eqref{eq:GibbsState} can be approximated in trace norm
\begin{equation}
\| G_{\beta,N} - \Gamma_{\beta,N} \|_1  \lesssim N^{-1/48},
\label{eq:norm-approximation-intro}
\end{equation}
with the reference state $\Gamma_{\beta,N}$ defined as
\begin{equation}
    \Gamma_{\beta,N} = \int_{\mathbb{C}} | z \rangle \langle z | \otimes G^{\mathrm{Bog}}(z) g^{\mathrm{BEC}}(z) \d z .
    \label{eq:referenceState-intro1}
\end{equation}

For technical reasons, we will need a modified reference state, obtained by replacing the interaction potential $\widehat{v}(p)$ with a high-momentum cutoff version $\widehat{u}(p)$:
\begin{equation} \label{def:cutoffpotential}
    \widehat u(p) = \begin{cases}
        \widehat v (p), & |p| \leq p_c, \\
        0, & |p| > p_c.
    \end{cases}
\end{equation}
Here, $p_c$ is an $N$-dependent parameter satisfying $ p_c \to \infty$ as $N \to \infty$ and $p_c = O(N^{\frac13})$; for instance, one can choose $p_c = N^{\frac13}$. 

We introduce the modified Bogoliubov Hamiltonian
\begin{equation}
    \widetilde {\mathcal H}^\Bog(z) = \frac12 \sum_{p \in \Lambda^*} \left[ (p^2 - \mu^\id + g^\id \widehat u(p)) (a_p^* \a_p + a_{-p}^* \a_{-p}) + g^{\mathrm{id}} \widehat{u}(p) \left( \frac{z^2}{|z|^2} a_p^* a_{-p}^* + \frac{\overline{z}^2}{|z|^2} \a_p \a_{-p} \right) \right],
    \label{eq:HBog}
\end{equation}
and the corresponding Gibbs state:
\begin{equation} \label{def:cutoffGibbs}
    \widetilde G^\Bog(z) = \frac{\exp( - \beta \widetilde{\mathcal H}^\Bog(z))}{\Tr_{\mathscr F_+}[\exp( - \beta \widetilde{\mathcal H}^\Bog(z))]}.
\end{equation}

As  is well known, the Bogoliubov Hamiltonian $\widetilde {\mathcal H}^\Bog=\widetilde {\mathcal H}^\Bog(1)$ can be diagonalized. To this end,  for every $p \neq 0$ we choose a~hyperbolic angle $\chi_p = \chi_{-p} \in \mathbb R$ and let 
$$c_p = \cosh(\chi_p), \,\,\,\, s_p = \sinh(\chi_p).$$
We introduce the operators
\begin{equation}
    b_p^* = c_p a_p^* + s_p \a_{-p}, \qquad \b_p = c_p \a_p + s_p a_{-p}^*,
    \label{eq:Bogoliubov}
\end{equation}
which satisfy the canonical commutation relations because $c_p^2-s_p^2=1$. We have the identity
\begin{equation}
    b_p^* \b_p + b_{-p}^* \b_{-p} = 2 s_p^2 + (c_p^2 + s_p^2) \left( a_p^* \a_p + a_{-p}^* \a_{-p} + \frac{2c_p s_p}{c_p^2 + s_p^2} (a_p^* a_{-p}^* + \a_p \a_{-p}) \right).
    \label{eq:bstarb_identity}
\end{equation}
A comparison of \eqref{BogoliubovHamiltonianNamed} with \eqref{eq:bstarb_identity} reveals that it is useful to choose the value of $\chi_p$ determined by 
\begin{equation}
\tanh(2 \chi_p) =  \frac{2c_p s_p}{c_p^2 + s_p^2} = \frac{g^\id \widehat u(p)}{p^2 - \mu^\id + g^\id \widehat u(p)}.   
\end{equation}
Using identities satisfied by hyperbolic function we work out that:
\begin{align}
    c_p^2 + s_p^2 &= \frac{p^2 -\mu^\id+ g^\id \widehat u(p)}{\epsilon(p)}, &
    2c_p s_p &= \frac{g^\id \widehat u(p)}{\epsilon(p)}, \label{eq:relationsCoefficients} \\
    c_p &= \sqrt{\frac{p^2-\mu^\id + g^\id \widehat u(p) + \epsilon(p)}{2 \epsilon(p)}}, &
    s_p &= \sqrt{\frac{p^2-\mu^\id + g^\id \widehat u(p) - \epsilon(p)}{2 \epsilon(p)}}, \nonumber
\end{align}
where
\begin{equation}
    \epsilon(p)  = (p^2 - \mu^\id)^{\frac12} \left( p^2 - \mu^\id + 2g^\id \widehat u(p) \right)^{\frac12}.
\end{equation}
\begin{remark}\label{rem:gammalpahtilde}
The same computations carried out for \eqref{BogoliubovHamiltonianNamed} lead to analogous formulas with $\widehat u(p)$ replaced by $\widehat v(p)$. For reference in calculations comparing states with and without cutoff $p_c$, we introduce
\begin{equation}    
\begin{aligned}        \label{eq:defvp}
        \nu_{p}& =  \frac{1}{2} \left( \frac{p^2-\mu^{\id}}{p^2-\mu^{\id} + 2 g^\id \widehat v(p)} \right)^{-1/4}-\frac{1}{2} \left( \frac{p^2-\mu^{\id}}{p^2-\mu^{\id} + 2 g^\id \widehat v(p)} \right)^{1/4},\\
        \omega(p)&=(p^2 - \mu^\id)^{\frac12} \left( p^2 - \mu^\id + 2g^\id \widehat v(p) \right)^{\frac12}
\end{aligned}
\end{equation}
which is the rewriting of $s_p$ and $\epsilon(p)$  with $\widehat u(p)$ replaced by $\widehat v(p)$.
\end{remark}
There exists \cite{NamNapSol-16} a~unitary $U$ on $\mathscr F_+$ such that 
\begin{equation}
 U^* \a_p U = \b_p   .
 \label{eq:ab_unitary}
\end{equation}

We can now express $\widetilde{\mathcal H}^{\mathrm{Bog}}$ in the form
\begin{equation} \label{eq:BogHamdiag}
    \widetilde{\mathcal H}^{\mathrm{Bog}} = E_0 + \sum_{p \in \Lambda_+^*} \epsilon(p) b_p^* \b_p = U \left( E_0 + \sum_{p \in \Lambda_+^*} \epsilon(p) a_p^* \a_p \right) U^*,
\end{equation}
where 
\begin{equation}
    E_0 = - \frac12 \sum_{p \in \Lambda_+^*} (p^2 - \mu^\id + g^\id \widehat u(p) - \epsilon(p)). 
\end{equation}

The state $\widetilde G^{\mathrm{Bog}}$, defined as in \eqref{eq:BogoliubovGibbsState} but with $\mathcal H^\Bog$ replaced with $\widetilde{ \mathcal H}^\Bog$, satisfies Wick's rule and is therefore fully characterized by its covariance:
\begin{equation}
\label{def:gammabogalphabog}
\begin{aligned}
\widetilde{\gamma}^{\rm{Bog}}(p)&:= \mathrm{Tr}[ \widetilde G^\Bog \, a_p^* \a_p ]= s_p^2+\frac{c_p^2+s_p^2}{e^{\beta \epsilon(p)}-1},\\
\widetilde{\alpha}^{\rm{Bog}}(p)&:=\mathrm{Tr} [\widetilde G^\Bog \a_p \a_{-p} ] = -2 c_p s_p \left( \frac{1}{2} + \frac{1}{e^{\beta \epsilon(p)}-1} \right).
\end{aligned}
\end{equation}
 
 We are ready to define the reference state with a cutoff:
\begin{equation} \label{def:cutoffapproxGibbs}
    \widetilde{\Gamma}_{\beta,N} = \int_{\mathbb{C}} | z \rangle \langle  z | \otimes \widetilde G^{\mathrm{Bog}}(z) \widetilde g^{\mathrm{BEC}}(z) \d z .
\end{equation}
The function $\widetilde g^\BEC(z)$ is the same as $g^\BEC(z)$ in \eqref{eq:GibbsDistributionDiscrete}, with $\mu^\BEC$ replaced with $\widetilde \mu^\BEC$ chosen so that the expectation of $\mathcal N$ in $\widetilde{\Gamma}_{\beta,N}$ equals $N$. Let us note that $\widetilde \mu^\BEC$ depends on $p_c$. Such a state had previously been used in \cite{BocDeuSto-24,CapDeu-23} to construct a trial state for the dilute Bose gas in the Gross--Pitaevskii limit. 

\subsection{Approximate Gibbs state with a cutoff} \label{app:cutoff}
We will now show  that $\| \Gamma_{\beta,N} - \widetilde{\Gamma}_{\beta,N} \|_1 \to 0$ holds provided that the cutoff $p_c $ is large enough. 

\begin{prop}
\label{prop:traceNormDifference}
   Assume $\beta$ satisfies \eqref{eq:beta_asymp} and that
    \begin{equation}
        \int_{\mathbb{C}} |z|^2 g^{\mathrm{BEC}}(z) \d z \gtrsim N^{2/3}
\label{eq:assumptionLowerBoundParticleNumberBEC}
    \end{equation}
    holds. Then, for $N$ and $p_{\rm{c}}$ large enough, we have
    \begin{equation}
        \Vert \Gamma_{\beta,N} - \widetilde{\Gamma}_{\beta,N} \Vert_1 \lesssim p_{\mathrm{c}}^{-3}.
\label{eq:traceNormDistanceTrialStates}
    \end{equation}
\end{prop}

\begin{remark}
    Using \cite[Lemma~3.4]{DeuNamNap-25} one checks that \eqref{eq:assumptionLowerBoundParticleNumberBEC} is equivalent to the condition $N g^\id \gtrsim N^{2/3}$.
\end{remark}

Before we prove the above proposition, we state and prove a lemma that compares the expected number of particles in the condensate in the two effective theories (with and without cutoff).

\begin{lem}
\label{lem:particleNumberAndChemicalPotentialBEC}
    For  $\beta$ as in  \eqref{eq:beta_asymp}, and for $N$ and $p_{\rm{c}}$ large enough, we have the bounds
    \begin{equation}
        \left| \int_{\mathbb{C}} |z|^2 \left( g^{\mathrm{BEC}}(z) - \widetilde{g}^{\mathrm{BEC}}(z) \right) \mathrm{d}z \right| \lesssim \frac{1}{\beta p_{\mathrm{c}}^5}
        \label{eq:DifferenceExpectedParticleNumberBEC}
    \end{equation}
    and
    \begin{equation}
        | \mu^{\mathrm{BEC}}(\beta,N) - \widetilde{\mu}^{\mathrm{BEC}}(\beta,N) | \lesssim \frac{f(\xi)}{\beta N p_{\mathrm{c}}^5},
        \label{eq:DifferenceChemicalPotentialNumberBEC}
    \end{equation}
    where 
    \begin{equation}
        f(x) = \begin{cases} 1 & \text{if } x \geq -1 \\ x^2 & \text{if } x < -1 \end{cases} \quad \text{ with } \quad \xi = \sqrt{\frac{\beta N}{4 \hat{v}(0)}} \min\{ \mu^{\mathrm{BEC}} , \widetilde{\mu}^{\mathrm{BEC}}  \}.
        \label{eq:perturbationTheoryChemicalPotential2_0}
    \end{equation}
\end{lem}

\begin{remark}\label{rem:gnadtildegexp}
    Using \eqref{eq:assumptionLowerBoundParticleNumberBEC} and \eqref{eq:DifferenceExpectedParticleNumberBEC} one easily checks that $\lim_{N \to \infty} p_{\mathrm{c}} = +\infty$ implies
    \begin{equation}
        \lim_{N \to \infty} \frac{\int_{\mathbb{C}} |z|^2 g^{\mathrm{BEC}}(z) \d z}{\int_{\mathbb{C}} |z|^2 \widetilde{g}^{\mathrm{BEC}}(z) \d z} = 1.
        \label{eq:perturbationTheoryChemicalPotential2_0b}
    \end{equation}
\end{remark}

\begin{proof}
    Since
    \begin{equation}
        \mathrm{Tr}[\mathcal{N} \Gamma_{\beta,N}] = N = \mathrm{Tr}[\mathcal{N} \widetilde{\Gamma}_{\beta,N}],
    \end{equation}
 a short computation shows
    \begin{equation}
        \int_{\mathbb{C}} |z|^2 \left( g^{\mathrm{BEC}}(z) - \widetilde{g}^{\mathrm{BEC}}(z) \right) \mathrm{d}z = \int_{\mathbb{C}} \Tr[ \mathcal{N}_+ \widetilde{G}^{\Bog}(z) ] \widetilde{g}^{\mathrm{BEC}}(z) \d z - \int_{\mathbb{C}} \Tr[ \mathcal{N}_+ G^{\Bog}(z) ] g^{\mathrm{BEC}}(z) \d z.
        \label{eq:AppComparisonStates1}
    \end{equation}
 Let $P_{\mathrm{B}} = \{ p \in \Lambda_+^* \ | \ |p| \leq p_{\mathrm{c}} \}$.  The definitions \eqref{def:cutoffpotential} and \eqref{eq:relationsCoefficients} imply that $s_p=0$ and $c_p=1$ when $p \in P_{\mathrm{B}}^{\mathrm{c}} = \Lambda_+^* \setminus P_{\mathrm{B}}$. Using \eqref{def:gammabogalphabog} and the analogous formula without the cutoff one can easily see that the right hand side of \eqref{eq:AppComparisonStates1} equals
    \begin{equation}
        \sum_{p \in P_{\mathrm{B}}^{\mathrm{c}}} \left( \frac{1}{\exp(\beta(p^2-\mu^{\id}))-1} - (1 + 2 \nu_p^2) \frac{1}{\exp(\beta \omega(p))-1} - \nu_p^2 \right).
        \label{eq:AppComparisonStates2}
    \end{equation}
Then,  using $0 \leq (1+x)^{-1/2} + (1+x)^{1/2} - 2 \leq x^2/4$ for $x \geq 0$ and $\mu^{\id}<0$, one gets for $\nu_p^2$ in \eqref{eq:defvp},
	\begin{equation}
		0 \leq \nu_{p}^2 \leq \frac{(\hat{v}(p) g^{\id})^2}{4 p^4}.
		\label{eq:boundvp}
	\end{equation}
    Moreover, 
	\begin{equation}
		\frac{1}{\exp(\beta \omega(p))-1} \leq \frac{1}{\exp(\beta(p^2-\mu^{\id}))-1} \leq \frac{1}{\beta p^2},
		\label{eq:pwboundgammaB}
	\end{equation}
	which follows from $\omega(p) \geq p^2-\mu^{\id}$, $\mu^{\id} < 0$, and $(\exp(x)-1)^{-1} \leq 1/x$ for $x \geq 0$. 
    
    With these two bounds we estimate
    \begin{align}
    &\left| \sum_{p \in P_{\mathrm{B}}^{\mathrm{c}}} \left( \frac{1}{\exp(\beta(p^2-\mu^{\id}))-1} - (1 + 2 \nu_p^2) \frac{1}{\exp(\beta \omega(p))-1} - \nu_p^2 \right) \right| \leq \nonumber   \\
    &\hspace{2cm}  \sum_{p \in P_{\mathrm{B}}^{\mathrm{c}}} \left( \frac{\hat{v}^2(p)}{4 p^4} \left( 1 + \frac{2}{\beta p^2} \right) \right) 
   + \sum_{p \in P_{\mathrm{B}}^{\mathrm{c}}} \left( \frac{1}{\exp(\beta(p^2-\mu^{\id}))-1} - \frac{1}{\exp(\beta \omega(p))-1} \right). 
    \label{eq:AppComparisonStates3}
    \end{align}
    Since $\sum_{p \in \Lambda_+^*} (1+|p|) \hat{v}(p) < +\infty$ the first term on the right-hand side is bounded by a constant times $p_{\mathrm{c}}^{-6}(1+\beta^{-1} p_{\mathrm{c}}^{-2})$. To estimate the term in the second line we use the identity
    \begin{equation}
        \frac{1}{\exp(\beta \omega(p))-1} = \frac{1}{\exp(\beta(p^2-\mu^{\id}))-1} - \int_0^1 \frac{\beta(p^2-\mu^{\id}) \left( \sqrt{1 + \frac{2 g^{\id} \hat{v}(p)}{p^2-\mu^{\id}}} - 1 \right)}{4 \sinh^2\left( \beta(t (p^2-\mu^{\id}) + (1-t) \omega(p))/2 \right)} \d t.
        \label{eq:AppComparisonStates4}
    \end{equation}
    With $|\sqrt{x+1} - 1| \leq x/2$ for $x \geq 0$   we get
    \begin{equation}
        \left| \sqrt{1 + \frac{2 g^{\id} \hat{v}(p)}{p^2-\mu^{\id}}} - 1 \right| \leq \frac{g^{\id} \hat{v}(p)}{p^2-\mu^{\id}} \leq \frac{\hat{v}(p)}{p^2 - \mu^{\id}}.
        \label{eq:AppComparisonStates5}
    \end{equation}
    In combination, \eqref{eq:AppComparisonStates4}, \eqref{eq:AppComparisonStates5}, $\omega(p) \geq p^2-\mu^{\id}$, $\mu^{\id} < 0$, and $\sinh(x) \geq x$ for $x \geq 0$ imply
    \begin{align}
       & \sum_{p \in P_{\mathrm{B}}^{\mathrm{c}}} \left( \frac{1}{\exp(\beta(p^2-\mu^{\id}))-1} - \frac{1}{\exp(\beta \omega(p))-1} \right) \\
        &\hspace{2cm} \leq \sum_{p \in P_{\mathrm{B}}^{\mathrm{c}}} \frac{\beta \hat{v}(p)}{4 \sinh^2\left(\frac{\beta(p^2-\mu^{\id})}{2} \right)} 
        \leq \frac{1}{\beta p_{\mathrm{c}}^5} \sum_{p \in P_{\mathrm{B}}^{\mathrm{c}}} \hat{v}(p) |p|. 
        \label{eq:AppComparisonStates6}
    \end{align}

    Putting \eqref{eq:AppComparisonStates3} and \eqref{eq:AppComparisonStates5} together, we obtain
    \begin{equation}
        \left| \int_{\mathbb{C}} \Tr[ \mathcal{N}_+ \widetilde{G}^{\Bog}(z) ] \widetilde{g}^{\mathrm{BEC}}(z) \d z - \int_{\mathbb{C}} \Tr[ \mathcal{N}_+ G^{\Bog}(z) ] g^{\mathrm{BEC}}(z) \d z \right| \lesssim \frac{1}{\beta p_{\mathrm{c}}^5}
    \end{equation}
    which by \eqref{eq:AppComparisonStates1} implies \eqref{eq:DifferenceExpectedParticleNumberBEC}.

    The bound \eqref{eq:DifferenceChemicalPotentialNumberBEC} for the difference of the chemical potentials is a direct consequence of \eqref{eq:DifferenceExpectedParticleNumberBEC} and \cite[Lemma~A.3]{DeuNamNap-25}.
\end{proof}
We are now prepared to give the proof of Proposition~\ref{prop:traceNormDifference}.

\begin{proof}[Proof of Proposition~\ref{prop:traceNormDifference}]
    We start by noting that
    \begin{align}
        \Vert \widetilde{\Gamma}_{\beta,N} - \Gamma_{\beta,N} \Vert_1 \leq& \int_{\mathbb{C}} \left\Vert |z \rangle \langle z | \otimes \widetilde{G}^{\mathrm{Bog}}(z) \right\Vert_1 \left| \widetilde{g}^{\mathrm{BEC}}(z) - g^{\mathrm{BEC}}(z) \right| \mathrm{d}z \nonumber \\ 
        &+ \int_{\mathbb{C}} \left\Vert |z \rangle \langle z | \otimes \left( \widetilde{G}^{\mathrm{Bog}}(z) - G^{\mathrm{Bog}}(z) \right) \right\Vert_1  g^{\mathrm{BEC}}(z) \d z.
        \label{eq:traceNormDifferenceTrialStates1}
    \end{align}
    The operator in the first norm on the right-hand side is positive and has trace equal to $1$. Accordingly, the first term equals the $L^1$-norm difference of $\widetilde{g}^{\mathrm{BEC}}$ and $g^{\mathrm{BEC}}$. When we also use $\Vert P \otimes A \Vert_1 = \Vert A \Vert_1$ for a rank-one orthogonal projection $P$ and a trace-class operator $A$ to rewrite the second term, we obtain the bound 
    \begin{equation}
        \Vert \widetilde{\Gamma}_{\beta,N} - \Gamma_{\beta,N} \Vert_1 \leq \int_{\mathbb{C}} \left| \widetilde{g}^{\mathrm{BEC}}(z) - g^{\mathrm{BEC}}(z) \right| \mathrm{d}z + \int_{\mathbb{C}} \left\Vert \widetilde{G}^{\mathrm{Bog}}(z) - G^{\mathrm{Bog}}(z) \right\Vert_1  g^{\mathrm{BEC}}(z) \d z.
        \label{eq:traceNormDifferenceTrialStates2}
    \end{equation}

    Let us have a closer look at the first term on the right-hand side of \eqref{eq:traceNormDifferenceTrialStates2}. We define the relative entropy between probability densities $\varrho(z)$ and $\sigma(z)$ on $\mathbb{C}$ by
    \begin{equation}
        \mathcal{H}(\varrho,\sigma) = \int_{\mathbb{C}} \varrho(z) \left( \ln(\varrho(z)) - \ln(\sigma(z)) \right) \d z. 
    \end{equation}
    A short computation shows that 
    \begin{align}
        \frac{1}{\beta} \mathcal{H}(g^{\mathrm{BEC}},\widetilde{g}^{\mathrm{BEC}}) =& -\frac{1}{\beta} \ln\left( \int_{\mathbb{C}} \exp\left( -\beta \left( \frac{\hat{v}(0)}{2N} |z|^4 - \mu^{\mathrm{BEC}} |z|^2 \right) \right) \d z \right) + ( \mu^{\mathrm{BEC}} - \widetilde{\mu}^{\mathrm{BEC}} ) M \nonumber \\
        &+ \frac{1}{\beta} \ln\left( \int_{\mathbb{C}} \exp\left( -\beta \left( \frac{\hat{v}(0)}{2N} |z|^4 - \widetilde{\mu}^{\mathrm{BEC}} |z|^2 \right) \right) \d z \right),
    \end{align}
    where $M = \int |z|^2 g^{\mathrm{BEC}}(z) \d z$. The first term on the right-hand side is a concave function of $\mu^{\mathrm{BEC}}$, and therefore satisfies
    \begin{align}
        &-\frac{1}{\beta} \ln\left( \int_{\mathbb{C}} \exp\left( -\beta \left( \frac{\hat{v}(0)}{2N} |z|^4 - \mu^{\mathrm{BEC}} |z|^2 \right) \right) \d z \right) \leq (\widetilde{\mu}^{\mathrm{BEC}} - \mu^{\mathrm{BEC}}) \widetilde{M} \nonumber \\
        &\hspace{5cm}- \frac{1}{\beta} \ln\left( \int_{\mathbb{C}} \exp\left( -\beta \left( \frac{\hat{v}(0)}{2N} |z|^4 - \widetilde{\mu}^{\mathrm{BEC}} |z|^2 \right) \right) \d z \right)
        \label{eq:traceNormDifferenceTrialStates4}
    \end{align}
    with $\widetilde{M} = \int |z|^2 \widetilde{g}^{\mathrm{BEC}}(z) \d z$.     Putting these two inequalities together, we find
    \begin{equation}
        \frac{1}{\beta} \mathcal{H}(g^{\mathrm{BEC}},\widetilde{g}^{\mathrm{BEC}}) \leq ( \mu^{\mathrm{BEC}} - \widetilde{\mu}^{\mathrm{BEC}} ) (M-\widetilde{M} ) \lesssim \frac{f(\xi)}{\beta^2 p_{\mathrm{c}}^{10} N}. 
    \end{equation}
    To obtain the last bound we applied Lemma~\ref{lem:particleNumberAndChemicalPotentialBEC}. Finally, an application of the classical Pinsker's inequality $\mathcal{H}(\varrho,\sigma) \geq \frac{1}{2} \Vert \varrho - \sigma \Vert_{L^1(\mathbb{C})}^2$ allows us to conclude that
    \begin{equation}
        \int_{\mathbb{C}} \left| \widetilde{g}^{\mathrm{BEC}}(z) - g^{\mathrm{BEC}}(z) \right| \mathrm{d}z \lesssim \frac{ \sqrt{f(\xi)} }{ p_{\mathrm{c}}^5 \sqrt{\beta N}}.
        \label{eq:traceNormDifferenceTrialStates11}
    \end{equation}  
    We have thus obtained a bound for the first term on the right-hand side of \eqref{eq:traceNormDifferenceTrialStates2}. It remains to consider the second term.

    For this term we will apply a similar argument based on the quantum version of Pinsker's inequality, which reads $\mathcal{S}(\Gamma,\Gamma') \geq \frac{1}{2} \Vert \Gamma - \Gamma' \Vert_1^2$, see e.g. \cite[Theorem~1.15]{OhyPet-93}. Here 
    \begin{equation}
        \mathcal{S}(\Gamma,\Gamma') = \Tr[ \Gamma \left( \ln(\Gamma) - \ln(\Gamma') \right) ]
        \label{eq:traceNormDifferenceTrialStates12}
    \end{equation}
    denotes the quantum relative entropy of $\Gamma$ with respect to $\Gamma'$. A short computation shows that
    \begin{align}
      \frac{1}{\beta}  \mathcal{S}(\widetilde{G}^{\mathrm{Bog}}(z),G^{\mathrm{Bog}}(z)) =&\frac{g^{\id}}{2}\sum_{p \in P_{\mathrm{B}}^{\mathrm{c}}} \hat{v}(p) \Tr[ ( 2 a_p^* a_{p} + (z^2/|z|^2) a_p^* a_{-p}^* + (\overline{z}^2/|z|^2) a_p a_{-p} ) \widetilde{G}^{\mathrm{Bog}}(z) ] \nonumber \\
        &+ \frac{1}{\beta} \sum_{p \in P_{\mathrm{B}}^{\mathrm{c}}} \left( \ln \left( 1 - \exp(-\beta(p^2 - \mu^{\id})) \right) - \ln \left( 1 - \exp(-\beta \omega(p) ) \right) \right). 
\label{eq:traceNormDifferenceTrialStates13}
    \end{align}
    The first summand in the second term is a concave function of $\beta(p^2 - \mu^{\id})$, and hence 
    \begin{equation}
        \frac{1}{\beta} \sum_{p \in P_{\mathrm{B}}^{\mathrm{c}}} \ln \left( 1 - \exp(-\beta(p^2 - \mu^{\id})) \right) \leq \frac{1}{\beta} \sum_{p \in P_{\mathrm{B}}^{\mathrm{c}}} \left( \ln \left( 1 - \exp(-\beta \omega(p)) \right) + \frac{\beta((p^2-\mu^{\id})- \omega(p)) }{\exp(\beta \omega(p)) - 1} \right).
        \label{eq:traceNormDifferenceTrialStates14}
    \end{equation}
    The second term on the right-hand side equals
    \begin{align}
        \sum_{p \in P_{\mathrm{B}}^{\mathrm{c}}} \frac{(p^2-\mu^{\id}) \left(1 -  \sqrt{ 1 + \frac{2 g^{\id} \hat{v}(p)}{p^2-\mu^{\id}} } \right) }{\exp(\beta \omega(p)) - 1} \leq& -g^{\id} \sum_{p \in P_{\mathrm{B}}^{\mathrm{c}}}  \frac{\hat{v}(p)}{\exp(\beta \omega(p)) - 1} \nonumber \\
        &+ \frac{1}{2} \sum_{p \in P_{\mathrm{B}}^{\mathrm{c}}} \frac{1}{\exp(\beta \omega(p)) - 1} \left( \frac{\hat{v}(p)}{p^2-\mu^{\id}} \right)^2,        \label{eq:traceNormDifferenceTrialStates15}
    \end{align}
    where we used $1 - \sqrt{1+x} \leq -x/2 +x^2/8$ for $x \geq 0$ to obtain the bound. The second term on the right-hand side is bounded from above by a constant times $\beta^{-1} p_{\mathrm{c}}^{-8}$.
    
    Using \eqref{def:gammabogalphabog} one immediately sees that the first term on the right-hand side of \eqref{eq:traceNormDifferenceTrialStates13} equals
    \begin{equation}
        g^{\id} \sum_{p \in P_{\mathrm{B}}^{\mathrm{c}}}  \frac{\hat{v}(p)}{\exp(\beta (p^2-\mu^{\id})) - 1}.
        \label{eq:traceNormDifferenceTrialStates16}
    \end{equation}
    Putting these bounds together, we have shown that the relative entropy satisfies 
    \begin{equation}
       \frac{1}{\beta}   \mathcal{S}(\widetilde{G}^{\mathrm{Bog}}(z),G^{\mathrm{Bog}}(z)) \leq g^{\id} \sum_{p \in P_{\mathrm{B}}^{\mathrm{c}}} \hat{v}(p) \left( \frac{1}{\exp(\beta (p^2-\mu^{\id}))) - 1} - \frac{1}{\exp(\beta \omega(p)) - 1} \right) + \frac{C}{\beta p_{\mathrm{c}}^{8}}.
        \label{eq:traceNormDifferenceTrialStates17}
    \end{equation}
    Using \eqref{eq:AppComparisonStates4} and $ \sqrt{1+x} -1\leq x/2$ for $x \geq 0$ again, one easily checks that first term on the right-hand side of \eqref{eq:traceNormDifferenceTrialStates17} is bounded from above by a constant times $\beta^{-1} p_{\mathrm{c}}^{-6}$. In combination with the quantum version of Pinsker's inequality from above, this implies the final estimate 
    \begin{equation}
        \Vert \widetilde{G}^{\mathrm{Bog}}(z) - G^{\mathrm{Bog}}(z) \Vert_1 \lesssim \frac{1}{ p_{\mathrm{c}}^3}
        \label{eq:traceNormDifferenceTrialStates18}
    \end{equation}
    for the trace norm distance of $\widetilde{G}^{\mathrm{Bog}}(z)$ and $G^{\mathrm{Bog}}(z)$. When we put \eqref{eq:traceNormDifferenceTrialStates2}, \eqref{eq:traceNormDifferenceTrialStates11} and \eqref{eq:traceNormDifferenceTrialStates18} together, we obtain a proof of \eqref{eq:traceNormDistanceTrialStates}. To obtain this result we also used that \eqref{eq:assumptionLowerBoundParticleNumberBEC}, \eqref{eq:perturbationTheoryChemicalPotential2_0b} and \cite[Lemma~A.1~(c)]{DeuNamNap-25} imply $|\min\{ \mu^{\mathrm{BEC}}, \widetilde{\mu}^{\mathrm{BEC}} \}| \lesssim 1$. 
    \end{proof}

\subsection{Statistical properties of the condensate fraction} \label{sec:gwasserstein} Recall that in Theorem \ref{thm:semiclassicalApproximation}(c) we have introduced the random variable $Z$ which has distribution $g^{\BEC}$ defined in \eqref{eq:gBEC}. 
Several properties of $Z$ have been proven in \cite{DeuNamNap-25}. One of them estimates the variance of $|Z|^2$. More precisely, in \cite[Lemma A.8.(a)]{DeuNamNap-25} it has been shown that in the condensed phase ($\kappa>1$) we have
\begin{equation}\label{eq:varZsq}
  \mathbf{E}[(|Z|^2-\mathbf{E}[|Z|^2])^2]=\frac{N}{\beta \hat{v}(0)}+o(N^{\frac53}).  
\end{equation}
In particular, since $\mathbf{E}[|Z|^2]=Ng^{\id}+O(N^{2/3})$, we have 
\begin{equation}\label{eq:varZsqNgid}
  \mathbf{E}[(|Z|^2-Ng^{\id})^2]=\frac{N}{\beta \hat{v}(0)}+o(N^{\frac53}).
    \end{equation}
Furthermore, in \cite[Lemma A.10.(a)]{DeuNamNap-25} it is shown that if $X$ is a standard Gaussian random variable, then
\begin{equation} \label{eq:condnormal}
    \sqrt{\frac{\beta\hat{v}(0)}{N}}(|Z|^2-\mathbf{E}[|Z|^2]) \xrightarrow{N\to \infty} X
\end{equation}
in distribution. The proofs of \eqref{eq:varZsq} and \eqref{eq:condnormal}, together with \eqref{eq:DifferenceChemicalPotentialNumberBEC} and Remark \ref{rem:gnadtildegexp}, imply that these statements  hold true whether $Z$ has density $g^{\BEC}$ or $\widetilde{g}^{\BEC}$.

For further reference we shall need a stronger convergence result than \eqref{eq:condnormal}. 

\begin{prop} \label{prop:W2condtoX}
  Let $Z$ be a random variable with density $g^{\BEC}$ and let $X$ be a standard Gaussian random variable. Then, for $\kappa>1$ there exists $c>0$ such that 
    \begin{equation}\label{eq:W2condtoX}
        W_2 \left( \sqrt{\frac{\beta}{N}} (|Z|^2 - \mathbf E[|Z|^2]), \frac{1}{\sqrt{\widehat v(0)}} X \right) \lesssim e^{- c N^{\frac13}}.
    \end{equation}
\end{prop}
\begin{proof}
Let $Y \sim \mathcal N (\mu, \sigma^2)$ with $\mu = \frac{N \widetilde \mu^\BEC}{\widehat v(0)}$ and $\sigma^2 = \frac{N}{\beta \widehat v(0)}$. Furthermore, let $Y_{\geq 0}$ and $Y_{<0}$ be the random variables obtained by conditioning $Y$ on the events $Y \geq 0$ and $Y<0$, respectively. We~observe that 
\begin{equation}\label{eq:1-p}
    1-p \coloneq \mathbf P(Y <0) \lesssim e^{-c N^{\frac13}}
    \end{equation}
    for some $c >0$ because $\frac{\mu^2}{\sigma^2} \gtrsim N^{\frac13}$. Secondly, we note that $Y_{\geq 0} \stackrel{d}{=} |Z|^2$. 
We construct a coupling of $Y$ and $Y_{\geq 0}$ by introducing $B \sim \mathrm{Bernoulli}(p)$ independent of $Y_{\geq 0}$ and $Y_{<0}$. Then we realize $Y$ on the product of the probability spaces of $Y_{\geq 0}$, $Y_{<0}$ and $B$ by $Y = B Y_{\geq 0} + (1-B) Y_{<0}$. 

We have the bound
\begin{equation}
\begin{aligned}\label{eq:y-y0}
    W_2(Y,Y_{\geq 0})^2 \leq \mathbf E[(Y-Y_{\geq 0})^2] &= (1-p) \mathbf E[(Y_{<0}-Y_{\geq 0})^2] \\ 
   & \leq 2 (1-p) \mathbf E[ (Y_{< 0} - \mu)^2 + (Y_{\geq 0} -\mu)^2]. 
\end{aligned}
\end{equation}
The first term can be estimated using Cauchy-Schwarz:
\begin{equation}\label{eq:y-y01}
    (1-p) \mathbf E[ (Y_{< 0} - \mu)^2 ] = \mathbf E[(Y - \mu)^2 \mathds 1(Y < 0)] \leq \sqrt{\mathbf E[(Y-\mu)^4]} \sqrt{\mathbf E[\mathds 1(Y <0)]}  =\sqrt{3 \sigma^4}  \sqrt{1-p},
\end{equation}
and the second term directly:
\begin{equation}\label{eq:y-y02}
    (1-p) \mathbf E[ (Y_{\geq 0} - \mu)^2 ] = \frac{1-p}{p} \mathbf E[(Y - \mu)^2 \mathds 1(Y \geq 0)] \leq \frac{1-p}{p} \sigma^2.
\end{equation}
The proof is completed by the following calculation
\begin{align}
W_2 \left( \sqrt{\frac{\beta}{N}} (|Z|^2 - \mathbf E[|Z|^2]), \frac{1}{\sqrt{\widehat v(0)}} X \right)  & = \sqrt{\frac{\beta}{N}} W_2 (Y_{\geq 0} - \mathbf E[Y_{\geq 0}],Y-\mathbf E[Y]) \\
& \leq \sqrt{\frac{\beta}{N}} W_2(Y_{\geq 0}, Y) + \sqrt{\frac{\beta}{N}} |\mathbf E[Y_{\geq 0} - \mathbf E[Y]]| \nonumber \\
& \leq 2 \sqrt{\frac{\beta}{N}} \mathbf E[ (Y_{\geq 0} - Y)^2], \nonumber
\end{align}
where in the last step we used Cauchy-Schwarz inequality. Using \eqref{eq:y-y0},\eqref{eq:y-y01},\eqref{eq:y-y02} we obtain
\begin{equation*}
  W_2 \left( \sqrt{\frac{\beta}{N}} (|Z|^2 - \mathbf E[|Z|^2]), \frac{1}{\sqrt{\widehat v(0)}} X \right) \leq 4 \sigma^2 \sqrt{\frac{N}{\beta} } \left( \sqrt{3} \sqrt{1-p} + \frac{1-p}{p} \right),  
\end{equation*}
which, together with \eqref{eq:1-p}, yields the desired result.
\end{proof}

\section{Semiclassical analysis and proof of Theorem 1} \label{sec:semiclassicalApproximation}
Let $(\mathcal H_k)_{k=1}^\infty$ be a sequence of Hilbert spaces, each with a distinguished unit vector $\Omega_k \in \mathcal H_k$. The infinite tensor product $\bigotimes_{k=1}^\infty \mathcal H_k$ is a Hilbert space defined as the completion of $\bigcup_{K=1}^\infty \bigotimes_{k=1}^K \mathcal H_k$, where we identify the Hilbert space $\bigotimes_{k=1}^K \mathcal H_k$ with a closed subspace of $\bigotimes_{k=1}^{K+1} \mathcal H_k $ using the isometric embedding $h \mapsto h \otimes \Omega_{K+1}$. The notation $\bigotimes_{k=1}^\infty \mathcal H_k$ is a slight abuse of notation because the infinite tensor product depends on the choice of the distinguished vectors, not only on the Hilbert spaces $\mathcal H_k$. If~all $\mathcal H_k$ are separable, so is their tensor product. 

The exponential property of Fock spaces allows to realize $\mathscr F_+$ as the infinite tensor product 
\begin{equation}
    \mathscr F_+ = \bigotimes_{p \in A} \mathscr F_{\pm p}, \qquad \text{where} \qquad \mathscr F_{\pm p} = \mathscr{F}(\mathrm{span} \{ e^{\mathrm{i}p \cdot x}, e^{-\mathrm{i}p \cdot x } \}),
\end{equation}
and the distinguished vector in the two-mode Fock space $\mathscr F_{\pm p}$ is the Fock vacuum $(1,0,0,\dots)$. We recall the set $A$, defined before \eqref{eq:familyOfGaussianRandomVariables}. The Bogoliubov state $\widetilde G^\Bog$ is a~product operator under this decomposition of the Hilbert space:
\begin{equation}
    \widetilde G^{\Bog} = \bigotimes_{p \in A} \widetilde G^\Bog_{\pm p}, \quad \text{where} \quad \widetilde G^{\Bog}_{\pm p} = \frac{\exp(- \beta \widetilde{\mathcal H}^\Bog_{\pm p})}{\mathrm{Tr}[\exp(- \beta \widetilde{\mathcal H}^\Bog_{\pm p})]}.
    \label{eq:G_tensor_decomp}
\end{equation}
Here (recall the notation in \eqref{eq:z=1notation}) $\widetilde{\mathcal H}^{\mathrm{Bog}}_{\pm p}$ is the two-mode Bogoliubov Hamiltonian:
\begin{equation}
 \widetilde{\mathcal H}^{\mathrm{Bog}}_{\pm p}=  (p^2-\mu^\id +g^{\mathrm{id}}\hat{u}(p)) (a_p^* \a_p +a_{-p}^* \a_{-p})+ g^{\mathrm{id}} \hat{u}(p) \left( a_p^* a_{-p}^* +  \a_p \a_{-p} \right).
    \label{BogoliubovHamiltonianTwomodes}
\end{equation}

We begin the discussion of the semiclassical structure of $\widetilde G^{\Bog}$ by analyzing its tensor factors. Let us first review the definitions of various operator symbols used in this context (see e.g.\ \cite{Folland} for an introduction to this subject). Consider the Weyl operator 
\begin{equation}\label{def:Weyl}
W_{\pm p}(f_p,f_{-p}) = \exp \left( f_p a_p^* + f_{-p} a_{-p}^* - \overline f_p \a_p - \overline f_{-p} \a_{-p} \right).
\end{equation}
The quantum characteristic function of $\widetilde G_{\pm p}$ is defined as
\begin{equation}
    \chi_{\pm p}(f_p, f_{-p}) = \mathrm{Tr}[\widetilde G_{\pm p}^\Bog \, W_{\pm p}(f_p,f_{-p})]. 
    \label{eq:quantum_characteristic}
\end{equation}
It is an analog of the characteristic function in probability theory. One can express $\widetilde G^{\Bog}_{\pm}$ in terms of the quantum characteristic function as follows:
\begin{equation}
    \widetilde G_{\pm p}^\Bog = \int_{\mathbb C^2} \chi_{\pm p}(f_p, f_{-p}) W_{\pm p}(-f_p,-f_{-p}) \d f_p \d f_{-p},
    \label{eq:Wigner_Weyl_rep_of_G}
\end{equation}
where $ \d f_{ p} = \frac{1}{\pi} \d \mathrm{Re}(f_p) \d \mathrm{Im}(f_p)$. The Fourier transform of $\chi_{\pm p}$ is called the Weyl-Wigner symbol:
\begin{equation}
    \Sigma_{\pm p}^{W}(\phi_p,\phi_{-p}) = \int_{\mathbb C^2} \chi_{\pm p}(f_p,f_{-p}) e^{\phi_p \overline f_p + \phi_{-p} \overline f_{-p} - \overline \phi_p f_p - \overline \phi_{-p} f_{-p}} \d f_p \d f_{-p}.
    \label{eq:Wigner_function}
\end{equation}
 Now consider the coherent states in $\mathscr F_{\pm p}$:
\begin{equation}
    \Omega_{\phi_p, \phi_{-p}} = W_{\pm p}(\phi_p,\phi_{-p}) \Omega = e^{- \frac12 |\phi_p|^2 -\frac12 |\phi_{-p}|^2} e^{\phi_p a_p^* + \phi_{-p} a_{-p}^*} \Omega,
\end{equation}
where $\Omega$ is the Fock vacuum. The lower symbol of $\widetilde G^{\Bog}_{\pm p}$ is defined as
 \begin{equation}
    \Sigma^{\downarrow}_{\pm p}(\phi_p,\phi_{-p}) = \langle \Omega_{\phi_p, \phi_{-p}} | \widetilde G^{\Bog}_{\pm p} \Omega_{\phi_p, \phi_{-p}} \rangle. 
\end{equation}
Inserting \eqref{eq:Wigner_Weyl_rep_of_G} into the definition of $\Sigma^\downarrow_{\pm p}$ and using \eqref{eq:Wigner_function} one finds the standard relation:
\begin{equation}
\Sigma^{\downarrow}_{\pm p}(\phi_p,\phi_{-p})  = \int_{\mathbb C^2} 4 e^{-2 |\phi_p - f_p|^2 -2 |\phi_{-p} - f_{-p}|^2} \Sigma^W_{\pm}(f_p,f_{-p}) df_p df_{-p}. 
\label{eq:lower_from_Weyl}
\end{equation}
The upper symbol of $\widetilde G^\Bog_{\pm p}$ is a function $\Sigma^\uparrow_{\pm p}$ such that
\begin{equation}
\widetilde G^{\Bog}_{\pm p} = \int_{\mathbb C^2} | \Omega_{\phi_{p},\phi_{-p}} \rangle \langle  \Omega_{\phi_{p},\phi_{-p}} | \, \Sigma^{\uparrow}_{\pm p} (\phi_p , \phi_{-p}) \d \phi_p \d \phi_{-p}.
\label{eq:upper_symbol}
\end{equation}
We highlight that not every bounded operator admits an upper symbol. The existence of an upper symbol of $\widetilde G^{\Bog}_{\pm p}$ is discussed in the lemmas below. If there exists an upper symbol, then inserting \eqref{eq:upper_symbol} into \eqref{eq:quantum_characteristic} and using \eqref{eq:Wigner_function} one finds
\begin{equation}
    \Sigma^W_{\pm p}(\phi_p ,\phi_{-p}) = \int_{\mathbb C^2} 4 e^{-2 |\phi_p - f_p|^2 -2 |\phi_{-p} - f_{-p}|^2} \Sigma^\uparrow_{\pm}(f_p,f_{-p}) df_p df_{-p}.
    \label{eq:Weyl_from_upper}
\end{equation}

\begin{lem}[Semiclassical symbols of $\widetilde G^\Bog_{\pm p}$] \label{lem:symbols}
$\Sigma^W_{\pm p}(\phi_p,\phi_{-p})$ is the probability density of a centered complex Gaussian random vector $(\widehat \Phi_N(p),\widehat \Phi_N(-p))$ with covariance and pseudo-covariance
\begin{align}
	  \mathbf{E} \begin{pmatrix}
		|\widehat \Phi_N(p)|^2 & \widehat \Phi_N(p) \overline{\widehat \Phi_N(-p)} \\ \overline{\widehat \Phi_N(p)} \widehat \Phi_N(-p) & |\widehat \Phi_N(-p)|^2 		\end{pmatrix} &=  \begin{pmatrix}
	\widetilde{\gamma}^{\rm{Bog}}(p) + \frac12 & 0 \\ 0 & \widetilde{\gamma}^{\rm{Bog}}(p) + \frac12 
	\end{pmatrix},     \label{eq:WW-cov} \\
    \mathbf{E} \begin{pmatrix}
	\widehat \Phi_N(p)^2 & \widehat \Phi_N(p) \widehat \Phi_N(-p) \\ \widehat \Phi_N(p) \widehat \Phi_N(-p) & \widehat \Phi_N(-p)^2 
	\end{pmatrix} &=   \begin{pmatrix}
	0 & \widetilde{\alpha}^{\rm{Bog}}(p) \\ \widetilde{\alpha}^{\rm{Bog}}(p) & 0 
	\end{pmatrix}. \nonumber
\end{align}
The lower symbol $\Sigma^\downarrow_{\pm p}$ is a centered complex Gaussian probability density characterized by
\begin{align}
	  \mathbf{E}  \begin{pmatrix}
		|\widehat \Phi_N(p)|^2 & \widehat \Phi_N(p) \overline{\widehat \Phi_N(-p)} \\ \overline{\widehat \Phi_N(p)} \widehat \Phi_N(-p) & |\widehat \Phi_N(-p)|^2 		\end{pmatrix} &=   \begin{pmatrix}
	\widetilde{\gamma}^{\rm{Bog}}(p) + 1 & 0 \\ 0 & \widetilde{\gamma}^{\rm{Bog}}(p) + 1 
	\end{pmatrix},  \\
    \mathbf{E} \begin{pmatrix}
	\widehat \Phi_N(p)^2 & \widehat \Phi_N(p) \widehat \Phi_N(-p) \\ \widehat \Phi_N(p) \widehat \Phi_N(-p) & \widehat \Phi_N(-p)^2 
	\end{pmatrix} &=   \begin{pmatrix}
	0 & \widetilde{\alpha}^{\rm{Bog}}(p) \\ \widetilde{\alpha}^{\rm{Bog}}(p) & 0 
	\end{pmatrix}. \nonumber
\end{align}
Suppose that $|\widetilde{\alpha}^{\rm{Bog}}(p)| < \widetilde{\gamma}^{\rm{Bog}}(p)$. Then $\widetilde G^\Bog_{\pm p}$ admits an upper symbol $\Sigma^\uparrow_{\pm p}$, which is a~centered complex Gaussian probability density characterized by
\begin{align}
	  \mathbf{E}  \begin{pmatrix}
		|\widehat \Phi_N(p)|^2 & \widehat \Phi_N(p) \overline{\widehat \Phi_N(-p)} \\ \overline{\widehat \Phi_N(p)} \widehat \Phi_N(-p) & |\widehat \Phi_N(-p)|^2 		\end{pmatrix} &=  \begin{pmatrix}
	\widetilde{\gamma}^{\rm{Bog}}(p) & 0 \\ 0 & \widetilde{\gamma}^{\rm{Bog}}(p) 
	\end{pmatrix}, \label{eq:upper_symbol_covariance}  \\
    \mathbf{E} \begin{pmatrix}
	\widehat \Phi_N(p)^2 & \widehat \Phi_N(p) \widehat \Phi_N(-p) \\ \widehat \Phi_N(p) \widehat \Phi_N(-p) & \widehat \Phi_N(-p)^2 
	\end{pmatrix} &=   \begin{pmatrix}
	0 & \widetilde{\alpha}^{\rm{Bog}}(p) \\ \widetilde{\alpha}^{\rm{Bog}}(p) & 0 
	\end{pmatrix}. \nonumber
\end{align}
\end{lem}
\begin{proof}
Using \eqref{eq:ab_unitary} we can write
\begin{equation}
  \chi_{\pm p}(f_p,f_{-p})=(1-e^{-\beta \epsilon(p)})^2 {\rm{Tr}}_{\mathscr F_{\pm p}} \big[ e^{- \beta \epsilon(p)(a^*_p \a_p +a^*_{-p} \a_{-p})} U^* W_{\pm p}(f_p,f_{-p}) U \big] ,  
\end{equation}
 where $U$ is chosen as in \eqref{eq:ab_unitary}. 
Using \eqref{eq:Bogoliubov}  
we find
\begin{equation}
    U^* W_{\pm p}(f_p,f_{-p})U = W_{\pm p}(c_p f_p + s_p \overline f_{-p}, c_p f_{-p} + s_p \overline f_{p}). 
\end{equation}
Next, we compute the trace using the coherent state representation:
\begin{equation}
    {\rm{Tr}}_{\mathscr F_{\pm p}} \big[ T \big] = \int_{\mathbb C^2} \langle \Omega_{\phi_p, \phi_{-p}} | T \Omega_{\phi_p,\phi_{-p}} \rangle \d \phi_p \d \phi_{-p},
    \label{eq:trace_coh_rep}
\end{equation}
with $T = e^{- \beta \epsilon(p) (a_p^* \a_p + a_{-p}^* \a_{-p})} W_{\pm p} (c_p f_p + s_p \overline f_{-p}, c_p f_{-p} + s_p \overline f_{p})$. We evaluate $\langle \Omega_{\phi_p, \phi_{-p}} | T \Omega_{\phi_p,\phi_{-p}} \rangle$ and compute the resulting Gaussian integral in \eqref{eq:trace_coh_rep}. The final result is 
\begin{equation}
  \chi_{\pm p}(f_p,f_{-p})=\exp \left(  -\frac12 \coth{\frac{\beta \epsilon(p)}{2}} (|c_p f_p + s_p \overline f_{-p}|^2+|c_p f_{-p} + s_p \overline f_{p}|^2) \right).
\end{equation}
This function is the characteristic function of the Gaussian probability described in \eqref{eq:WW-cov}. 

We compute the lower symbol by combining \eqref{eq:lower_from_Weyl} with the additivity of Gaussian covariances under convolution. Using \eqref{eq:Weyl_from_upper}, we find that if an upper symbol exists in the space of tempered distributions, it is Gaussian and satisfies \eqref{eq:upper_symbol_covariance}. Equation \eqref{eq:upper_symbol_covariance} defines a covariance of a Gaussian only if $|\widetilde{\alpha}^{\rm{Bog}}(p)| \leq \widetilde{\gamma}^{\rm{Bog}}(p)$, with a degenerate Gaussian (supported on a~proper $\mathbb R$-linear subspace) in the case of equality. 
\end{proof}

We remark that quantities $\widetilde{\gamma}^{\rm{Bog}}(p) +\frac12 $, $\widetilde{\gamma}^{\rm{Bog}}(p)$ and $\widetilde{\gamma}^{\rm{Bog}}(p) +1$ in Lemma \ref{lem:symbols} are the averages of $\frac12 (a_p^* \a_p + \a_p a_p^*)$, $a_p^* \a_p$ and $\a_p a_p^*$. These are operators with, respectively, Weyl-Wigner, lower and upper symbol $|\widehat \Phi_N(p)|^2$.  The fact that lower (resp. upper) symbol of the observable is used in the description of the upper (resp. lower) symbol of the state reflects the duality between upper and lower symbols. The~Weyl-Wigner construction is self-dual.

\begin{lem} \label{lem:cutoff_sufficient}
The condition $|\widetilde{\alpha}^{\rm{Bog}}(p)| < \widetilde{\gamma}^{\rm{Bog}}(p)$, required for the existence of the upper symbol $\Sigma^{\uparrow}_{\pm p}$, is satisfied for all sufficiently large $N$, possibly dependent on $p$. If the cutoff parameter $p_c$ in \eqref{def:cutoffpotential} is chosen so that $p_c = O(N^{\frac13})$, then $|\widetilde{\alpha}^{\rm{Bog}}(p)| < \widetilde{\gamma}^{\rm{Bog}}(p)$ holds for large enough $N$ independent of $p$.
\end{lem}
\begin{proof}
One checks that $|\widetilde \alpha^{\rm{Bog}}(p)| < \widetilde \gamma^{\rm{Bog}}(p)$ is equivalent to $\frac{s_p}{c_p} < e^{- \beta \epsilon(p)}$. Then we observe that $\limsup_{N \to \infty} \frac{s_p}{c_p} <1$, while $\lim_{N \to \infty} e^{- \beta \epsilon(p)} =1$. 

For the second claim, we argue that there exists no sequence $(p_n,N_n)_{n=1}^\infty$ of $p_n \in \Lambda^*_+$ and $N$-values, with corresponding $\beta_n = \beta(N_n)$, such that $N_n \to \infty, \ p_n = O(N_n^{\frac13})$, and $\frac{s_{p_n}}{c_{p_n}} \geq e^{-\beta_n \epsilon (p_n)}$ for all $n$. It is enough to consider the cases of constant $p_n$ and $p_n \to \infty$. The former is already excluded above. In~the latter case, $\frac{s_{p_n}}{c_{p_n}} \to 0$ while $\liminf_{n \to \infty} e^{- \beta_n \epsilon(p_n)} > 0$, so $\frac{s_{p_n}}{c_{p_n}} < e^{-\beta_n \epsilon (p_n)}$ for large $n$. 
\end{proof}

From now on we assume that $p_c = O(N^{\frac13})$ and $N$ is large enough such that $|\widetilde{\alpha}^{\rm{Bog}}(p)| < \widetilde{\gamma}^{\rm{Bog}}(p)$ holds for all $p \in \Lambda^*_+$. We consider the random Fourier series, also called a field,
    \begin{equation}
        \Phi_N(x) = \sum_{p \in \Lambda^*_+} \widehat \Phi_N(p) e^{\mathrm{i}px},
        \label{eq:random_Phi}
    \end{equation}
where $(\widehat \Phi_N(p),\widehat \Phi_N(-p))$ is a random vector with distribution $\Sigma_{\pm p}^\uparrow$, and $(\widehat \Phi_N(p),\widehat \Phi_N(-p))$ is independent of $(\widehat \Phi_N(q), \widehat \Phi_N(-q))$ for $q \neq \pm p$. For the underlying probability space, one can be choose the space $X $ of all complex sequences $(\phi_p)_{p \in \Lambda^*_+}$, equipped with the product topology and the product $\sigma$-algebra. 

\begin{lem}
    Suppose that $(X_i)_{i \in I}$ is a countable collection of second-countable topological spaces. Let $X = \prod_{i \in I} X_i$ be equipped with the product topology. Then the Borel $\sigma$-algebra of $X$ coincides with the product $\sigma$-algebra, that is the $\sigma$-algebra of subsets of $X$ generated by the sets $\prod_{i \in I} E_i$, where $E_i \subset X_i$ is Borel and $E_i = X_i$ for all but finitely many $i \in I$.
\end{lem}

We omit the proof of the elementary lemma above. Its significance here is that the product $\sigma$-algebra on $X \cong \prod_{ p \in A} \mathbb C^2$ is the natural domain of definition for the product measure 
\begin{equation}
    \d \mathbf P_{\Phi_N}(\phi) = \prod_{p \in A} \Sigma_{\pm p}^\uparrow(\phi_p, \phi_{-p}) \d \phi_p \d \phi_{-p}.
    \label{eq:product_measure}
\end{equation}

\begin{lem} \label{lem:Phi_regularity}
$\Phi_N(x)$ is almost surely a smooth function on $\Lambda$. In particular, $\Phi_N \in L^2(\Lambda)$ almost surely, and $\mathbf P_{\Phi_N}$ restricts to a Borel probability measure on $L^2(\Lambda)$. 
\end{lem}
\begin{proof}
For this discussion, let $X'$ be the space of all complex sequences $(\phi_p)_{p \in \Lambda^*}$. Thus $X' = \mathbb C \times X$.

Every sequence $(\phi_p) \in X'$ defines a formal Fourier series:
\begin{equation}
    \phi(x) = \sum_{p \in \Lambda^*} \phi_p e^{\mathrm{i}px},
\end{equation}
while sequences in $X$ correspond to Fourier series with vanishing constant term. We identify the sequence $(\phi_p)$ with the corresponding Fourier series $\phi$. The Fourier series $\phi$ is an element of the Sobolev space $H^s(\Lambda)$ provided that
\begin{equation}
  \| \phi \|_{H^s(\Lambda)}^2 = \sum_{p \in \Lambda^*} (1+p^2)^s |\phi_p|^2 < \infty.
\end{equation}
The function $X' \ni \phi \mapsto \| \phi \|_{H^s(\Lambda)}^2 \in [0,\infty]$ is the supremum of a countable family of continuous functions, so it is Borel. Since $ H^s(\Lambda)$ is the preimage of $[0,\infty)$ under $\| \cdot \|_{H^s(\Lambda)}^2$, it is a Borel subset of $X'$. Similarly, every open ball in $H^s(\Lambda)$ is Borel in $X'$. By separability of $H^s(\Lambda)$, every open set in $H^s(\Lambda)$ is the union of a countable collection of open balls, so it is Borel in $X'$. Therefore, a Borel probability measure on $X'$ such that $H^s(\Lambda)$ is of measure $1$ can be restricted to a Borel probability measure on $H^s(\Lambda)$. 

By the monotone convergence theorem, for all $s \in \mathbb R$,
\begin{equation}
 \mathbf E[ \| \Phi_N \|^2_{H^s(\Lambda)}] = \sum_{p \in \Lambda^*_+} (1+p^2)^s \widetilde \gamma^{\rm{Bog}}(p) < \infty,
\end{equation}
because $\widetilde \gamma^{\rm{Bog}}(p)$ decays with $p$ faster than any power. Therefore, $\| \Phi_N \|_{H^s(\Lambda)} $ is finite $\mathbf P_{\Phi_N}$-almost surely. We have $\mathcal C^\infty(\Lambda) = \cap_{s=0}^\infty H^s(\Lambda)$ by Sobolev embeddings, so $\Phi_N$ is smooth almost surely. 

Specializing the discussion above to $s=0$, we have $\mathbf P_{\Phi_N}(L^2(\Lambda))=1$, and therefore $\mathbf P_{\Phi_N}$ restricts to a~Borel probability measure on $L^2(\Lambda)$.
\end{proof}

We highlight that Lemma \ref{lem:Phi_regularity} is true for the probability distribution constructed from upper symbols, because the covariance of upper symbols decays with $p$. This property is not shared by the Weyl-Wigner and lower symbols, and $\Phi_N(x)$ would be a singular distribution if we chose $(\widehat \Phi_N(p),\widehat \Phi_N(-p))$ distributed according to $\Sigma^W_{\pm p}$ or $\Sigma^\downarrow_{\pm p}$. 

\begin{proof}[Proof of Theorem \ref{thm:semiclassicalApproximation} (a)]
    We invoke the tensor product decomposition \eqref{eq:G_tensor_decomp} of $\widetilde G^\Bog$ and the upper symbol representation \eqref{eq:upper_symbol},
    \begin{equation}
        \widetilde G^\Bog = \bigotimes_{p \in A} \int_{\mathbb C^2} | \Omega_{\phi_p , \phi_{-p}} \rangle \langle \Omega_{\phi_p ,\phi_{-p}} | \Sigma^\uparrow_{\pm p}(\phi_p , \phi_{-p}) \d \phi_p \d \phi_{-p} = \int | \Omega_\phi \rangle \langle  \Omega_\phi | \d \mathbf P_{\Phi_N} (\phi),
    \end{equation}
    where we used that $\Omega_\phi$ is a product vector, $\Omega_\phi = \bigotimes_{p \in A} \Omega_{\phi_p,\phi_{-p}}$. We note that $\Omega_\phi$ is a well-defined element of $\mathcal F_+$, because by Lemma \ref{lem:Phi_regularity} we can assume that $\phi \in L^2(\Lambda)$. Next, we observe that
    \begin{equation}
        \widetilde G^\Bog (z) = \left( \frac{z}{|z|} \right)^{\mathcal N_+} \widetilde G^\Bog \left( \frac{z}{|z|} \right)^{- \mathcal N_+} = \int | \Omega_{\frac{z}{|z|} \phi} \rangle \langle  \Omega_{\frac{z}{|z|} \phi} | \d \mathbf P_{\Phi_N} (\phi),
    \end{equation}
    and therefore,
    \begin{align}
        \widetilde \Gamma_{\beta, N} &= \int | z \rangle \langle z | \otimes | \Omega_{\frac{z}{|z|} \phi} \rangle \langle  \Omega_{\frac{z}{|z|} \phi} | \widetilde g^{\BEC}(z) \d \mathbf P_{\Phi_N} (\phi) dz  \\
        &= \int | \Omega_{z + \frac{z}{|z|} \phi} \rangle \langle  \Omega_{z + \frac{z}{|z|} \phi} | \widetilde g^{\BEC}(z) \d \mathbf P_{\Phi_N} (\phi) dz. \nonumber
    \end{align}
    Assuming that $N g^\id \geq c N^{\frac23}$ and taking $p_c(N) \geq c N^{\frac{1}{144}}$ (which is compatible with the requirement $p_c(N) \leq c N^{\frac13}$), the claim follows from  \eqref{eq:norm-approximation-intro} and Proposition \ref{prop:traceNormDifference}.
    \end{proof}
    
Let us explain how the integral representation of the approximate Gibbs state $\widetilde \Gamma_{\beta, N}$ derived above will be used in the forthcoming computations. If $T \in B(\mathcal F)$, then
\begin{align}
  \mathrm{Tr} [T \, G_{\beta,N}] + O(N^{-\frac{1}{48}}) &=   \mathrm{Tr} [T \, \widetilde \Gamma_{\beta,N}] \\
  &= \int \langle \Omega_{\phi_z} | T \Omega_{\phi_z} \rangle \widetilde g^{\mathrm{BEC}}(z) \d \mathbf P_{\Phi_N}(\phi) dz = \mathbf{E}[ \langle \Omega_{\Phi_{N,Z}} | T \Omega_{\Phi_{N,Z} } \rangle ]. \nonumber
\end{align}
The resulting expression is the expectation value (with respect to the randomness governing the field $\Phi_{N,Z}(x) = Z + \frac{Z}{|Z|} \Phi_N(x)$) of the lower symbol of $T$, that is of the matrix element $\langle \Omega_{\Phi_{N,Z}} | T \Omega_{\Phi_{N,Z} } \rangle$ with a~random coherent state $\Omega_{\Phi_{N,Z}}$. Here the distribution of $|Z|$ is governed by the function $\widetilde g^{\BEC}(z)$, and~in particular $\frac{Z}{|Z|}$ is distributed uniformly on the unit circle. As we will discuss, the covariance of $\Phi_N(x)$ simplifies in the limit $N \to \infty$, in which $\Phi_N(x)$ can be replaced with an asymptotic field $\Psi$.

Let us suppose that $\beta(N)$ is chosen in such a way that the limit $m^2 = - \lim_{N \to \infty} \mu^\id$ exists and is finite. We recall that this holds in particular if $\kappa <1$, in which case $m^2 = 0$ (see Proposition \ref{prop:limits} and the discussion below it). We define $\Sigma^\infty_{\pm p}(\psi_p, \psi_{-p}) \d \psi_p \d \psi_{-p}$ to be the probability distribution of the centered complex Gaussian random vector $(\widehat \Psi(p), \widehat \Psi(-p))$ with covariance and pseudo-covariance
\begin{align}
	 & \mathbf{E} \begin{pmatrix}
		|\widehat \Psi(p)|^2 & \widehat \Psi(p) \overline{\widehat \Psi(-p)} \\ \overline{\widehat \Psi(p)} \widehat \Psi(-p) & |\widehat \Psi(-p)|^2 		\end{pmatrix} = \frac{p^2 +m^2 + g(\kappa) \widehat v(p)}{(p^2 +m^2) (p^2 +m^2 + 2 g(\kappa) \widehat v(p))}  \begin{pmatrix}
	1 & 0 \\ 0 & 1 
	\end{pmatrix},  \nonumber \\
	 & \mathbf{E} \begin{pmatrix}
	\widehat \Psi(p)^2 & \widehat \Psi(p) \widehat \Psi(-p) \\ \widehat \Psi(p) \widehat \Psi(-p) & \widehat \Psi(-p)^2 		
	\end{pmatrix} = -\frac{g(\kappa)\widehat v(p)}{(p^2+m^2) (p^2+m^2 + 2g(\kappa) \widehat v(p))}  \begin{pmatrix}
	0 & 1 \\ 1 & 0 
	\end{pmatrix}.
\label{eq:familyOfGaussianRandomVariables_with_mu}
\end{align}
More explicitly, the density $\Sigma^\infty(\psi_p,\psi_{-p})$ is given by
\begin{align}
    \Sigma^\infty_{\pm p}(\psi_p, \psi_{-p}) = & (p^2 +m^2)(p^2 +m^2 + 2 g(\kappa) \widehat v(p)) \label{eq:Sigma_infty_explicit} \\
    &\cdot \exp \left( - (p^2 +m^2 +2 g(\kappa) \widehat v(p)) (|\psi_p|^2 + |\psi_{-p}|^2) -2 g(\kappa) \widehat v(p) \, \mathrm{Re}(\psi_p \psi_{-p})\right). \nonumber
\end{align}

Definition \eqref{eq:familyOfGaussianRandomVariables_with_mu} generalizes \eqref{eq:familyOfGaussianRandomVariables}, which was formulated for $\kappa >1$, and therefore $m^2 =0$. While it depends on two parameters $g(\kappa)$ and $m^2$, these are never simultaneously nonzero. This convention allows for a unified treatment of the case $\kappa >1$ and the case $\kappa=1$ with finite $m^2$. Our proofs of Theorems \ref{thm:semiclassicalApproximation} and \ref{thm:thermaldistr} cover $\kappa=1$ with finite $m^2$, but we omit this case from explicit Theorem statements to keep them accessible. 

We remark that each of the semiclassical symbols $\Sigma^{\bullet}_{\pm p}$, suitably rescaled, converges to $\Sigma^\infty_{\pm p}$:
\begin{equation}
\beta^{-2}    \Sigma^\bullet_{\pm p}(\beta^{-\frac12} f_p , \beta^{-\frac12} f_{-p})  \xrightarrow{N \to \infty} \Sigma^\infty_{\pm p}(f_p,f_{-p}), \qquad  \bullet \in \{ W, \uparrow, \downarrow \}.
\end{equation}
As explained above, we prefer to use the upper symbol, and its specific modes of convergence will be discussed in detail below.

In complete analogy with \eqref{eq:random_Phi}, we introduce the random Fourier series
\begin{equation}
    \Psi(x) = \sum_{p \in \Lambda^*_+} \widehat \Psi(p) e^{\mathrm{i}px},
\end{equation}
where $(\widehat \Psi(p), \widehat \Psi(-p))$ is distributed according to $\Sigma_{\pm p}^\infty$ and independent of $(\widehat \Psi (q), \Psi_{-q})$ for $q \neq \pm p$. We~denote the probability distribution of $\Psi$ by $\mathbf P_{\Psi}$. 

The random function $\Psi(x)$ is much less regular than $\Phi_N(x)$ because its covariance decays with $p$ only as fast as $\frac{1}{p^2}$, like for a Gaussian free field. Verification of the lemma below is routine and follows similar lines as the proof of Lemma \ref{lem:Phi_regularity}.

\begin{lem}\label{lem:Psi_reg}
$\Psi(x)$ is almost surely in the Sobolev space $H^s(\Lambda)$ for $s < - \frac12$, and almost surely not in $H^s(\Lambda)$ for $s \geq - \frac12$. The distribution $\mathbf P_\Psi$ restricts to a Borel probability measure on $H^s(\Lambda)$ for $s < - \frac12$.  
\end{lem}

To establish Theorem 1 (b), we construct a specific coupling of $\Phi_N$ and $\Psi$ --- that is, a joint distribution compatible with the prescribed marginals. This coupling will also feature in subsequent proofs.

\begin{proof}[Proof of Theorem \ref{thm:semiclassicalApproximation} (b)]
Let $(\eta_{p, \pm})_{p \in A}$ be independent centered complex Gaussian random variables satisfying $\mathbf E[|\eta_p|^2] = 1$ and $\mathbf E[\eta_p^2] = 0$. We can realize the random variable $\Phi_N$ on the probability space of $\eta$ via, for every $p \in A$,
\begin{subequations}
\label{eq:Phi_from_eta}
\begin{align}
 \widehat \Phi_N(p) & = \frac{1}{\sqrt{2}} \left( \sqrt{\widetilde \gamma^\Bog(p) + \widetilde \alpha^\Bog(p)} \, \eta_{p,+} + \sqrt{\widetilde \gamma^\Bog(p) - \widetilde \alpha^\Bog(p)} \, \eta_{p,-} \right)  , \\
 \widehat \Phi_N(-p) &=  \frac{1}{\sqrt{2}} \left( \sqrt{\widetilde \gamma^\Bog(p) + \widetilde \alpha^\Bog(p)} \, \overline{\eta_{p,+}} - \sqrt{\widetilde \gamma^\Bog(p) - \widetilde \alpha^\Bog(p)} \, \overline{\eta_{p,-}} \right).
\end{align}
\end{subequations}
Then $(\widehat \Phi_N(p),\widehat \Phi_N(-p))$ are Gaussian, and one checks that \eqref{eq:upper_symbol_covariance} holds. Similarly, we can also generate $\Psi$ from the same random seed $\eta$ as follows:
\begin{subequations}
\begin{align}
 \widehat \Psi(p) & = \frac{1}{\sqrt{2}} \left( \frac{1}{\sqrt{p^2 +m^2 + 2 g(\kappa) \widehat v(p)}} \, \eta_{p,+} + \frac{1}{\sqrt{p^2 +m^2}} \, \eta_{p,-} \right) \label{eq:Psi_eta} ,  \\
 \widehat \Psi(-p) &=  \frac{1}{\sqrt{2}} \left( \frac{1}{\sqrt{p^2 +m^2 + 2 g(\kappa) \widehat v(p)}} \, \overline{\eta_{p,+}} - \frac{1}{\sqrt{p^2 +m^2}} \, \overline{\eta_{p,-}} \right). \nonumber
\end{align}
\end{subequations}
To measure the squared distance between $(\widehat \Phi_N(p),\widehat \Phi_N(-p))$ and $(\widehat \Psi(p),\widehat \Psi(-p))$, we compute 
\begin{align}
   \sum_\pm |\beta^{\frac12} \widehat \Phi_N(\pm p) - \widehat \Psi(\pm p)|^2 = &\left| \beta^{\frac12} \sqrt{\widetilde \gamma^\Bog(p) + \widetilde \alpha^\Bog(p)} - \frac{1}{\sqrt{p^2 +m^2+ 2 g(\kappa) \widehat v(p)}} \right|^2  |\eta_{p,+}|^2 \label{eq:Phi_Psi_dist} \\
     +  &\left| \beta^{\frac12} \sqrt{\widetilde \gamma^\Bog(p) - \widetilde \alpha^\Bog(p)} - \frac{1}{\sqrt{p^2 +m^2 }} \right|^2   |\eta_{p,-}|^2  .   \nonumber
\end{align}
Let us inspect the coefficients of the right hand side above:
\begin{align}
    & \left| \beta^{\frac12} \sqrt{\widetilde \gamma^\Bog(p) + \widetilde \alpha^\Bog(p)}  - \frac{1}{\sqrt{p^2 +m^2 + 2 g(\kappa) \widehat v(p)}} \right|  = \frac{\left| \beta (\widetilde \gamma^\Bog(p) + \widetilde \alpha^\Bog(p)) - \frac{1}{p^2 +m^2 + 2 g(\kappa) \widehat v(p)} \right|}{\left| \beta^{\frac12} \sqrt{\widetilde \gamma^\Bog(p) + \widetilde \alpha^\Bog(p)} + \frac{1}{\sqrt{p^2 +m^2 + 2 g(\kappa) \widehat v(p)}} \right|} \nonumber \\
     & \leq  \frac{1}{|p|} | \beta (\widetilde \gamma^\Bog(p) + \widetilde \alpha^\Bog(p)) (p^2 +m^2 + 2 g(\kappa) \widehat v(p)) -1 |   \label{eq:formula_68}  . 
\end{align}
Let us now note that $
    \widetilde \gamma^\Bog(p) + \widetilde \alpha^\Bog(p)  = \frac{p^2 - \mu^\id}{\epsilon(p)} \frac{1}{e^{\beta \epsilon(p)}-1} + O \left( \frac{\widehat v(p)}{p^2} \right)$, and therefore:  
\begin{align}
   & | \beta (\widetilde \gamma^\Bog(p) + \widetilde \alpha^\Bog(p)) (p^2 +m^2 + 2 g(\kappa) \widehat v(p)) -1 | 
 \leq  c \beta \widehat v(p) + \left| \frac{\beta \epsilon(p)}{e^{\beta \epsilon(p)}-1} \frac{p^2 +m^2 + 2 g(\kappa) \widehat v(p)}{p^2 - \mu^\id + 2 g^\id \widehat u(p)} -1 \right|  \nonumber  \\
& \leq c \beta \widehat v(p) + c \left| \frac{\beta \epsilon(p)}{e^{\beta \epsilon(p)}-1} -1 \right| + c \frac{|\mu^\id +m^2| + |g^\id \widehat u(p) - g(\kappa) \widehat v(p)|}{p^2} \label{eq:formula_73} \\
& \leq  c \min \{ 1 , \beta p^2 \} + c \frac{| \mu^\id +m^2| + |g -g^\id| + p_c^{-1}}{p^2} \leq c (\beta p^2)^{\frac{\delta}{2}} +  c \frac{| \mu^\id +m^2| + |g -g^\id| + p_c^{-1}}{p^2},  \nonumber
\end{align}
where $\delta \in (0,2]$. In the third inequality we used the assumed summability of $p \widehat v(p)$ to bound $|\widehat u(p) - \widehat v(p)| \lesssim p_c^{-1}$. We assume that $p_c \to \infty$, so $| \mu^\id +m^2| + |g -g^\id| + p_c^{-1} \leq c_N = o(1)$, see Proposition \ref{prop:limits}. We can treat the second term in \eqref{eq:Phi_Psi_dist} the same way, obtaining:
\begin{equation}
       |\beta^{\frac12} \widehat \Phi_N(p) - \widehat \Psi(p)|^2 + |\beta^{\frac12} \widehat \Phi_N(-p) - \widehat \Psi(-p)|^2 \leq c_N (p^2)^{-1+\delta}  (|\eta_{p,+}|^2 + |\eta_{p,-}|^2) .
    \label{eq:sum_of_squares_bound}
\end{equation}
Therefore, for $s < -\frac12 - \delta$ and $t>0$ small enough,
\begin{align}
  \mathbf E[  e^{t c_N^{-1} \| \beta^{\frac12} \Phi_N - \Psi \|_{H^s(\Lambda)}^2} ]  \leq \prod_{p \in A} \mathbf E[e^{  c t  (p^2)^{s -1+\delta} (|\eta_{p,+}|^2 + |\eta_{p,-}|^2) }]  = \prod_{p \in A} \frac{1}{1-c t (p^2)^{s - 1 +\delta}} \lesssim 1 , \nonumber
\end{align}
where the equality holds by explicit Gaussian integration, and in the last inequality we used the convergence of the series $\sum_{p \in A} (p^2)^{s -1 + \delta} < \infty$. In particular, $\mathbf E[\| \beta^{\frac12} \Phi_N - \Psi \|_{H^s(\Lambda)}^{2k}] \lesssim k! \left(  \frac{c_N}{t} \right)^k$. 
\end{proof}

The proof above establishes the convergence $\beta^{\frac12} \Phi_N \to \Psi$ in a setting slightly more general than stated in Theorem \ref{thm:semiclassicalApproximation}, namely including $\kappa=1$ with finite $m^2$. If we strengthen the assumptions slightly, we can get an explicit rate of convergence.

\begin{cor} \label{cor:1b_quantitative}
    Suppose that $\beta = \kappa \beta_c (1+O(N^{-\frac13}))$ with $\kappa>1$, and take $p_c = N^{ \frac13}$. Then for $\delta \in (0,1]$, $s < -\frac12 - \delta$, and $t$ small enough:
    \begin{equation}
     \mathbf E[  e^{t \beta^{-\delta} \| \beta^{\frac12} \Phi_N - \Psi \|_{H^s(\Lambda)}^2} ] \lesssim 1.
    \end{equation}
\end{cor}
\begin{proof}
If $\kappa>1$, we have $m^2 =0$ and $\mu^\id = O(N^{-\frac13})$. The assumed rate of convergence of $\frac{\beta}{\beta_c}$ implies that also $|g-g^\id|$ is $O(N^{-\frac13})$ by \eqref{eq:particle_partition}. Therefore, we can take $c_N = \beta^{\delta}$ in \eqref{eq:sum_of_squares_bound}. 
\end{proof}

Finally, we slightly shift gears to prove the last point in Theorem \ref{thm:semiclassicalApproximation}.

\begin{proof}[Proof of Theorem \ref{thm:semiclassicalApproximation} (c)] Let $\widetilde \Theta_{\beta,N}$ be the  point process associated to $\widetilde \Gamma_{\beta,N}$. To verify that this is a Cox point process, by  \cite[Prop. 6.2.II]{DaleyVere-Jones2003} it is enough to compute the characteristic function: 
\begin{align}
    \mathbf E[ e^{\mathrm{i} \widetilde \Theta_{\beta,N}(f)}] = \mathrm{Tr}[\widetilde \Gamma_{\beta,N} \, e^{\mathrm{i} \mathrm{d} \Gamma(f)}] &= \mathbf E[\langle \Omega_{\Phi_{N,z}} | e^{\mathrm{i} \mathrm{d} \Gamma(f) } \Omega_{\Phi_{N,z}} \rangle ] \\
    &= \mathbf E[ \exp \left( \langle \Phi_{N,z} | (e^{\mathrm{i}f}-1) \Phi_{N,z} \rangle \right)] = \mathbf E[e^{\int_\Lambda (e^{\mathrm{i}f(x)}-1) \d \lambda(x)} ], \nonumber
\end{align}
where $\d \lambda(x) = |\Phi_{N,z}(x)|^2 \d x = | |Z| + \Phi_N(x) |^2 \d x$. 

The statement then follows from the fact that
for a bounded Borel-measurable function $F : \mathbb{R} \to \mathbb{C}$ we have (cf. \eqref{eq:prop:empiricalMeasure1})
    \begin{equation*}
        \mathbf{E}[ F(\Theta_{\beta,N}(f)) ] = \Tr[ F(\d \Gamma(f) ) G_{\beta,N}],
\end{equation*}
and from part (a) of the Theorem.
\end{proof}

\section{Proof of Theorem \ref{thm:thermaldistr}}\label{sec:thermaldistr}
Let us now discuss linear functions of $\Psi$. First consider $f \in H^s(\Lambda)$ with $s > \frac12$. Then, by duality of $H^s(\Lambda)$ with $H^{-s}(\Lambda)$ and Lemma \ref{lem:Psi_reg}, the linear form
\begin{equation}
    \langle f | \Psi \rangle = \sum_{p \in \Lambda^*_+} \overline f_p \widehat \Psi(p) 
    \label{eq:linear_form}
\end{equation}
is defined. Indeed, the series in \eqref{eq:linear_form} converges absolutely $\mathbf P_\Psi$-almost surely. In fact, $\langle f | \Psi \rangle$ is a~well-defined random variable for a much larger class of distributions $f$.  

\begin{lem} \label{lem:linear_observables_def}
The map $f \mapsto \mathrm{Re} \, \langle f | \Psi \rangle $ extends to an $\mathbb R$-linear operator $H^{-1}(\Lambda) \to L^2(d \mathbf P_\Psi,\mathbb R)$ which is a homeomorphism onto its image. 
\end{lem}
\begin{proof}
    By density of $C^\infty(\Lambda)$ in $H^{-1}(\Lambda)$, it suffices to observe that
    \begin{align}
        \| \mathrm{Re} \, \langle f | \Psi \rangle \|_{L^2(\mathrm{d} \mathbf P_\Psi)}^2 &= \mathbf E[ ( \mathrm{Re} \, \langle f | \Psi \rangle )^2 ] \\
        &= \frac12 \sum_{p \in A} \left( (|f_p|^2 + |f_{-p}|^2) \mathbf E[|\widehat \Psi(p)|^2] + (f_p f_{-p} + \overline{f_p f_{-p}}) \mathbf E[\widehat \Psi(p) \widehat \Psi(-p)]  \right) \nonumber
    \end{align}
    can be upper-bounded and lower-bounded by $c \| f \|_{H^{-1}(\Lambda)}^2$.
\end{proof}

\begin{lem} \label{lem:gaussian_TV}
    Let $g_{\Sigma_1}, g_{\Sigma_2}$ be centered Gaussian probability densities on $\mathbb R^n$ with covariance matrices $\Sigma_1, \Sigma_2 >0$. Let $\Delta = \Sigma_1^{-\frac12} \Sigma_2 \Sigma_1^{-\frac12} -1$, and suppose that $\Delta+1 \geq c $ with $c \in (0,\infty)$. Then:
    \begin{equation}
        \| g_{\Sigma_1} - g_{\Sigma_2} \|_{L^1(\mathbb R^n)} \leq \sqrt{\frac{1}{8c} \mathrm{tr}(\Delta^2)}.
    \end{equation}
\end{lem}
\begin{proof}
   We compute the relative entropy (Kullback-Leibler divergence):
\begin{align}
    \mathcal H(g_{\Sigma_2} , g_{\Sigma_1}) &= \int_{\mathbb R^n} g_{\Sigma_2}(x) \log \frac{g_{\Sigma_2}(x)}{g_{\Sigma_1}(x)} \d x = \frac{1}{2} \log \frac{\det(\Sigma_1)}{\det(\Sigma_2)} + \frac12 \mathrm{tr}(\Sigma_1^{-1} \Sigma_2 -1) \\
  & =  - \frac12 \log \det (\Sigma_1^{-\frac12} \Sigma_2 \Sigma_1^{-\frac12}) + \frac12 \mathrm{tr}(\Sigma_1^{-\frac12} \Sigma_2 \Sigma_1^{-\frac12} -1) = \frac12 \mathrm{tr}(f(\Delta)), \nonumber
\end{align}
where $f(x) = - \log(1+x) +x $. If $x+1 \geq c$, then:
\begin{equation}
    f(x) = \int_0^x \frac{t}{1+t} \d t \leq \frac{1}{c} \int_0^x t \d t = \frac{t^2}{2c},
\end{equation}
and therefore $\mathcal H(g_{\Sigma_2} , g_{\Sigma_1}) \leq \frac{1}{4c} \mathrm{tr}(\Delta^2)$ holds if $\Delta +1 \geq c$. The result follows from Pinsker's inequality.
\end{proof}

\begin{prop} \label{prop:linear_observables}
    Let $f_1,\dots,f_n \in H^{-1}(\Lambda)$. Consider the random variables 
    \begin{equation}
        X_i = \mathrm{Re} \, \langle f_i | \beta^{\frac12} \Phi_N \rangle, \qquad Y_i = \mathrm{Re} \, \langle f_i | \Psi \rangle.
    \end{equation}
    \begin{enumerate}[label=(\alph*)]
    \item As $N \to \infty$, we have convergence $(X_1,\dots,X_n) \to (Y_1,\dots,Y_n)$ in total variation. If $\kappa>1$, $\beta = \kappa \beta_c (1 + O(N^{-\frac13}))$, and $f_1,\dots,f_n \in H^{-1 + \delta}(\Lambda)$ with $\delta \in (0,\frac12]$, we have an explicit rate: 
    \begin{equation}
        d_{\mathrm{TV}}((X_1,\dots,X_n), (Y_1,\dots,Y_n)) \lesssim \beta^\delta. 
    \end{equation}
    \item $(X_1,\dots,X_n)$ and $(Y_1,\dots,Y_n)$ can be realized as random variables on a common probability space so that $\lim_{N \to \infty} \mathbf{E}[ \sum_{i=1}^n |X_i - Y_i|^p]=0$ for all $p \in [1,\infty)$. If $\kappa>1$, $\beta = \kappa \beta_c (1 + O(N^{-\frac13}))$, and $f_1,\dots,f_n \in H^{-1 + \delta}(\Lambda)$ with $\delta \in (0,\tfrac12]$, then for some $t>0$ we have:
    \begin{equation}
        \mathbf E[e^{ t \beta^{-2 \delta} \sum_{i=1}^n |X_i - Y_i|^2 }] \lesssim 1.
    \end{equation}
    \end{enumerate}
    
\end{prop}
\begin{proof}
In both parts of the proof it is enough to consider the case of $f_i$ linearly independent over $\mathbb R$.

  (a)  Let $\Sigma_2,\Sigma_1$ be the covariance matrices of $(X_1,\dots,X_n)$ and $(Y_1,\dots,Y_n)$, respectively:
    \begin{equation}
        (\Sigma_2)_{ij} = \mathbf E[X_i X_j], \qquad (\Sigma_1)_{ij} = \mathbf E[Y_i Y_j].   
        \label{eq:Sigma_1_Sigma_2}
    \end{equation}
    For $c\in \mathbb{R}^n$ we have 
    $$c^T \Sigma_1 c = \sum_{i=1}^n \sum_{j=1}^n c_i \mathbf{E}[Y_i Y_j] c_j = \mathbf{E}\left[ \left( \sum_{i=1}^n c_i Y_i \right)^2 \right] = \left\| \sum_{i=1}^n c_i Y_i \right\|_{L^2(d \mathbf P_\Psi)}^2.$$
    Thus, by the linear independence of $f_i$ and Lemma \ref{lem:linear_observables_def}, the matrix $\Sigma_1$ is positive and invertible. Below we will check that $\Sigma_2$ converges to $\Sigma_1$ in the limit $N \to \infty$. Denoting $\Delta = \Sigma_1^{-\frac12} \Sigma_2 \Sigma_{1}^{-\frac12}-1$, we have $\Delta \xrightarrow{N \to \infty} 0$. Therefore, $\Delta+1 \geq c $ holds for some $c > 0$ and all $N$ large enough. By Lemma \ref{lem:gaussian_TV}, 
    \begin{align}
        d_{\mathrm{TV}}((X_1,\dots,X_n),(Y_1,\dots,Y_n)) &\leq \frac{1}{\sqrt{8c}} \sqrt{\mathrm{tr}(\Delta^2)} \\
        & = \frac{1}{\sqrt{8 c}} \sqrt{\mathrm{tr}( (\Sigma_1^{-1} (\Sigma_2 - \Sigma_1))^2)} \leq \sqrt{\frac{\| \Sigma_1 ^{-1} \|}{8c}} \sqrt{\mathrm{tr}((\Sigma_2-\Sigma_1)^2)}. \nonumber
    \end{align}
 Since we are not keeping track of $n$-dependent constants here, it is enough to establish the claimed convergence for matrix elements of $\Sigma_1 - \Sigma_2$.
    One verifies that
    \begin{align}
    |(\Sigma_1 - \Sigma_2)_{ij}| \leq c \sum_{p \in A} & \left( \left| \mathbf E[|\widehat \Psi(p)|^2 - \beta |\widehat \Phi_N(p)|^2] \right| + \left| \mathbf E[\widehat \Psi(p) \widehat \Psi(-p)- \beta \widehat \Phi_N(p) \widehat \Phi_N(-p)] \right| \right) \\
    & \cdot \sqrt{|f_{i,p}|^2 + |f_{i,-p}|^2} \sqrt{|f_{j,p}|^2 + |f_{j,-p}|^2} . \nonumber 
    \end{align}
    The coefficient $\left| \mathbf E[|\widehat \Psi(p)|^2 - \beta |\widehat \Phi_N(p)|^2] \right| + \left| \mathbf E[\widehat \Psi(p) \widehat \Psi(-p)- \beta \widehat \Phi_N(p) \widehat \Phi_N(-p)] \right|$ converges to zero for every $p$, and is bounded by $\frac{c}{p^2}$, so the claimed convergence follows from the dominated convergence theorem. The quantitative statement is obtained by bounding:
    \begin{align}
       & \left| \mathbf E[|\widehat \Psi(p)|^2 - \beta |\widehat \Phi_N(p)|^2] \right| + \left| \mathbf E[\widehat \Psi(p) \widehat \Psi(-p)- \beta \widehat \Phi_N(p) \widehat \Phi_N(-p)] \right|  \label{eq:boundcp}\\
     &  \leq \sum_{\pm} \left| \mathbf E[(|\widehat \Psi(p)|^2 \pm \widehat \Psi(p) \widehat \Psi(-p)) - \beta (|\widehat \Phi_N(p)|^2 \pm \widehat \Phi_N(p) \widehat \Phi_N(-p))] \right| \nonumber \\
       & = \left| \frac{1}{p^2 +m^2 + 2 g(\kappa) \widehat v(p)} - \beta (\widetilde \gamma^\Bog(p) + \widetilde \alpha^\Bog(p)) \right| + \left| \frac{1}{p^2 +m^2} - \beta (\widetilde \gamma^\Bog(p) - \widetilde \alpha^\Bog(p)) \right|  \nonumber \\
       & \leq \frac{c}{p^2} \left( (\beta p^2)^\delta + \frac{| \mu^\id + m^2| + |g(\kappa) - g^\id| +p_c^{-1}}{p^2} \right), \nonumber
    \end{align}
    where in the final step we used \eqref{eq:formula_73} for the first term and the analogous estimate for the second term. 
    
    (b) Let $\eta_1,\dots,\eta_n$ be independent standard Gaussian random variables. We can realize $X_i$ and $Y_i$ on the probability space of $(\eta_1,\dots,\eta_n)$ as follows:
    \begin{equation}
        X_i = \sum_j (\Sigma_2^{\frac12})_{ij} \eta_j, \qquad Y_i = \sum_j (\Sigma_1^{\frac12})_{ij} \eta_j,
    \end{equation}
    where $\Sigma_1,\Sigma_2$ are as in \eqref{eq:Sigma_1_Sigma_2}. Then $\sum_i |X_i- Y_i|^2 = \sum_{jk} M_{jk} \eta_j \eta_k$, where $M = (\Sigma_1^{\frac12} - \Sigma_2^{\frac12})^{2}$. We have $\Sigma_2 \to \Sigma_1$ in the limit $N \to \infty$, so all moments of $\sum_i |X_i- Y_i|^2 = \sum_{jk} M_{jk} \eta_j \eta_k$ converge to zero. To get a quantitative bound, we need to control the convergence $M \to 0$; for that it suffices to observe that $\Sigma_2$ is strictly positive, so the square root function is Lipschitz on some neighborhood of $\Sigma_2$. Therefore, for large enough $N$ we have $\| M \| \leq c \| \Sigma_1 - \Sigma_2 \|^2$. We can now invoke estimates of $\Sigma_1 - \Sigma_2$ presented in the proof of (a).
\end{proof}

The next Proposition describes fluctuations of the $L^2$ norm $\| \Phi_N \|^2$, asymptotically for $N \to \infty$, in~terms of the random Fourier series $:|\Psi(x)|^2:$ introduced in \eqref{eq:WickProduct}. Let us highlight that fluctuations are of order $N^{\frac23}$, while $\| \Phi_N \|^2$ is of order $N$. 

\begin{prop} \label{prop:semiclassics_Wick_square}
For every $q \in [1,\infty)$, we have convergence in Wasserstein $q$-distance as $N \to \infty$:
\begin{equation}
    \beta \| \Phi_N \|^2 - \mathbf E[\beta \| \Phi_N \|^2 ] \longrightarrow \int_\Lambda :|\Psi(x)|^2: \d x.
\end{equation}
If $\kappa > 1$ and $\beta = \kappa \beta_c (1+O(N^{-\frac16}))$, we have convergence with an explicit rate $\beta^\delta$ for every $\delta < \frac{1}{2}$.
\end{prop}
We remark that this statement is generalized by the convergence of random variables
\begin{equation}
    \beta \int_\Lambda (|\Phi_N(x)|^2 - \mathbf E[|\Phi_N(x)|^2]) f(x) \d x \longrightarrow   \int_\Lambda :|\Psi(x)|^2: f(x) \d x,
\end{equation}
which can be shown to hold for instance for every $f \in H^{-\frac12}(\Lambda)$. Here we present the proof only for $f=1$, which is the only case that will be needed to prove our main results. 
\begin{proof}
    We note that
    \begin{subequations}
    \begin{align}
        \| \Phi_N \|^2 - \mathbf{E}[\| \Phi_N \|^2] &= \sum_{p \in \Lambda_+^*}  (|\widehat \Phi_N(p)|^2  - \mathbf{E}[|\widehat \Phi_N(p)|^2] ), \\
        \int :|\Psi(x)|^2: \d x &= \sum_{p \in \Lambda^*_+} (|\widehat \Psi(p)|^2 - \mathbf{E}[|\widehat \Psi(p)|^2]).
    \end{align}
    \end{subequations}
Constructing $\Phi_N$ and $\Psi$ from random variables $\eta_{p, \pm}$ as in the proof of Theorem \ref{thm:semiclassicalApproximation} (b), we find
\begin{equation}
  Z \coloneq  \beta \| \Phi_N \|^2 - \mathbf{E}[ \beta \| \Phi_N \|^2] - \int :|\Psi(x)|^2: \d x = \sum_{p \in A} \left( c_{p,+} (|\eta_{p,+}|^2-1) + c_{p,-} (|\eta_{p,-}|^2-1) \right),
\end{equation}
where the coefficients $c_{p,\pm}$ are
\begin{subequations}
\begin{align}
    c_{p,+ } & =   \beta (\widetilde \gamma^\Bog(p) + \widetilde \alpha^\Bog(p))-\frac{1}{p^2 +m^2 + 2 g(\kappa) \widehat v(p)}, \\
    c_{p,-} & =  \beta (\widetilde \gamma^\Bog(p) -\widetilde \alpha^\Bog(p))-\frac{1}{p^2 +m^2}.
\end{align}
\end{subequations}
We have $|c_{p,\pm}| \leq \frac{c_N}{p^2}$ for some positive $c_N = o(1)$. Therefore, for small enough $t \in \mathbb R$,
\begin{equation}
    \mathbf E[e^{t c_N^{-1} Z}] = \prod_{p \in A} \prod_{\pm} \mathbf E[e^{tc_N^{-1} c_{p,\pm} (|\eta_{p,\pm}|^2-1)}] =  \prod_{p \in A} \prod_{\pm} \frac{e^{-t c_N^{-1} c_{p,\pm}}}{1-tc_N^{-1} c_{p,\pm}} \leq e^{c \sum_{p \in A}  \frac{1}{p^4}} \lesssim 1,
\end{equation}
where in the second last inequality we have used that 
$$\ln \left( \frac{e^{-x}}{1-x} \right) \le C x^2  $$
for $x$ sufficiently small. Therefore, for even integer $q$ we find
\begin{equation}
    \mathbf E[Z^q] \leq q! (tc_N^{-1})^{-q} \frac12 \mathbf E[ e^{tc_N^{-1} Z} + e^{-tc_N^{-1} Z}] \lesssim c_N^q
\end{equation}
where we used that by the Taylor series expansion for $\cosh x$ one has
$$\frac{1}{2} (e^x + e^{-x}) \ge \frac{x^q}{q!} \implies x^q \le q! \cdot \frac{1}{2} (e^x + e^{-x}).$$
To get more quantitative control, we can estimate $c_{p, \pm}$ as in \eqref{eq:formula_73} and \eqref{eq:boundcp}, where we need to choose $\delta$ such that $\sum_{p \in A} \frac{1}{p^4} (\beta p^2)^\delta < \infty$. 
\end{proof}

We are now ready to prove Theorem \ref{thm:thermaldistr}.

\begin{proof}[Proof of Theorem \ref{thm:thermaldistr} (a)]
Let us define the random vector $(\widetilde A_{f_1}, \dots, \widetilde A_{f_n})$ by
\begin{equation}
   \mathbf{P}(\widetilde A_{f_j} \in B_j \text{ for } j=1,\dots,n) = \Tr\!\left[ \prod_{j=1}^n \mathds{1}\big( a(f_j) + a^*(f_j) \in B_j \big)\, \widetilde \Gamma_{\beta, N} \right].
\end{equation}
Then $d_{\mathrm{TV}}((A_{f_1}, \dots,A_{f_n}),(\widetilde A_{f_1}, \dots , \widetilde A_{f_n})) \leq \| G_{\beta , N} - \widetilde \Gamma_{\beta,N} \|_1 \lesssim N^{- \frac{1}{48}}$ by Theorem \ref{thm:semiclassicalApproximation}. To find the distribution of $(\widetilde A_{f_1}, \dots, \widetilde A_{f_n})$, we compute the characteristic function: 
\begin{align}
    \mathbf E[e^{\mathrm{i} \sum_j t_j \widetilde A_{f_j}}]    &= \mathrm{Tr}[ e^{\mathrm{i} \sum_j t_j (a^*(f_j) +a(f_j))} \widetilde \Gamma_{\beta, N} ]  \\
   & = \mathbf E [ \langle \Omega_{\Phi_{N,z}} | e^{\mathrm{i} \sum_j t_j  (a^*(f_j) +a(f_j))} \Omega_{\Phi_{N,z}} \rangle ]  
    = e^{- \frac{1}{2} \sum_j t_j^2} \int_{-\pi}^\pi \mathbf E[e^{2 \mathrm{i} \sum_j  \mathrm{Re} (e^{\mathrm{i} \theta} \langle f_j | \Phi_N \rangle})] \frac{\mathrm{d} \theta}{2 \pi}. \nonumber
\end{align}
This characteristic function coincides with that of $X_j + 2 \, \mathrm{Re} ( e^{\mathrm{i} \Theta} \langle f_j | \Phi_N \rangle ) $, where $(X_1,\dots,X_n,\Theta)$ are mutually independent and independent of $\Phi_N$, each $X_j$ is a standard Gaussian random variable, and $\Theta$ is uniformly distributed on $[- \pi, \pi ]$. Since $\widetilde A_{f_j}$ were defined by specifying their joint distribution, we~are free to adopt a particular representation and identify:
   $\widetilde A_{f_j} =  X_j + 2 \, \mathrm{Re} ( e^{\mathrm{i} \Theta} \langle f_j |  \Phi_N \rangle ).$  

    To establish the main claim, it is enough to bound the total variation distance conditionally on $\Theta = \theta$, uniformly in $\theta$. Since the random variables $\beta^{\frac12} X_j + 2 \, \mathrm{Re} ( e^{\mathrm{i} \theta} \langle f_j | \beta^{\frac12}  \Phi_N \rangle )$ and $2 \, \mathrm{Re}  \langle f_j | e^{\mathrm{i} \theta} \Psi \rangle$ are Gaussian, their total variation distance can be bounded using Lemma \ref{lem:gaussian_TV}. Under the assumptions detailed in Proposition \ref{prop:linear_observables} (a), we obtain an explicit rate of convergence. Although the term $\beta^{\frac12} X_j$ in $\widetilde A_{f_j}$ is absent in Proposition \ref{prop:linear_observables}, it merely shifts the covariance matrix by $O(\beta)$, and the proof of Proposition \ref{prop:linear_observables} (a) remains valid.
\end{proof}

\begin{df}
    Let $X=(X_1,\dots,X_k)$ be a random variable valued in $\mathbb R_+^k$. We say that a random variable $Y=(Y_1,\dots,Y_k)$ is mixed Poisson with intensity $X$ if $Y_j \in \mathbb N$ almost surely, and 
    \begin{equation}
        \mathbf P( Y_1=n_1, \dots,Y_k = n_k ) = \mathbf E \Big[ \prod_{j=1}^k \frac{X_j^{n_j} e^{-X_j}}{n_j!} \Big].
    \end{equation}
    We highlight that this definition concerns only the laws of $X$ and $Y$, and does not assume that they are defined on a common probability space. 
\end{df}

\begin{lem} \label{lem:mixed_Poisson}
    Let $X=(X_1,\dots,X_k)$ be a random variable valued in $\mathbb R_+^k$, and let $Y=(Y_1,\dots,Y_k)$ be a~mixed Poisson random variable with intensity $X$. Then the characteristic function of $Y$ is 
    \begin{equation}
        \mathbf E \Big[ \prod_{j=1}^k e^{\mathrm{i}t_j Y_j} \Big] = \mathbf E \Big [ \prod_{j=1}^k e^{(e^{\mathrm{i}t_j}-1) X_j} \Big].
    \end{equation}
    If $p$ is a positive integer and $X_j$ have finite $p$th moments, then $Y_j$ have finite $p$th moment and 
    \begin{equation}
        \mathbf E[Y_j^p] \leq c_p \mathbf E[X_j + X_j^p].
        \label{eq:mixed_Poisson_moment}
    \end{equation}
    If in addition $p$ is even, we have the following bound on the Wasserstein distance of $X$ and $Y$:
    \begin{equation}
        W_p(X,Y) \leq c_p \sum_{j=1}^k \mathbf E[X_j+X_j^{\frac{p}{2}}]^{\frac{1}{p}}. 
        \label{eq:mixed_Poisson_Wasser}
    \end{equation}
\end{lem}
\begin{proof}
    We omit the evaluation of the characteristic function. We construct $X,Y$ on a common probability space $\mathbb R_+^k \times \mathbb N^k$, with $X$ and $Y$ given by projections on the respective Cartesian factors. The joint probability measure is uniquely determined by specifying the prescribed marginal law of $X $ and requiring that conditionally on $X_1=x_1,\dots,X_k=x_k$, random variables $(Y_1,\dots,Y_k)$ are independent Poisson with intensities $(x_1,\dots,x_k)$. By the law of total expectation and standard formulas for moments of Poisson random variables,
    \begin{equation}
        \mathbf E[Y_j^p] = \mathbf E[\mathbf E[Y_j^p|X]] = \sum_{l=1}^p c_{p,l} \mathbf E[  X_j^l] \leq c_p \mathbf E[X_j + X_j^p],
    \end{equation}
    where in the last step we estimated terms with $1 < l < p$ using Young's inequality. Now let $p$ be even. To bound the Wasserstein distance, we use the coupling of $X$ and $Y$ introduced above:
    \begin{equation}
        W_p(X,Y)^p \leq \sum_{j=1}^k \mathbf E[(Y_j-X_j)^p]  = \sum_{j=1}^k\sum_{l=1}^{\frac{p}{2}} c_{p,l}' \mathbf E[X_j^l] \leq  c_p' \sum_{j=1}^k \mathbf E[X_j + X_j^{\frac{p}{2}}],
    \end{equation}
    where we used the fact that $p$th centered moment of a Poisson random variable with intensity $x$ is given by a polynomial in $x$ of degree $\frac{p}{2}$. 
\end{proof}

\begin{proof}[Proof of Theorem \ref{thm:thermaldistr} (b)]
We begin by comparing $Y=(N_{p_1},\dots,N_{p_n})$ with $\widetilde Y=(\widetilde N_{p_1},\dots,\widetilde N_{p_n})$:
    \begin{equation}
    \mathbf{P}( \widetilde N_{p_j} \in B_j \text{ for } j=1,\dots,n )
    = \Tr\!\left[ \prod_{j=1}^n \mathds{1}\big(a_{p_j}^* \a_{p_j} \in B_j\big)\, \widetilde \Gamma_{\beta,N} \right].
\end{equation}
We have $d_{\mathrm{TV}}(Y, \widetilde Y) \leq \| G_{\beta,N} - \widetilde \Gamma_{\beta,N} \|_1 \lesssim N^{-\frac{1}{48}}$. By \cite[Theorem 14]{DeuNamNap-25}, we also have $\mathbf E[N_{p_i}^4] \leq c \beta^{-4}$. As shown below (see the argument below \eqref{eq:tildeNcharacteritic}), $\mathbf E[\widetilde N_{p_i}^4] \leq c \beta^{-4}$ holds as well. By the maximal coupling theorem \cite[Theorem (5.2)]{Lindvall}, one can construct $Y$ and $\widetilde Y$ on a common probability space such that $\mathbf P(Y \neq \widetilde Y) = d_{\mathrm{TV}}(Y,\widetilde Y)$. Having chosen a coupling, we estimate using H\"older's inequality, for $q \in [1,4)$:
\begin{equation}
    \mathbf E[\sum_i |N_{p_i} - \widetilde N_{p_i}|^q] = \sum_i \mathbf E[ |N_{p_i} - \widetilde N_{p_i}|^q \mathds 1(Y \neq \widetilde Y)] \leq  \sum_i \mathbf E[ |N_{p_i} - \widetilde N_{p_i}|^4]^{\frac{q}{4}}  d_{\mathrm{TV}}(Y,\widetilde Y)^{\frac{4-q}{4}}.
\end{equation}
Therefore, $W_q(\beta Y , \beta \widetilde Y) \lesssim N^{- \frac{4-q}{192q}}$. To find the distribution of $\widetilde Y$, we compute the corresponding characteristic function:
\begin{equation}\label{eq:tildeNcharacteritic}
    \mathbf E[e^{\mathrm{i} \sum_j t_j \widetilde N_{p_j}}] = \mathbf E[\langle \Omega_{\Phi_{N,z}} | e^{\mathrm{i} \sum_j t_j \mathcal N_{p_j}} \Omega_{\Phi_{N,z}} \rangle ] = \mathbf E[ e^{ \sum_j (e^{\mathrm{i}t_j}-1) |\Phi_{p_j}|^2}]. 
\end{equation}
By Lemma \ref{lem:mixed_Poisson}, $\widetilde Y$ is mixed Poisson with intensity $X=(|\Phi_{p_1}|^2, \dots, |\Phi_{p_n}|^2)$. 
Since  $\beta |\Phi_{p_j}|^2$ has all moments bounded uniformly in $N$, the same is true for $\beta \widetilde N_{p_j}$ by the estimate \eqref{eq:mixed_Poisson_moment}. 
Moreover, for even integers $p$ we have $W_{p}(\beta \widetilde Y, \beta X) \leq c_p \beta^{\frac12}$ by \eqref{eq:mixed_Poisson_Wasser}. Defining $Z = (|\Psi_{p_1}|^2,\dots, |\Psi_{p_n}|^2)$, we have $\lim_{N \to \infty} W_{p}(\beta X, Z) =0$ by Proposition \ref{prop:linear_observables} (b) applied for plane-wave functions $f_i$, which also provides a rate of convergence under mild additional assumptions. Finally, using the triangle inequality for Wasserstein distance and the monotonicity of the $q$-distance in $q$, we conclude that for all $q \in [1,4)$:
\begin{equation}
    W_q(\beta Y, Z) \leq W_q(\beta Y, \beta \widetilde Y) + W_4(\beta \widetilde Y, \beta X) + W_4(\beta X, Z) \to 0.
\end{equation}

Let us now discuss the random vector $(|\widehat \Psi(p)|^2, |\widehat \Psi(-p)|^2)$. We start from probability density \eqref{eq:Sigma_infty_explicit} and introduce polar coordinates: $\psi_p = \sqrt{x} e^{\mathrm{i} \theta}$, $\psi_{-p} = \sqrt{y} e^{\mathrm{i} \varphi}$. Then $\d \psi_p \d \psi_{-p} = \frac{1}{(2 \pi)^2} \d x \d y \d \theta \d \varphi$. To~obtain the probability density $f(x,y)$, we have to integrate \eqref{eq:Sigma_infty_explicit} with respect to $\theta$ and $\varphi$:
\begin{equation}
    f(x,y) = c \int_0^{2 \pi} \int_0^{2 \pi}  e^{ - (p^2 +m^2 +2 g(\kappa) \widehat v(p)) (x+y) -2 g(\kappa) \widehat v(p) \sqrt{xy} \cos(\theta + \varphi)} \frac{\mathrm{d} \theta}{2 \pi} \frac{\mathrm{d} \varphi}{2 \pi}.
\end{equation}
where $c = (p^2 +m^2)(p^2 +m^2 + 2 g(\kappa) \widehat v(p))$. The integral can be computed using the Jacobi-Anger formula, which establishes \eqref{eq:Bessel_density}.
\end{proof}

\begin{proof}[Proof of Theorem \ref{thm:thermaldistr} (c)]
Following the approach in (a) and (b), we introduce the random variable $\widetilde N_+$:
\begin{equation}
    \mathbf{P}(\widetilde N_+ \in B) = \mathrm{Tr}[ \mathds 1(\mathcal N_+ \in B) \widetilde \Gamma_{\beta,N}].
\end{equation}
Analogously to the proof of (b), we argue that $W_q(\beta N_+-\mathbf E[\beta N_+], \beta \widetilde N_+-\mathbf E[\beta N_+]) \lesssim N^{- \frac{4-q}{192q}}$ for $q \in [1,4)$. Furthermore,
\begin{equation}
    W_q( \beta \widetilde N_+-\mathbf E[\beta N_+],\beta \widetilde N_+-\mathbf E[\beta \widetilde N_+] ) \leq \beta |\mathbf{E}[N_+] - \mathbf E[\widetilde N_+]  | \lesssim N^{-\frac{1}{48}}
\end{equation}
by \cite[Theorem 4 (a)]{DeuNamNap-25}. As in the proof of (b), we verify that $\widetilde N_+$ is mixed Poisson with intensity $\| \Phi_N \|^2$. Therefore, by Lemma \ref{lem:mixed_Poisson}, for every even positive integer $p$ we have:
\begin{align}
    W_p(\beta  \widetilde N_+ - \mathbf E[\beta \widetilde N_+], \beta \| \Phi_N \|^2 - \mathbf E[\beta \| \Phi_N \|^2]) &= \beta W_p (\widetilde N_+, \| \Phi_N \|^2) \\
    &\leq c_p \beta \, \mathbf E[ \| \Phi_N \|^2 + \| \Phi_N \|^p]^{\frac{1}{p}} \leq c_p' N^{-\frac16}, \nonumber
\end{align}
where we used that the moments of $\frac{1}{N} \| \Phi_N \|^2$ are bounded uniformly in $N$ by Proposition \ref{prop:semiclassics_Wick_square}. Using the triangle inequality and invoking Proposition \ref{prop:semiclassics_Wick_square} again, we have for $q \in [1,4)$:
\begin{equation}
    \beta \widetilde N_+ - \mathbf E[\beta \widetilde N_+] \xrightarrow{W_q} \int_\Lambda :|\Psi(x)|^2: \d x = \sum_{p \in \Lambda_+^*} \left( |\widehat \Psi(p)|^2 - \mathbf E[|\widehat \Psi(p)|^2] \right).
\end{equation}
Expressing $\Psi$ as in \eqref{eq:Psi_eta}, we find
\begin{equation}
    \sum_{p \in \Lambda_+^*} \left( |\widehat \Psi(p)|^2 - \mathbf E[|\widehat \Psi(p)|^2] \right) = \sum_{p \in A} \left( \frac{|\eta_{p,+}|^2-1}{p^2 +m^2 + 2 g(\kappa) \widehat v(p)}  + \frac{|\eta_{p,-}|^2-1}{p^2 +m^2} \right),
\end{equation}
which reproduces \eqref{eq:W_def} with $X_p = |\eta_{p,-}|^2$ and $Y_p = |\eta_{p,+}|^2$. 
\end{proof}

\section{Local counting statistics}
\subsection{Proof of Theorem \ref{thm:localcounting}} \label{sec:prooflocalcounting}
Here we deal with the system in the presence of a condensate.

\begin{proof}[Proof of Theorem \ref{thm:localcounting} (a)]
As discussed in Appendix \ref{app:GibbsPointProcess}, the random variable $\Theta_{\beta,N}(f)$ is defined for all real-valued $f \in L^1(\Lambda)$ and satisfies $\mathbf{E}[ \Theta_{\beta,N}(f) ] =  N \widehat f(0)$. Let $f_+ = f - \widehat f(0)$. We will now consider the variance of $\Theta_{\beta,N}(f)$, carrying out the forthcoming calculations first for smooth $f$ and then invoking standard density arguments, which we will not spell out in detail. We have
\begin{align} \label{eq:vartheta}
    \mathbf{Var}[\Theta_{\beta,N}(f)] & = \mathrm{Tr}[G_{\beta,N} \, (\mathrm{d} \Gamma(f) - N \widehat f(0))^2] \\
    & = \mathrm{Tr}[G_{\beta,N} \, ( \widehat f(0) (\mathcal N -N) +  a^*(f_+) \a_0 + a_0^* a(f_+) + \mathrm{d} \Gamma(Q f_+ Q))^2]. \nonumber 
\end{align}
Let us denote 
\begin{equation}
A = \widehat f(0) (\mathcal N -N) +  a^*(f_+) \a_0 + a_0^* a(f_+), \qquad     B = \mathrm{d} \Gamma(Q f_+ Q).
\end{equation}
We will show that $\mathrm{Tr}[G_{\beta, N}  A^2] = O(N^{\frac53})$ for $f \in L^2(\Lambda)$, and that $\mathrm{Tr}[ G_{\beta, N}  B^2] = o(N^{\frac53})$ for $f \in H^{\frac15}(\Lambda)$. Under the latter assumption, $    \mathrm{Tr}[G_{\beta,N} (A+B)^2] = \mathrm{Tr}[G_{\beta,N} A^2] + o(N^{\frac53})$ holds by Cauchy-Schwarz. More generally, $\mathrm{Tr}[ G_{\beta, N}  B^2] = o(N^2)$ for $f \in L^2(\Lambda)$, and we can still conclude that $\mathbf{Var}[\Theta_{\beta,N}(f)] = o(N^2)$. 

Let us first consider the $B^2$ term. Expressing $f$ with its Fourier series,
\begin{align}
    \mathrm{Tr}[G_{\beta,N} \mathrm{d} \Gamma(Qf_+ Q)^2] &= \sum_{p,q \in \Lambda_+^*} \widehat f(p) \widehat f(q) \mathrm{Tr}[G_{\beta,N} \mathrm{d} \Gamma(e^{\mathrm{i}px}) \mathrm{d} \Gamma(e^{\mathrm{i}qx})] \\
   & = \sum_{p \in \Lambda^*_+} |\widehat f(p)|^2 \mathrm{Tr}[G_{\beta,N} \mathrm{d} \Gamma(e^{\mathrm{i}px}) \mathrm{d} \Gamma(e^{-\mathrm{i}px})], \nonumber
\end{align}
where in the last step we dropped terms with $q \neq -p$ by translation invariance. Using translation invariance again, we can rewrite this expression using the operators $B_p = \mathrm{d} \Gamma(Q \cos (px) Q)$:
\begin{equation}
    \mathrm{Tr}[G_{\beta,N} d \Gamma(Qf_+ Q)^2] =  2 \sum_{p \in \Lambda_+^*} |\widehat f(p)|^2 \mathrm{Tr}[G_{\beta, N} B_p^2]. 
\label{eq:QfQ_term}
\end{equation}
By \cite[Theorem 13 (b)]{DeuNamNap-25}, we have the estimates
\begin{equation}
    \mathrm{Tr}[G_{\beta,N} B_p^2] \leq c (N^{\frac43} + N^{\frac13} p^2).
    \label{eq:Bp_bound}
\end{equation}
The second term is unbounded for $p \to \infty$, and for very large $p$ we will instead use the easier estimate
\begin{equation}
    \mathrm{Tr}[G_{\beta,N} B_p^2] \leq \mathrm{Tr}[G_{\beta,N} \mathcal N^2] \leq c N^2. 
    \label{eq:B_easy_est}
\end{equation}
If $f \in H^1(\Lambda)$, using \eqref{eq:Bp_bound} in \eqref{eq:QfQ_term} immediately gives $\mathrm{Tr}[G_{\beta,N} \mathrm{d} \Gamma(Qf_+ Q)^2] \leq c N^{\frac43} \| f \|_{H^1(\Lambda)}^2$. Let us consider a more general case, $f \in H^s(\Lambda)$ for $s \in [0,1)$. We split the sum in \eqref{eq:QfQ_term} into three regions: small momenta $|p| \leq N^{\frac12}$, intermediate momenta $N^{\frac12} < |p| \leq N^{\frac{5}{6}}$, and large momenta $|p| > N^{\frac{5}{6}}$. For small momenta, \eqref{eq:Bp_bound} is bounded by $N^{\frac43}$, and the corresponding contribution to the sum in \eqref{eq:QfQ_term} is bounded by $N^{\frac43} \| f \|_{L^2}^2$. For intermediate momenta, we bound $p^2 \leq |p|^{2s} N^{\frac53 (1-s)}$, obtaining:
\begin{equation}
    \sum_{N^{\frac12}<|p| \leq N^{\frac56}} |\widehat f(p)|^2 \mathrm{Tr}[G_{\beta, N} B_p^2] \leq c N^{2 -\frac53 s }     \sum_{N^{\frac12}<|p| \leq N^{\frac56}}  |p|^{2s} |\widehat f(p)|^2.
\end{equation}
The sum on the right hand side is $o(1)$ because it is bounded by the tail of a convergent series. Finally, for large momenta we use \eqref{eq:B_easy_est} and bound $1 \leq \frac{|p|^{2s}}{N^{\frac{5}{3}s }}$:
\begin{equation}
    \sum_{N^{\frac56} <|p| } |\widehat f(p)|^2 \mathrm{Tr}[G_{\beta, N} B_p^2] \leq c N^{2 - \frac53 s} \sum_{N^{\frac56} <|p| } |p|^{2s} |\widehat f(p)|^2.
\end{equation}
Again, the sum on the right hand side is $o(1)$. Combining the obtained estimates,
\begin{equation}\label{eq:Bsquareestimate}
    \mathrm{Tr}[G_{\beta,N} B^2] = O(N^{\frac43}) + o(N^{2 - \frac{5}{3}s}).
\end{equation}
In particular, for $s=0$ the right hand side is $o(N^2)$, for $s \geq \frac15$ it is $o(N^{\frac53})$, and for $s \geq \frac25$ it is $O(N^{\frac43})$.

Now consider $A =A_1 + A_2 $, where $A_1 = \widehat f(0) (\mathcal N - N)$ and $A_2 = a^*(f_+) \a_0 + a_0^* a(f_+)$. We have
\begin{equation}
    \mathrm{Tr}[G_{\beta,N} A^2] = \mathrm{Tr}[G_{\beta,N} (A_1^2 + A_2^2) ] + \mathrm{Tr}[G_{\beta,N} (A_1A_2 + A_2 A_1)]. 
\end{equation}
The second trace on the right hand side vanishes by the translation invariance of $G_{\beta,N}$. From  \cite[equation (11.27)]{DeuNamNap-25} we have 
\begin{equation}\label{eq:varianceN_DNN}
    \mathrm{Tr}[G_{\beta,N} \, A_1^2]  = \frac{N}{\beta} \frac{\widehat f(0)^2}{\widehat v(0)} + O (N^{\frac53 - \frac{1}{4}}).
\end{equation}
To deal with the term involving $A_2$ we use  \cite[Theorem 5 (a)]{DeuNamNap-25} to get  
\begin{equation} \label{eq:a+a0var}
     \mathrm{Tr}[G_{\beta,N} \, A_2^2] = N  g^{\id}\sum_{p\in \Lambda^*_+} \left(2|\hat{f}(p)|^2 \gamma^{\Bog}(p)+(\hat{f}(p)\hat{f}(-p)+\overline{\hat{f}(p)\hat{f}(-p)}) \alpha^{\Bog}(p)\right)+O(N^{\frac53-\frac{1}{96}}).
\end{equation}
Here $\gamma^{\Bog},\alpha^{\Bog}$ are given as in \eqref{def:gammabogalphabog} but with $\hat{u}(p)$ replaced with $\hat{v}(p)$ (cf. Remark \ref{rem:gammalpahtilde}). The first term on the RHS of \eqref{eq:a+a0var} above is $O(N^{5/3})$ for $f\in L^2(\Lambda)$. Putting together \eqref{eq:a+a0var}, \eqref{eq:varianceN_DNN}, \eqref{eq:Bsquareestimate} and  \eqref{eq:vartheta} shows that
\begin{equation} \label{eq:varthetafinal}
    \mathbf{Var}[\Theta_{\beta,N}(f)] = \frac{N}{\beta} \frac{\widehat f(0)^2}{\widehat v(0)} + N  g^{\id}\sum_{p\in \Lambda^*_+} \left(2|\hat{f}(p)|^2 \gamma^{\Bog}(p)+2 \operatorname{Re}(\hat{f}(p)\hat{f}(-p)) \alpha^{\Bog}(p)\right)  + o(N^{5/3})
\end{equation}
when $f\in H^s(\Lambda)$ with $s\geq\frac15$.
\end{proof}

\begin{remark}
    Let $\widetilde \Theta_{\beta,N}$ be the point process associated to $\widetilde \Gamma_{\beta,N}$. Then a direct computation similar to the one in the proof above yields
    \begin{equation}
    \label{eq:vartildethetafinal}
    \mathbf{Var}[\widetilde\Theta_{\beta,N}(f)] = \frac{N}{\beta} \frac{\widehat f(0)^2}{\widehat v(0)} + N  g^{\id}\sum_{p\in \Lambda^*_+} \left(2|\hat{f}(p)|^2 \widetilde\gamma^{\Bog}(p)+2 \operatorname{Re}(\hat{f}(p)\hat{f}(-p)) \widetilde\alpha^{\Bog}(p)\right)  + o(N^{5/3})
    \end{equation}
where $\widetilde\gamma^{\Bog}$ and $\widetilde\alpha^{\Bog}$ are defined in \eqref{def:gammabogalphabog}. In particular, since $\widetilde\gamma^{\Bog}$ and $\widetilde\alpha^{\Bog}$ coincide with $\gamma^{\Bog}$ and $\alpha^{\Bog}$ for $|p| \leq p_c$ it is easy to see from \eqref{eq:varthetafinal} and \eqref{eq:vartildethetafinal} that
\begin{equation}
\label{eq:varthetaminusvartildetheta}
\mathbf{Var}[\Theta_{\beta,N}(f)]-\mathbf{Var}[\widetilde\Theta_{\beta,N}(f)]=o(N^{\frac53})
\end{equation}
when $f\in H^s(\Lambda)$ with $s\geq\frac15$. 
\end{remark}

\begin{proof}[Proof of Theorem \ref{thm:localcounting} (b)]
We present the proof for $n=1$ for brevity. Let $Y=\Theta_{\beta,N}(f) $ with $f \in H^s(\Lambda)$ for some $s \geq \frac15$. We will compare $Y$ with $\widetilde Y = \widetilde \Theta_{\beta,N}(f) $, where $\widetilde \Theta_{\beta,N}$ is the point process associated to $\widetilde \Gamma_{\beta,N}$. That is, $\widetilde Y$ is a~real random variable with distribution
\begin{equation}
    \mathbf P(\widetilde Y \in B) = \mathrm{Tr}[ \mathds 1(\mathrm{d} \Gamma(f)  \in B ) \widetilde \Gamma_{\beta,N}  ].
\end{equation}
We note that $\mathbf E[Y]= \mathbf E[\widetilde Y] = N \widehat f(0)$. 

If $f$ is bounded, then $d_{\mathrm{TV}}(Y,\widetilde Y) \lesssim  N^{- \frac{1}{48}}$. This estimate extends to all $f \in L^1(\Lambda)$ because the total variation distance is lower-semicontinuous with respect to convergence in distribution (see also the discussion around \eqref{eq:intensity_measure_is_Lebesgue} in Appendix \ref{app:GibbsPointProcess}). Invoking the maximal coupling theorem, we~construct $Y, \widetilde Y$ on a joint probability space such that $\mathbf P(Y \neq \widetilde Y)= d_{\mathrm{TV}}(Y,\widetilde Y)$. Then
\begin{align}
    \mathbf E[(Y-\widetilde Y)^2] &= \mathbf E[(Y-\widetilde Y)^2 \mathds 1(Y \neq \widetilde Y)] \\ 
    &\leq 2 \mathbf E[((Y- \mathbf E[Y])^2 + (\widetilde Y - \mathbf E[\widetilde Y])^2) \mathds 1( Y \neq \widetilde Y )] \nonumber \\
     &= 2 \mathbf{Var}[Y] - 2 \mathbf{Var}[\widetilde Y] + 4 \mathbf E[(\widetilde Y - \mathbf E[\widetilde Y])^2 \mathds 1(Y \neq \widetilde Y)]. \nonumber
\end{align}
From \eqref{eq:varthetaminusvartildetheta} we have $\mathbf{Var}[Y] - \mathbf{Var}[\widetilde Y] = o( N \beta^{-1})$, with an explicit rate if $s > \frac15$. We~estimate the second term of the right hand side using H\"older's inequality:
\begin{equation}
\mathbf E[(\widetilde Y - \mathbf E[\widetilde Y])^2 \mathds 1(Y \neq \widetilde Y)] \leq \mathbf E[(\widetilde Y-\mathbf E[\widetilde Y])^{2q}]^{\frac{1}{q}} d_{\mathrm{TV}}(Y, \widetilde Y)^{1-\frac{1}{q}}
\end{equation}
for any $q > 1$. One verifies that $\mathbf E[(\widetilde Y-\mathbf E[\widetilde Y])^{2q}]^{\frac{1}{q}} \leq c_q N \beta^{-1}$, and therefore:
\begin{equation}
    \mathbf E[(\widetilde Y - \mathbf E[\widetilde Y])^2 \mathds 1(Y \neq \widetilde Y)] \leq c_\delta N \beta^{-1} N^{- \frac{1}{4 8} + \delta}
\end{equation}
holds for every $\delta >0$. Combining the obtained inequalities, we find that
\begin{equation}
 W_2 \left( \sqrt{\frac{\beta}{N}} (Y- \mathbf E[Y]), \sqrt{\frac{\beta}{N}} (\widetilde Y-\mathbf E[Y]) \right) = W_2 \left( \sqrt{\frac{\beta}{N}} Y, \sqrt{\frac{\beta}{N}} \widetilde Y \right) \xrightarrow{N \to \infty} 0.
   \label{eq:local_linear_state_comparison_W2}
\end{equation}

Let us now specialize to an indicator function $f = \mathbf 1_D$. Therefore, $Y = N_D$. If $D$ has smooth boundary, then $f \in H^s(\Lambda)$ for all $s < \frac12$ \cite[Theorem~11.4]{LionsMagenes1972}, so \eqref{eq:local_linear_state_comparison_W2} holds. By part (a) of the theorem, $\widetilde Y$ is mixed Poisson with intensity $M=\| \mathbf 1_D \Phi_{N,z} \|^2$. Therefore, by Lemma \ref{lem:mixed_Poisson}:
\begin{align}
 W_2 \left( \sqrt{\frac{\beta}{N}} (\widetilde Y- \mathbf E[\widetilde Y]) , \sqrt{\frac{\beta}{N}} (M-\mathbf E[M]) \right) & =    W_2 \left( \sqrt{\frac{\beta}{N}} \widetilde Y , \sqrt{\frac{\beta}{N}} M \right) \\
& = \sqrt{\frac{\beta}{N}} W_2(\widetilde Y,M) \leq \sqrt{\beta |D|} \xrightarrow{N \to \infty} 0. \nonumber
\end{align}
Let us now analyze the random intensity $M$. We can decompose $M - \mathbf E[M] = M_1 + M_2$, where
\begin{subequations}
\begin{align}
    M_1 & =  (|Z|^2 - \mathbf E[|Z|^2]) |D| + 2 \sqrt{N g^\id} \, \mathrm{Re} \, \langle \mathbf 1_D | \Phi_N \rangle, \label{eq:Z1} \\
    M_2 & = 2 (|Z| - \sqrt{N g^\id}) \, \mathrm{Re} \, \langle \mathbf 1_D | \Phi_N \rangle + \langle \Phi_N | \mathbf 1_D \Phi_N \rangle - \mathbf E[ \langle \Phi_N | \mathbf 1_D \Phi_N \rangle]. \label{eq:Z_2}
\end{align}
 \end{subequations}
We will see that the terms in $M_2$ are negligible and bound:
\begin{equation}
    W_2 \left(\sqrt{\frac{\beta}{N}} (M- \mathbf E[M]), \sqrt{\frac{\beta}{N}} M_1 \right) \leq \sqrt{\frac{\beta}{N}} \sqrt{\mathbf{Var}[M_2]} \leq c N^{ -\frac16}. \label{eq:2137}
\end{equation}
To get the second inequality in \eqref{eq:2137}, we separately estimate two terms in the definition of $M_2$ and invoke Cauchy-Schwarz. First, using \eqref{eq:varZsqNgid}, we obtain
\begin{align}
    \mathbf{Var} \left[  (|Z|-\sqrt{N g^\id}) \, \mathrm{Re} \, \langle \mathbf 1_D | \Phi_N \rangle  \right] \leq \frac{1}{N g^\id \beta} \mathbf E[(|Z|^2 - Ng^\id)^2 ] \mathbf E[ (\mathrm{Re} \, \langle \mathbf 1_D | \beta^{\frac12} \Phi_N \rangle)^2 ] \leq c \frac{1}{g^\id \beta^2}.
\end{align}
Then, using \eqref{def:gammabogalphabog}, we obtain
\begin{align}
    \mathbf{Var}[ \langle \Phi_N | \mathbf 1_D \Phi_N \rangle  ] &= \sum_{p,q \in \Lambda^*_+} |\widehat {\mathbf 1}_D(p-q)|^2 (\widetilde \gamma^\Bog(p) \widetilde \gamma^\Bog(q) + \widetilde \alpha^\Bog(p) \widetilde \alpha^\Bog(q)) \\
    & \lesssim  \sum_{p,k \in \Lambda^*} |\widehat {\mathbf 1}_D(k)|^2 \frac{1}{\beta (|p|^2+1)} \frac{1}{\beta (|p+k|^2+1)} \nonumber \\
   & \lesssim  \beta^{-2}  \sum_{k \in \Lambda^*} |\widehat {\mathbf 1}_D(k)|^2 \cdot \max_{k \in \Lambda^*} \sum_{p \in \Lambda^*} \frac{1}{|p|^2 +1} \frac{1}{|p+k|^2 +1} \lesssim \beta^{-2}. \nonumber
\end{align}
This completes the proof of \eqref{eq:2137}. Finally,
\begin{align}
    W_2 \left( \sqrt{\frac{\beta}{N}} M_1,  \frac{1}{\sqrt{\widehat v(0)}} X +  2 \sqrt{g(\kappa)} \, \mathrm{Re} \, \langle \mathbf 1_D | \Psi \rangle \right)^2 & \leq W_2 \left( \sqrt{\frac{\beta}{N}} (|Z|^2 - \mathbf E[|Z|^2]), \frac{1}{\sqrt{\widehat v(0)}} X \right)^2 \\
    & +    W_2 \Big (2 \sqrt{g^\id} \, \mathrm{Re} \, \langle \mathbf 1_D | \beta^{\frac12} \Phi_N \rangle , 2 \sqrt{g(\kappa)} \, \mathrm{Re} \, \langle \mathbf 1_D | \Psi \rangle \Big)^2. \nonumber 
\end{align}
Here, the first term converges to zero by \eqref{eq:W2condtoX}, and the second term converges to zero because $g^\id $ converges to $g(\kappa)$, and $ \mathrm{Re} \, \langle \mathbf 1_D | \beta^{\frac12} \Phi_N \rangle \xrightarrow{W_2} \mathrm{Re} \, \langle \mathbf 1_D | \Psi \rangle $ by Proposition \ref{prop:linear_observables}.

In the more general case of a Borel set $D \subset \Lambda$ not necessarily with smooth boundary, it is still true that distributions of random variables $Y,\widetilde Y$ are close in the total variation distance, and the arguments above establish in particular that $\sqrt{\frac{\beta}{N}}(\widetilde Y - \mathbf E[\widetilde Y] )$ converges to $\frac{1}{\sqrt{\widehat v(0)}} X +  2 \sqrt{g(\kappa)} \, \mathrm{Re} \, \langle \mathbf 1_D | \Psi \rangle$ in distribution. 
\end{proof}

\begin{proof}[Proof of Theorem \ref{thm:localcounting} (c)]
    By Theorem \ref{thm:semiclassicalApproximation} (c), it is enough to consider the point process $\widetilde \Theta_{\beta,N}$ instead of $\Theta_{\beta,N}$. Let $D \subset \mathbb R^3$ be a bounded Borel set and let $D_N = N^{-\frac13} D$. Then $D_N \subset (- \frac12 ,\frac12)^3$ for large enough $N$, and we view $D_N$ as a subset of $\Lambda$ via the inclusion $(-\frac12,\frac12)^3 \subset \Lambda$. Let $f $ be a continuous function on $\mathbb R^3$ supported on $D$, and let $f_N(x) = f(N^{\frac13} x)$, so that $f_N$ is supported on $D_N$. We have to show that
\begin{equation}
    \mathbf E[\exp(\mathrm{i} \widetilde \Theta_{\beta,N}(f))] \xrightarrow{N \to \infty} \mathbf E\left[ \exp \left(\int (e^{\mathrm{i}f(x)}-1) |\sqrt{g (\kappa}) + H(x)|^2 \d x \right)  \right]
\end{equation}
where $H(x)$ is the Gaussian random field introduced in Section \ref{sec:bosonPointProcess}. This reduces to proving the weak convergence of random variables
    \begin{equation}
\int  ||Z| + \Phi_N(x)|^2 f_N(x) \d x \xrightarrow{N \to \infty} \int |\sqrt{g (\kappa)} + H(x)|^2 f(x) \d x.
\label{eq:microstat_proof1}
    \end{equation}
    Let us now observe that
    \begin{align}
        &\mathbf E \left[ \left| \int  \big ||Z| + \Phi_N(x) \big|^2 f_N(x) \d x - \int  \big|\sqrt{N g(\kappa)} + \Phi_N(x) \big|^2 f_N(x) \d x   \right| \right] \label{eq:microstat_proof2}  \\
        = & \mathbf E \left[ \left| \int \left( |Z|^2 - N g(\kappa) + 2 \left(|Z| - \sqrt{N g(\kappa)} \right) \mathrm{Re} \, \Phi_N(x) \right) f_N(x) \d x  \right| \right] \nonumber \\
        \leq & \frac{1}{N}\mathbf E \left[ \big| |Z|^2 - N g (\kappa) \big| \right] \int |f(x)| \d x + 2 \, \mathbf E[ \big| |Z| - \sqrt{N g(\kappa)} \big|^2]^{\frac12} \int \mathbf E[|\Phi_N(x)|^2]^{\frac12} |f_N(x)| \d x. \nonumber
    \end{align}
    We have $\mathbf E \left[ \big| |Z|^2 - N g (\kappa) \big| \right] = o(N)$, so the first term converges to zero as $N \to \infty$. For the second term, we have $\mathbf E[ \big| |Z| - \sqrt{N g(\kappa)} \big|^2]^{\frac12} = o(N^{\frac12})$ and
    \begin{equation}
        \int \mathbf E[|\Phi_N(x)|^2]^{\frac12} |f_N(x)| \d x \leq N^{\frac12} \int |f_N(x)| \d x = N^{-\frac12} \int |f(x)| \d x.
    \end{equation}
    We conclude that the expectation in \eqref{eq:microstat_proof2} converges to zero as $N \to \infty$, so 
    \begin{equation}
        \int  \big ||Z| + \Phi_N(x) \big|^2 f_N(x) \d x - \int  \big|\sqrt{N g(\kappa)} + \Phi_N(x) \big|^2 f_N(x) \d x \to 0
    \end{equation}
    in probability. Therefore, by Slutsky's theorem, verification of \eqref{eq:microstat_proof1} reduces to showing
    \begin{equation}
\int  |\sqrt{N g(\kappa)} + \Phi_N(x)|^2 f_N(x) \d x \xrightarrow{N \to \infty} \int |\sqrt{g (\kappa)} + H(x)|^2 f(x) \d x
\label{eq:microstat_proof3}
    \end{equation}
in distribution.    

Next, we will replace the field $\Phi_N(x)$ with
    \begin{equation}
        \Phi_N'(x) = \sum_{p \in \Lambda_+^*} \widehat \Phi_N'(p) e^{\mathrm{i}px},
    \end{equation}
    where $\widehat \Phi_N'(p)$ are independent complex Gaussian random variables with the covariance
    \begin{equation}
        \mathbf E[ |\widehat \Phi_N'(p)|^2  ] = \frac{1}{e^{\beta p^2}-1} 
    \end{equation}
    and vanishing pseudo-covariance. We construct $\Phi_N$ and $\Phi_N'$ on a common probability space using the definition \eqref{eq:Phi_from_eta} for $\Phi_N$ and
    \begin{subequations}
    \begin{align}
    \widehat \Phi_N'(p) & = \frac{1}{\sqrt{2 (e^{\beta p^2}-1)}} (\eta_{p,+}+ \eta_{p,-}), \\
    \widehat \Phi_N'(-p) & = \frac{1}{\sqrt{2 (e^{\beta p^2}-1)}} (\overline{\eta_{p,+}}- \overline{\eta_{p,-}}). 
    \end{align}
    \end{subequations}
    Then it is easy to see that
    \begin{equation}
        \mathbf E[ \| \Phi_N - \Phi_N' \|_{L^2(\Lambda)}^2] = o(N),
    \end{equation}
    and therefore by translation invariance
    \begin{equation}
        \mathbf E[ |\Phi_N(x) - \Phi_N'(x)|^2 ] = o(N).
    \end{equation}
    Using this result and Slutsky's theorem, we can reduce \eqref{eq:microstat_proof3} to
    \begin{equation}
        \int  |\sqrt{N g(\kappa)} + \mathbf 1_{D_N}(x) \Phi_N'(x)|^2 f_N(x) \d x \xrightarrow{N \to \infty} \int |\sqrt{g (\kappa)} + \mathbf 1_D(x) H(x)|^2 f(x) \d x.
        \label{eq:microstat_proof4}
    \end{equation}
    Here we introduced the indicators of $D_N$ and $D$, which is harmless because the respective integrands vanish away from $D_N$ and $D$. The change of variables $x = N^{-\frac13} y$ brings the left hand side to the form
    \begin{equation}
        \int  |\sqrt{N g(\kappa)} + \Phi_N'(x)|^2 f_N(x) \d x = \int |\sqrt{g(\kappa)} + K_N(y)|^2 f(y) \d y,
    \end{equation}
    where $K_N$ is the $L^2(D)$-valued Gaussian random field 
    \begin{equation}
        K_N(y) = N^{-\frac12} \mathbf 1_D(y) \Phi_N'(N^{-\frac13}y),
    \end{equation}
    which has the covariance (understood as integral kernel of a positive operator $\Sigma_N$ on $L^2(D)$):
    \begin{align}
       \Sigma_N(x,y) := \mathbf E[\overline {K_N(x)} K_N(y) ] &= \mathbf 1_D(x) \mathbf 1_D(y) \frac{1}{N} \sum_{p \in \Lambda_+^*} \frac{e^{\mathrm{i}p N^{-\frac13} (x-y)}}{e^{\beta p^2}-1} \label{eq:SigmaN_SigmaInfty} \\
       & = \mathbf 1_D(x) \mathbf 1_D(y) \frac{1}{N} \left( \int_{\mathbb R^3} \frac{e^{\mathrm{i}p N^{-\frac13} (x-y)}}{e^{\beta p^2}-1} \frac{\d p}{(2 \pi)^3} + O(\beta^{-1}) \right) \nonumber \\
       & = \mathbf 1_D(x) \mathbf 1_D(y) \left( \int_{\mathbb R^3} \frac{e^{\mathrm{i}q(x-y)}}{e^{\beta N^{\frac23} q^2}-1} \frac{\d q}{(2 \pi)^3} + O(N^{- \frac13}) \right), \nonumber \\
       & = \mathbf 1_D(x) \mathbf 1_D(y) \left( \int_{\mathbb R^3} \frac{e^{\mathrm{i}q(x-y)}}{e^{\beta_\infty(\kappa) q^2}-1} \frac{\d q}{(2 \pi)^3} + o(1)) \right) .\nonumber
    \end{align}
    Here in the passage from the first to the second line we used the $\mu \to 0$ limiting form of \eqref{eq:n_function} in Appendix \ref{app:lemmas} (note the sum above is over $p \neq 0$), passing from the second to the third line we made the change of variables $p = N^{\frac13} q$, and~in the final step we used the convergence $\beta N^{\frac23} \to \beta_\infty(\kappa)$.
    
    Let us compare to the covariance $\Sigma_\infty \in B(L^2(D))$ of the field $K_\infty(x)=\mathbf 1_D(x) H(x)$:
    \begin{equation}
        \Sigma_\infty(x,y) = \mathbf E[\overline{K_\infty(x)} K_\infty(y)] = \mathbf 1_D(x) \mathbf 1_D(y) \int_{\mathbb R^3} \frac{e^{\mathrm{i}q(x-y)}}{e^{\beta_\infty(\kappa) q^2}-1} \frac{\d q}{(2 \pi)^3}.
    \end{equation}
    The integral kernels of $\Sigma_N,\Sigma_\infty$ are supported on a bounded set $D$ and pointwise close for large $N$ uniformly in $x,y$. Therefore, $\Sigma_N \xrightarrow{N \to \infty} \Sigma_\infty$ in the Hilbert-Schmidt sense. Furthermore, we have the trace formulas
    \begin{subequations}
    \begin{align}
    \mathrm{Tr}[\Sigma_\infty] &= |D| \int_{\mathbb R^3} \frac{1}{e^{\beta_\infty(\kappa)q^2}-1} \frac{\d q}{(2 \pi)^3}, \\
    \mathrm{Tr}[\Sigma_N] & = |D| \frac{1}{N} \sum_{p \in \Lambda^*_+} \frac{1}{e^{\beta p^2}-1}.
    \end{align}
    \end{subequations}
    The considerations in \eqref{eq:SigmaN_SigmaInfty} now show that we have convergence $\mathrm{Tr}[\Sigma_N] \to \mathrm{Tr}[\Sigma_\infty]$. Combined with the Hilbert-Schmidt convergence $\Sigma_N \to \Sigma_\infty$ and the positivity of $\Sigma_N,\Sigma_\infty$, this implies that $\Sigma_N \to \Sigma_\infty$ in trace class. By Powers–Størmer inequality \cite[Lemma~4.1]{PowersStormer1970},
    \begin{equation}
        \mathrm{Tr}[ (\Sigma_N^\frac12 - \Sigma_\infty^\frac12)^2 ] \leq \mathrm{Tr}[ |\Sigma_N - \Sigma_\infty| ]
    \end{equation}
    we also have convergence of square roots: $\Sigma_N^\frac12 \to \Sigma_\infty^\frac12$ in the Hilbert-Schmidt sense.   

    We can now construct $K_N(x)$ and $K_{\infty}(x)=\mathbf 1_D(x) H(x)$ on a common probability space using standard Gaussian coupling. Let $W$ be $L^2(D)$ complex Gaussian white noise:
    \begin{equation}
        W(x) = \sum_{n=1}^\infty W_n f_n(x),
        \label{eq:white_noise}
    \end{equation}
    where $(W_n)_{n=1}^\infty$ is a sequence of independent complex Gaussian random variables with covariance $\mathbf E[|W_n|^2]=1$ and vanishing pseudo-covariance, and $(f_n)_{n=1}^\infty$ is an orthonormal basis of $L^2(D)$. Here we omit treatment of the trivial case $|D|=0$, in which $L^2(D) = \{ 0 \}$. We remark that \eqref{eq:white_noise} does not define a random variable valued in $L^2(D)$ because $\sum_{n=1}^\infty |W_n|^2 = \infty$ almost surely. The expression \eqref{eq:white_noise} can be given precise meaning by viewing it as a random variable valued in some larger vector space, whose exact choice is largely immaterial. We define the coupling fields as the images of $W$ under square roots of their respective covariance operators: 
    \begin{equation}
        K_N(x)  = \sum_{n=1}^\infty W_n (\Sigma_N^\frac12 f_n)(x), \qquad  K_\infty(x) = \sum_{n=1}^\infty W_n (\Sigma_\infty^\frac12 f_n)(x).
    \end{equation}
    Using the Hilbert-Schmidt property of $\Sigma_N^\frac12$ and $\Sigma_\infty^\frac12$, one verifies that the above series converge almost surely in $L^2(D)$ and therefore define $L^2(D)$-valued random variables. Under this coupling, we have
    \begin{equation}
        \mathbf E[ \| K_N - K_\infty \|^2_{L^2(D)} ] = \mathrm{Tr}[(\Sigma_N^{\frac12} - \Sigma_\infty^\frac12)^2] \xrightarrow{N \to \infty} 0.
    \end{equation}
    It follows that
    \begin{equation}
        \mathbf E\left[  \left| \int  (|\sqrt{ g (\kappa)} + K_N(x)|^2 - |\sqrt{ g (\kappa)} + K_\infty(x)|^2 ) f(x) \d x \right|   \right] \to 0,
    \end{equation}
    and in particular \eqref{eq:microstat_proof3} holds.
\end{proof}
    
\subsection{High-temperature phase}\label{sec:noncond}
In this section we consider local particle statistics for $m^2 = \infty$. This includes the high-temperature phase $\kappa <1$, but is also possible for $\kappa=1$ (cf.~Proposition \ref{prop:limits} and the discussion below). The state approximation in \eqref{eq:norm_approx_ideal}, which shows that mean-field interactions have a~negligible effect on the system in the absence of Bose-Einstein condensation, allows us to reduce the problem to the case of the ideal gas. 

We will count particles in subsets $D \subset \Lambda$ scaled by an $N$-dependent factor $\delta \leq 1$:
\begin{equation}\label{def:setD}
    D = \delta D_0  \qquad  \text{with} \qquad  D_0 \subset (-\tfrac12,\tfrac12)^3, \quad  |D_0| \neq 0, \quad |\partial D_0|=0.
\end{equation}
We probe the system at different length-scales by considering various choices of $\delta$. 

Unlike in the preceding treatment of the condensation phase, we will not detail the random field representation. Instead, we derive limiting distributions of $N_D$ using more elementary arguments.
\begin{proof}[Sketch of proof of Theorem \ref{thm:high_temperature}]
By \eqref{eq:norm_approx_ideal}, it is enough to prove our claims for the non-interacting system. This is included (and slightly extended) in Propositions \ref{prop:localexpvar}, \ref{prop:micro_limit}, \ref{prop:special_limit} below, where we consider limiting distributions of $N_D$ for a single set $D$ to simplify presentation.
\end{proof}

In the remaining part of this section we restrict attention to the non-interacting system, that is with $v=0$. To emphasize that we are considering the ideal gas, we will denote the expectation values and variances of random variables with a superscript $\id$, and the Gibbs state of the system by $G^{\id}$. In~all forthcoming propositions, we implicitly make the assumptions of Theorem \ref{thm:high_temperature}. 

We let $\widetilde N$ be the random variable corresponding to the operator $\mathcal N$. That is,
\begin{equation}
    \mathbf E^\id[f(\widetilde N)] = \mathrm{Tr}[G^{\id} \, f(\mathcal N)] 
\end{equation}
for any bounded continuous function $f$. In particular, $\mathbf E^\id[\widetilde N] = N$. 

\begin{lem} \label{lem:id_variance}
The variance of $N_+$ and $\widetilde N$ in the ideal gas is given by 
\begin{equation}
\mathbf{Var}^\id[\widetilde N] =(1 +o(1)) \begin{cases}
            (4 \pi \beta)^{-\frac32} \mathrm{Li}_{\frac12}(z_\infty), & \kappa<1, \\
            \frac{1}{8 \pi} \beta^{-2} (- \mu^\id)^{- \frac12}, & \kappa =1.
            \end{cases}
\label{eq:ideal_variance_N}
\end{equation}
where $z_\infty = \lim_{N \to \infty} e^{\beta \mu^\id}$.
\end{lem}
\begin{proof}
We have 
\begin{equation}
    \mathbf{Var}^\id[N_+] = \sum_{p \in \Lambda^*} \frac{e^{- \beta (p^2 - \mu^\id)}}{(1-e^{- \beta (p^2 - \mu^\id)})^2} = \sum_{m=1}^\infty m e^{m \beta \mu^\id} (\vartheta (0 ; 4 \pi i m \beta)^3-1),
    \label{eq:Var0_basic}
\end{equation}
    where in the second step we expand the summand in power series in $e^{\beta (p^2 - \mu^\id)}$, interchange the order of sums and recognize the Jacobi theta function, see Appendix \ref{app:lemmas}. We split the sum into parts $m \leq M$ and $m > M$, where $M = \lfloor (4 \pi \beta)^{-1} \rfloor$. 
    
    By Lemma \ref{lem:gaussian}, the large $m$ contribution is non-negative and bounded, up to a~constant, by
    \begin{equation}
        \sum_{m=M+1}^\infty m e^{m \beta (\mu^\id-4 \pi^2)} = \frac{M x^{M+1}}{1-x} + \frac{x^{M+1}}{(1-x)^2},
    \label{eq:Var_id_largem}
    \end{equation}
    where $x = e^{\beta \mu^\id} e^{-4 \pi^2 \beta}$. The expression \eqref{eq:Var_id_largem} is $O(\beta^{-2} (1- \mu^\id)^{-1} e^{\mu^\id/4 \pi})$. To get this bound we use: $x^{M+1} \leq e^{-\pi} e^{\frac{\mu^\id}{4\pi}}$, $M \leq (4 \pi \beta)^{-1}$, $1-x \geq x \log x^{-1}$, and $x \geq c$ for some $c>0$ (which holds by (5) in Proposition \ref{prop:limits}).

    We consider the sum over $m \leq M$. Again by Lemma \ref{lem:gaussian}:
    \begin{align}
    \sum_{m=1}^M m e^{m \beta \mu^\id} ((4 \pi m \beta)^{-\frac32}-1) &\leq    \sum_{m=1}^M m e^{m \beta \mu^\id} (\vartheta (0 ; 4 \pi i m \beta)^3 -1 )  \label{eq:label12345} \\
    &\leq \sum_{m=1}^M m e^{m \beta \mu^\id} (4 \pi m \beta)^{- \frac32} (1 + c e^{- \frac{1}{4 m \beta}}). \nonumber
    \end{align}
    As we will see, terms other than $\sum_{m=1}^M m e^{m \beta \mu^\id} (4 \pi m \beta)^{- \frac32}$ contribute only to errors. 
    
    Let us bound the negative term in the first line of \eqref{eq:label12345}:
    \begin{align}
    \sum_{m=1}^M m e^{m \beta \mu^\id} \leq \sum_{m=1}^\infty m e^{m \beta \mu^\id} = \frac{e^{\beta \mu^\id}}{(1- e^{\beta \mu^\id})^2} = O (\beta^{-2} (- \mu^\id)^{-2}),
    \end{align}
    where we used $1 - e^{\beta \mu^\id} \geq - \beta \mu^\id e^{\beta \mu^\id}$ and $e^{- \beta \mu^\id} = O(1)$. These estimates are good if $\mu^\id$ does not tend to zero, in which case it is better to use $e^{\beta \mu^\id} \leq 1$ and
    \begin{equation}
        \sum_{m=1}^M m = \frac{M(M+1)}{2} = O(\beta^{-2}).
    \end{equation}
    Hence $\sum_{m=1}^M m e^{m \beta \mu^\id}$ is $O(\beta^{-2} (1- \mu^\id)^{-2})$.

    Next, we consider the term
    \begin{equation}
        \beta^{-\frac32} \sum_{m=1}^M m^{-\frac12} e^{m \beta \mu^\id} e^{- \frac{1}{4 m \beta}}. 
         \label{eq:Var_id_smallm_upb_err}
    \end{equation}
    Bounding $e^{m \beta \mu^\id} \leq 1$ we see that the sum is $O(\beta^{-2})$. This is suboptimal if $\mu^\id \to - \infty$. In this case we optimize with respect to $m$ to verify that
    \begin{equation}
        e^{m \beta \mu^\id} e^{- \frac{1}{4 m \beta}} \leq e^{- \sqrt{- \mu^\id}}.
    \end{equation}
    This allows to bound \eqref{eq:Var_id_smallm_upb_err} by
    \begin{equation}
        (4 \pi \beta)^{-\frac32} e^{- \sqrt{-\mu^\id}} \sum_{m=1}^M m^{-\frac12} \leq (4 \pi \beta)^{-\frac32} e^{- \sqrt{-\mu^\id}} \int_0^{\frac{1}{4 \pi \beta}} t^{-\frac12} \d t = \frac{1}{8 \pi^2} \beta^{-2} e^{- \sqrt{-\mu^\id}}.
    \end{equation}

    The above calculations establish that
    \begin{equation}
        \mathbf{Var}^\id[ N_+] = (4 \pi \beta)^{-\frac32} \sum_{m=1}^M m^{-\frac12} e^{m \beta \mu^\id} + O(\beta^{-2} (1 - \mu^\id)^{-2}).  
        \label{eq:Var_Lisum}
    \end{equation}
    The error term does not contribute to the limit both in case (1) and (2). To establish~(1), we compare the summation in \eqref{eq:Var_Lisum} with the definition of the polylogarithm \eqref{eq:polylog_def}:
    \begin{equation}
        \sum_{m=1}^M m^{-\frac12} e^{m \beta \mu^\id} = \mathrm{Li}_{\frac12}(e^{\beta \mu^\id}) +O(e^{\frac{\mu^\id}{4 \pi}}) \to \mathrm{Li}_{\frac12}(z_\infty).
    \end{equation}
    The exponential decay with $\mu^\id$ of the remainder term $\sum_{m=M+1}^\infty m^{-\frac12} e^{m \beta \mu^\id}$ can be concluded by estimating $m^{-\frac12} \leq 1$, computing the resulting sum and using $z_\infty <1$. 
    
    In the case (2), we compare the sum in \eqref{eq:Var_Lisum} with an integral:
    \begin{align}
    \sum_{m=1}^M m^{-\frac12} e^{m \beta \mu^\id} &= \int_0^M t^{-\frac12} e^{\beta \mu^\id t} \d t + O(1) \\
  &  = (- \beta \mu^\id)^{- \frac12} \int_0^{- \frac{\mu^\id}{4 \pi}} s^{-\frac12} e^{-s} \d s = (- \beta \mu^\id)^{-\frac12} (\sqrt{\pi} + O( e^{\frac{\mu^\id}{4 \pi}} )) +O(1). \nonumber
    \end{align}
    
   We have established the claim for $N_+$. Now note that
    \begin{equation}
       \mathbf{Var}^{\mathrm{id}}[\widetilde N] -  \mathbf{Var}^{\mathrm{id}}[N_+] = \frac{e^{\beta \mu^{\mathrm{id}}}}{(1- e^{\beta \mu^{\mathrm{id}}})^2} \leq (- \beta \mu^{\mathrm{id}})^{-2} = o(\mathbf{Var}^\id[N_+]).
   \end{equation}
    \end{proof}

    \begin{prop} \label{prop:localexpvar}
    Suppose that $\sqrt{- \mu^\id} \delta \to \infty$. Then
    \begin{equation}
        \mathbf{E}^\id [ N_D] = N |D| \to \infty, \qquad \mathbf{Var}^\id[ N_D] = \mathbf{Var}^\id[ \widetilde N] |D| (1+o(1)) \to \infty.
    \end{equation}
    Moreover, $\frac{N_D - N |D|}{\mathbf{Var}[ N_D]^{\frac12} }$ converges in distribution to a standard normal variable  as $N \to \infty$. 
    \end{prop}
    \begin{proof}
        We have $\mathbf{E}^\id [\mathcal N_D] = N |D| = N \delta^3 |D_0|$ and $\delta^3 \gg (- \mu^\id)^{- \frac32} \geq c N^{-1}$. Hence $\mathbf{E}^\id [ N_D] \to \infty$. Similarly, by Lemma \ref{lem:id_variance}, $\mathbf{Var}^\id[\widetilde N] |D| \to \infty$. We will now prove that $\mathbf{Var}^\id[ N_D]$ is asymptotic to $\mathbf{Var}^\id[\widetilde N] |D|$. We have the identity
    \begin{align}
        \mathbf{Var}^{\id}[\mathcal N_D] &= N |D| + \int_{D \times D} \gamma^{\id}(x-y)^2 \d x \d y 
    = \mathbf{Var}^{\id}[\widetilde N] |D| - \int_{D \times \Lambda \setminus D} \gamma^{\id}(x-y)^2 \d x \d y, 
    \end{align}
    where 
    \begin{equation}
        \gamma^{\id}(x) = \sum_{p \in \Lambda^*} \frac{e^{\mathrm{i}px}}{e^{\beta (p^2 - \mu^\id)}-1}.
    \end{equation}
    We have to show that, under the running assumptions,
    \begin{equation}
  \frac{1}{\delta^3 \mathbf{Var}^\id[\widetilde N] }  \int_{D \times (\Lambda \setminus D)} \gamma^{\id}(x-y)^2 \d x \d y \to 0.
  \label{eq:formula_67}
    \end{equation}
    Using the bound \eqref{eq:n_function} we can replace the function $\gamma^{\id}$ above with
    \begin{equation}
        \gamma^{\id}_c(x) = \int_{\mathbb R^3} \frac{e^{\mathrm{i}px}}{e^{\beta(p^2 - \mu^\id)}-1} \frac{ \d p}{(2 \pi)^3}.
    \end{equation}
    Now pick $\epsilon >0$ and let
    \begin{equation}
        D^\epsilon := \{ x \in D \, | \, \mathrm{dist}(x,\partial D) \leq \epsilon \}.
    \end{equation}
    We note that $D^\epsilon = \delta \cdot D_0^{ \frac{\epsilon}{\delta}}$. We will choose $\epsilon$ to be $N$-dependent and satisfying $\frac{\epsilon}{\delta} \to 0$, $\sqrt{-\mu^\id} \epsilon \to \infty$. This is possible because $\sqrt{- \mu^\id}\delta \to \infty$. We~have: 
    \begin{equation}
        |D_0^{\frac{\epsilon}{\delta}}| \to |\partial D_0| =0, \quad \text{hence:} \quad |D^\epsilon| = \delta^3 |D_0^{\frac{\epsilon}{\delta}}| = o(\delta^3).
    \end{equation}
    Now we split the integration over $D$ in \eqref{eq:formula_67}, and bound the two terms separately:
    \begin{align}
        \int_{D \times (\Lambda \setminus D)} \gamma^{\id}_c(x-y)^2 \d x \d y &= \int_{D^\epsilon \times (\Lambda \setminus D)} \gamma^{\id}_c(x-y)^2 \d x \d y + \int_{(D \setminus D^\epsilon) \times (\Lambda \setminus D)} \gamma^{\id}_c(x-y)^2 \d x \d y \nonumber \\
        & \leq |D^\epsilon| \int_{\Lambda} \gamma_c(z)^2 dz + |D| \int_{\mathbb R^3 \setminus B(0,\epsilon)} \gamma_c(z)^2 \d z  \\
         & = o(\delta^3 \mathrm{Var}^\id[\mathcal N]) + O(\delta^3 \beta^{-2} (-\mu^\id)^{-\frac12} e^{-2 \sqrt{-\mu^\id} \epsilon})  , \nonumber
    \end{align}
    where we used \eqref{eq:gamma_decay} to bound the integral of $\gamma_c^2$ over the complement of $B(0,\epsilon)$, ball of radius $\epsilon$. Since $\mathbf{Var}^\id[\widetilde N])$ is comparable with $\beta^{-2} (-\mu^\id)^{-\frac12}$ and $e^{-2 \sqrt{-\mu^\id} \epsilon} \to 0$, convergence \eqref{eq:formula_67} is established. 

    To obtain convergence in distribution of $ N_D$, we consider the characteristic function
   \begin{equation}
       \phi(t) = \mathbf{E}^{\mathrm{id}}[\exp \left(it ( N_D - \mathbf{E}^{\mathrm{id}}[N_D]) \right)]
   \end{equation}
   with $t = \frac{s}{\mathbf{Var}^{\id}[ N_D]^\frac12}$ and fixed $s$ in the limit $N \to \infty$. We have
   \begin{align}
       \phi(t) & = \frac{\mathrm{Tr}[\Gamma(e^{\mathrm{i}t \mathbf 1_D}) \Gamma(e^{\beta (\Delta + \mu^{\mathrm{id}})})]}{\mathrm{Tr}[\Gamma(e^{\beta (\Delta + \mu^{\mathrm{id}})})]} e^{- \mathrm{i}t N |D|} = \frac{\det(1 - e^{\beta (\Delta + \mu^{\mathrm{id}})})}{\det(1-e^{\mathrm{i}t \mathbf 1_D} e^{\beta (\Delta + \mu^{\mathrm{id}})}) } e^{- \mathrm{i}t N |D|} \label{eq:ideal_characteristic_func} \\
       &= \det \left( (1-e^{\mathrm{i}t \mathbf 1_D} e^{\beta (\Delta + \mu^{\mathrm{id}})}) (1- e^{\beta (\Delta + \mu^{\mathrm{id}})})^{-1}  \right)^{-1} e^{-\mathrm{i}t N|D|}  = \det \left(1 - A \right)^{-1} e^{-\mathrm{i}t N |D|} , \nonumber
   \end{align}
    where we introduced the operator
    \begin{equation}
        A := (e^{\mathrm{i}t \mathbf 1_D}-1) (e^{ \beta (- \Delta - \mu^\id)}-1)^{-1} = (e^{\mathrm{i}t}-1) \mathbf 1_D (e^{ \beta (- \Delta - \mu^\id)}-1)^{-1}.
    \end{equation}
We note that $\widehat \gamma^{\id} := (e^{ \beta (- \Delta - \mu^\id)}-1)^{-1}$ is a convolution operator with kernel $\gamma^{\id}$. We will need to estimate certain Schatten norms of $A$:
   \begin{equation}
       \| A \|_2^2 = |e^{\mathrm{i}t}-1| \mathrm{Tr}(\mathbf 1_D \widehat \gamma^{\id} \mathbf 1_D \widehat \gamma^{\id})  \leq |t|^2 \mathbf{Var}^{\mathrm{id}}[ N_D], \qquad \| A \|_\infty \leq |t| \| \gamma\|_\infty \leq |t| (- \beta \mu^{\mathrm{id}})^{-1},
   \end{equation}
   and therefore:
   \begin{equation}
       \| A \|_3^3 \leq \| A \|_2^2 \| A \|_\infty \leq |t|^3 \mathbf{Var}^\id[N_D] (- \beta \mu^\id)^{-1} \leq c |s|^3 (\sqrt{- \mu^\id} \delta)^{- \frac32} =o(1).
   \end{equation}
   It follows \cite[Theorem 6.5]{Simon77} that $\mathrm{det}_3(1-A) = 1+o(1)$, and hence
   \begin{equation}
       \phi(t) = \exp \left( \mathrm{Tr}[A] + \frac12 \mathrm{Tr}[A^2] - \mathrm{i}t N |D| \right) (1+o(1)).
   \end{equation}
   Next, we observe that: 
   \begin{subequations}
    \begin{align}
       \mathrm{Tr}[A] &= (e^{\mathrm{i}t}-1) N |D| =  \left(it - \frac{t^2}{2} \right) N |D| +o(1), \\
       \mathrm{Tr}[A^2] & = (e^{\mathrm{i}t}-1)^2 \mathrm{Tr}[(\mathbf 1_D \gamma)^2]   \\
     & =   (e^{\mathrm{i}t}-1)^2 (\mathrm{Var}^{\mathrm{id}}[\mathcal N_D] - N |D|) = -t^2 (\mathrm{Var}^{\mathrm{id}}[\mathcal N_D]- N|D|) + o(1). \nonumber 
    \end{align}
    \end{subequations}
This gives:
   \begin{equation}
       \lim_{N \to \infty} \phi \left( \frac{s}{\mathrm{Var}[\mathcal N_D]^{\frac12}} \right) = e^{- \frac12 s^2}. 
   \end{equation}
   The result follows from Lévy's continuity theorem. 
    \end{proof}

    \begin{remark}
        The limiting random variable in Proposition \ref{prop:localexpvar} can be interpreted as $ \int_{D_{0}} W(x) \d x$, where $W$ is real $L^2(\mathbb R^3)$ Gaussian white noise. This connection could be made explicit by proving Proposition \ref{prop:localexpvar} via an approximation of $\Phi_N$ with white noise on the relevant length scales.  
    \end{remark}

    \begin{prop} \label{prop:micro_limit}
       Suppose that $\delta = N^{-\frac13}$. Then we have the convergence of characteristic functions
         \begin{equation}
\lim_{N \to \infty} \mathbf E[e^{\mathrm{i}t N_D}] =  \det \big(1-(e^{\mathrm{i}t}-1)\mathbf 1_{D_0} ( e^{\beta_\infty(\kappa)(- \Delta_{\mathbb R^3}+m_{\mathrm{eff}})^2} -1 )^{-1} \mathbf 1_{D_0} \big)^{-1} = \mathbf{E}[e^{\mathrm{i}tY}] .
             \label{eq:micro_characteristic_function}
         \end{equation}
         Here $Y$ is a mixed Poisson random variable with intensity $X = \int_{D_0} |H(x)|^2 \d x$, where $H$ is a Gaussian random field on $\mathbb R^3$ with the covariance $\mathbf E[\overline{H(x)}H(y)] = \int_{\mathbb R^3} \frac{e^{\mathrm{i}p(x-y)}}{e^{\beta_\infty(\kappa)(p^2 + m_{\mathrm{eff}}^2)}-1} \frac{\d p}{(2 \pi)^3}$ and vanishing pseudo-covariance, and $m_{\mathrm{eff}}^2  =\lim_{N \to \infty} (- N^{\frac23} \mu^\id)$.
    \end{prop}
    \begin{proof}[Proof of Proposition \ref{prop:micro_limit}]
    We consider the characteristic function (see the discussion around \eqref{eq:ideal_characteristic_func})
\begin{equation}
    \mathbf{E}^\id[e^{\mathrm{i} t  N_D}] = \det(1 - (e^{\mathrm{i}t}-1) B )^{-1}, \qquad B:= \mathbf 1_D (e^{\beta (-\Delta - \mu^{\id})}-1)^{-1} \mathbf 1_D. 
\end{equation}
We introduce the unitary operator
\begin{equation}
    V \ : \ L^2(\Lambda ) \cong L^2(  (-\tfrac12,\tfrac12)^3) \to L^2((-\tfrac12 \delta^{-1} , \tfrac12 \delta^{-1})), \qquad Vf(x) = \delta^{\frac32} f(\delta x),
\end{equation}
and we regard $L^2((-\tfrac12 \delta^{-1} , \tfrac12 \delta^{-1}))$ as a subspace in $L^2(\mathbb R^3)$. Letting $B(x,y)$ be the integral kernel of~$B$, the integral kernel of $VBV^*$ is:
\begin{equation}
    (VBV^*)(x,y) = \delta^3 B(\delta x, \delta y) =  \mathbf 1_{D_0}(x) \mathbf 1_{D_0}(y) \delta^3  \sum_{p \in \Lambda^*} \frac{e^{\mathrm{i}p \delta (x-y)}}{e^{\beta (p^2 - \mu^\id)}-1}.
\end{equation}
Formula \eqref{eq:n_function} in Appendix \ref{app:lemmas} gives us the approximation, uniform in $x,y$,
\begin{align}
    \delta^3  \sum_{p \in \Lambda^*} \frac{e^{\mathrm{i}p \delta (x-y)}}{e^{\beta (p^2 - \mu^\id)}-1} &= \delta^3 \int_{\mathbb R^3} \frac{e^{\mathrm{i}p \delta(x-y)}}{e^{\beta (p^2 - \mu^\id)}-1} \frac{\d p}{(2 \pi)^3} + o(1) \\
    & = \int_{\mathbb R^3} \frac{e^{\mathrm{i}q(x-y)}}{e^{\beta_\infty(\kappa)(q^2 + m_{\mathrm{eff}}^2)}-1} \frac{\d q}{(2 \pi)^3} + o(1) 
     = (e^{\beta_\infty(\kappa)(-\Delta_{\mathbb R^3} +m_{\mathrm{eff}}^2)}-1)^{-1}(x,y) + o(1), \nonumber
\end{align}
where in the second equality we performed a change of variables $p = \delta^{-1} q$ and used the convergences $\beta N^{\frac23} \to \beta_\infty(\kappa)$ and $- \beta \mu^\id \to m_{\mathrm{eff}}^2$. Since the integral kernel of $VBV^*$ converges uniformly to that of $\mathbf 1_{D_0} (e^{\beta_\infty(\kappa)(-\Delta_{\mathbb R^3} +m_{\mathrm{eff}}^2)}-1)^{-1} \mathbf 1_{D_0}$ and is supported on a bounded set, we have 
\begin{equation}
 VBV^* \to \mathbf 1_{D_0} (e^{\beta_\infty(\kappa)(-\Delta_{\mathbb R^3} +m_{\mathrm{eff}}^2)}-1)^{-1} \mathbf 1_{D_0}
 \label{eq:tralala}
\end{equation}
in the Hilbert-Schmidt space. One also verifies that $  \mathrm{Tr}[B] \longrightarrow \mathrm{Tr}[ \mathbf 1_{D_0} (e^{\beta_\infty(\kappa)(-\Delta_{\mathbb R^3} +m_\infty^2)}-1)^{-1} \mathbf 1_{D_0}]  $, and~therefore the convergence in \eqref{eq:tralala} holds also in the trace norm. Therefore,
\begin{equation}
    \mathrm{det}(1- (e^{\mathrm{i}t}-1)B) \longrightarrow \mathrm{det}(1- (e^{\mathrm{i}t}-1) \mathbf 1_{D_0} (e^{\beta_\infty(\kappa)(-\Delta_{\mathbb R^3} +m_{\mathrm{eff}}^2)}-1)^{-1} \mathbf 1_{D_0}). 
\end{equation}
To identify the determinant as a characteristic function of an explicit random variable $Y$, we consider its spectral product representation and compare with Lemma \ref{lem:mixed_Poisson}. We conclude that $Y$ is mixed Poisson with intensity
         \(
             X = \sum_{n=1}^\infty \lambda_n X_n,
         \)
         where $X_n$ are independent exponential random variables with rate parameter $1$ and $\lambda_1 \leq \lambda_2 \leq \dots$ are the eigenvalues of $\mathbf 1_{D_0} ( e^{\beta_\infty(\kappa)(- \Delta_{\mathbb R^3}+m_{\mathrm{eff}}^2)} -1 )^{-1} \mathbf 1_{D_0}$. Then we observe that $X \stackrel{d}{=} \int_{D_0} |H(x)|^2 \d x$. 
\end{proof}

   \begin{prop} \label{prop:special_limit}
       Suppose that $\beta^{-\frac12} \delta \to \infty$ and $\sqrt{- \mu^\id} \delta$ is bounded. Let $m_{\mathrm{eff}}^2 = \lim_{N \to \infty} (- \delta^2 \mu^\id)$. Then
       \begin{equation}
          \frac{1}{N \delta^3} \mathbf{E}^\id[ N_D] =|D_0| , \qquad \frac{1}{\beta^{-2} \delta^4} \mathbf{Var}^\id[N_D] \to  \int_{D_0 \times D_0} \frac{e^{-2 m_{\mathrm{eff}} |x-y|}}{16 \pi^2 |x-y|^2} \d x  \d y,
       \end{equation}
       and we have the convergence of characteristic functions
         \begin{equation}
            \lim_{N \to \infty} \mathbf{E}^\id \left[\exp \left(\mathrm{i}t \frac{N_D-N|D|}{\beta^{-1} \delta^2 } \right) \right] = \mathrm{det}_2 \Big( 1-\mathrm{i}t \mathbf 1_{D_0} ( -\Delta_{\mathbb R^3} + m_{\mathrm{eff}}^2 )^{-1} \mathbf 1_{D_0} \Big)^{-1} = \mathbf E[e^{\mathrm{i}tX}].
             \label{eq:special_characteristic_function}
         \end{equation}
         Here $X = \sum_{n=1}^\infty \lambda_n (X_n-1)$, where $X_n$ are independent exponential random variables with rate parameter $1$ and $\lambda_1 \leq \lambda_2 \leq \dots$ are the eigenvalues of $\mathbf 1_{D_0} ( - \Delta_{\mathbb R^3}+m_{\mathrm{eff}}^2 )^{-1} \mathbf 1_{D_0}$. 
    \end{prop}
    \begin{proof}
The proof is analogous to that of Proposition \ref{prop:micro_limit}. We only highlight the new elements. Recall
\begin{equation}
    \gamma^{\id}_c(x) = \int_{\mathbb R^3} \frac{e^{\mathrm{i}px}}{e^{\beta(p^2 - \mu^\id)}-1} \frac{\d p}{(2 \pi)^3} .
\end{equation}
We have to show that:
\begin{equation}
    \beta \delta \gamma^{\id}_c(\delta x) \longrightarrow \frac{e^{- m_{\mathrm{eff}} |x|}}{4 \pi |x|}
\end{equation}
in $L^2$ sense on compact subsets of $\mathbb R^3$. If we restrict to $|x| \geq \epsilon$ with $N$-dependent cutoff $\epsilon$, the~desired convergence follows from the pointwise approximation in \eqref{eq:gamma_decay} in Appendix \ref{app:lemmas}, provided that $N \epsilon^3 \to \infty$. We will choose $\epsilon$ so that we also have
\begin{equation}
    \int_{|x| \leq \epsilon} |\beta \delta \gamma_c(\delta x)|^2 \d x  \to 0, \qquad \int_{|x| \leq \epsilon} \left| \frac{e^{- m_\infty |x|}}{4 \pi |x|} \right|^2 \d x  \to 0.
\label{eq:core_vanishes_in_L2}
\end{equation}
The first integral is bounded by, up to constants,
\begin{equation}
    \beta^2 \delta^2 \gamma_c(0)^2 \epsilon^3 \leq  c N^{\frac23} \delta^2 \epsilon^3. 
    \end{equation}
Therefore, we need to take $\epsilon^3 \ll N^{-\frac23} \delta^{-2}$. This is compatible with the condition $N \epsilon^3 \gg 1$, because
\begin{equation}
    \frac{N^{-\frac23} \delta^{-2}}{N^{-1}} = N^{\frac13} \delta^{-2} \to \infty.
\end{equation}
We note that since $\delta \gg N^{-\frac13}$, $\epsilon^3 \ll N^{-\frac23} \delta^{-2}$ implies that $\epsilon \to 0$, so the second convergence in \eqref{eq:core_vanishes_in_L2} also holds.
\end{proof}

\begin{remark}
    The random variable $X$ in Proposition \ref{prop:special_limit} can be interpreted as $\int_{D_0} :|G(x)|^2: \d x$, where $G$ is a Gaussian random field on $\mathbb R^3$ with covariance $\mathbf E[\overline{G(x)} G(y)]=\frac{e^{- m_{\mathrm{eff}}|x-y|}}{4 \pi|x-y|}$ and vanishing pseudo-covariance, and $:|G(x)|^2:$ is its Wick square.
\end{remark}

\vspace{0.5cm}

    \textbf{Acknowledgments.} A. D. gratefully acknowledges support through the NSF Grant DMS-2555747. The work of M. N. and B. R. was supported by the National Science Centre (NCN) grant Sonata Bis 13
(project number 2023/50/E/ST1/00439).
    
\vspace{0.5cm}

    \textbf{AI Disclosure.} Language models GPT-5.6 Sol and Gemini 3.1 Pro have been used to improve the presentation, to help us learn mathematical concepts, to assist with computations, and in proofreading. All main ideas, mathematical formulations, statements of theorems, and the proofs were conceived and developed entirely by the authors.

\vspace{0.5cm}   

    \textbf{Data availability statement.} No datasets were generated or analysed during the current study.
\appendix
\section{Properties of the point process associated with the Gibbs state}
\label{app:GibbsPointProcess}
In this section, we discuss properties of the point process $\Theta_{\beta,N}$ defined in \eqref{eq:GibbsPointProcess} and its relation to the Gibbs state $G_{\beta,N}$ defined in \eqref{eq:GibbsDistributionDiscrete}. We assume that the reader is familiar with the notation introduced in Section~\ref{sec:Gibbs_point_process}. Before stating the main result of this section, we introduce some additional notation.

Let $f:\Lambda\to\mathbb R$ be a bounded Borel-measurable function. The \textbf{linear statistic} associated with $\Theta_{\beta,N}$ is defined by
\begin{equation}
    \Theta_{\beta,N}(f)
    =
    \sum_{j=1}^{M} f(X_j^{(M)}).
    \label{eq:linearStatistics}
\end{equation}
Since $G_{\beta,N}$ is translation invariant, we have
\begin{equation}
    \mathbf E[ \Theta_{\beta,N}(f) ] = N \int_{\Lambda} f(x) \mathrm{d}x,
    \label{eq:intensity_measure_is_Lebesgue}
\end{equation}
which allows to define the random variable $\Theta_{\beta,N}(f)$ for every $f \in L^1(\Lambda)$ by continuous linear extension.

For any $k \in\mathbb N$ and any family of mutually disjoint Borel sets $D_1,\ldots,D_k\subseteq\Lambda$, the joint distribution of the numbers of particles in these sets is given by
\begin{equation}
    (M_{D_1},\ldots,M_{D_k}) = \bigl(
    \Theta_{\beta,N}(\mathbf 1_{D_1}),\ldots,\Theta_{\beta,N}(\mathbf 1_{D_k}) \bigr).
    \label{eq:countingStatistics}
\end{equation}
The correlation functions (or factorial moment densities)
$m^{(k)}(x_1,\ldots,x_k)$ of $\Theta_{\beta,N}$ are defined by 
\begin{equation}
    \int_{\Lambda^k} f(x_1,\ldots,x_k) m^{(k)}(x_1,\ldots,x_k) \d(x_1,\ldots,x_k) = \sum_{n=k}^{\infty} p_n \sum_{j_1,\ldots,j_k}^{\neq} \mathbf{E}\left[ f(X_{j_1}^{(n)},\ldots,X_{j_k}^{(n)}) \right],
    \label{eq:correlationFunctions}
\end{equation}
for every bounded Borel-measurable function
$f:\Lambda^k\to\mathbb R$. The symbol $\sum^{\neq}_{j_1,\ldots,j_k}$ means summation over all distinct indices.

The second quantization $\d\Gamma(h)$ of a bounded operator $h$ on $L^2(\Lambda)$ is defined by
\begin{equation}
    \d\Gamma(h)\big|_{\mathscr{F}_n}
    = \sum_{j=1}^n h_j.
    \label{eq:secondQuantization}
\end{equation}
Here, $\mathscr{F}_n$ denotes the $n$-particle sector of the Fock space, and $h_j$ acts on the $j$th particle coordinate. Let $a_x = \sum_{p \in \Lambda^*} e^{\mathrm{i} p \cdot x} \a_p$ denote the operator (more precisely, an operator-valued distribution) annihilating a particle at $x \in \Lambda$, and let $a_x^*$ be its adjoint. For $k \in \mathbb{N}$, the $k$-particle density of the Gibbs state is given by
\begin{equation}
    \varrho^{(k)}(x_1,\ldots,x_k) = \Tr[ a_{x_1}^* \ldots a_{x_k}^* \a_{x_k} \ldots \a_{x_1} G_{\beta,N}].
    \label{eq:kParticleDensity}
\end{equation}

Finally, by a slight abuse of notation, we use the same symbol $f$ to denote both a function and the corresponding multiplication operator.

The following proposition makes precise the relationship between the Gibbs state and its associated point process.

\begin{prop}
\label{prop:GibbsProcess}
The following statements hold for the point process
$\Theta_{\beta,N}$ defined in \eqref{eq:GibbsPointProcess}
and the Gibbs state $G_{\beta,N}$ defined in
\eqref{eq:GibbsState}.
    \begin{enumerate}[label=(\alph*)]
        \item Let $F : \mathbb{R} \to \mathbb{C}$ be a bounded Borel-measurable function. Then
    \begin{equation}
        \mathbf{E}[ F(\Theta_{\beta,N}(f)) ] = \Tr[ F(\d \Gamma(f) ) G_{\beta,N}].
        \label{eq:prop:empiricalMeasure1}
    \end{equation}
    \item Let $k \in \mathbb{N}$, let $D_1,\ldots,D_k \subseteq \Lambda$ be mutually disjoint Borel sets, and let $B_1,\ldots,B_k \subseteq \mathbb{N}_0$. Then:
    \begin{equation}
        \mathbf{P}( M_{D_1} \in B_1, \ldots, M_{D_k} \in B_k ) = \Tr\left[ \prod_{j=1}^k \mathds{1}(\d \Gamma(\mathds{1}_{D_j}) \in B_j ) G_{\beta,N} \right].
        \label{eq:prop:empiricalMeasure1b}
    \end{equation}
    \item For any $k \in \mathbb{N}$ we have
    \begin{equation}
        m^{(k)}(x_1,\ldots,x_k) = \varrho^{(k)}(x_1,\ldots,x_k).
        \label{eq:prop:empiricalMeasure1c}
    \end{equation}
    \end{enumerate}
\end{prop}

\begin{proof}
The right-hand side of \eqref{eq:prop:empiricalMeasure1} can be written as
\begin{equation}
    \mathbf{E}[ F(\Theta_{\beta,N}(f)) ] = \mathbf{E}[ \mathbf{E} [ F(\Theta_{\beta,N}(f)) | M ] ].
    \label{eq:prop:empiricalMeasure2}
\end{equation}
Then, for every $n \in \mathbb{N}_0$, \eqref{eq:GibbsPointProcess} allows us to write 
\begin{equation}
    \mathbf{E} [ F(\Theta_{\beta,N}(f)) | M = n ] = \int_{\Lambda^n} F \left( \sum_{j=1}^n f(x_j) \right) G_n(X_n;X_n) \d X_n.
    \label{eq:prop:empiricalMeasure3}
\end{equation}
On the right-hand side we used the notation $X_n = (x_1,\ldots,x_n)$ and $\d X_n = \d(x_1,\ldots,x_n)$. For $n=0$ it is interpreted as $F(0)$. We remark that the restriction of the integral kernel of $G_n$ to the diagonal is a well-defined function in $L^1(\Lambda^n)$ because of the spectral representation
\begin{equation}
    G_n(X_n;X_n') = \sum_{m=0}^\infty \lambda_{n,m} \psi_{n,m}(X_{n}) \overline{\psi_{n,m}(X_n')},
\end{equation}
where $\lambda_{n,m}$ are the eigenvalues of $G_n$, which satisfy $\lambda_{n,m} \geq 0$ and $\sum_{m=0}^\infty \lambda_{n,m} =1$, and $\psi_{n,m}$ are the corresponding eigenvectors. Coming back to \eqref{eq:prop:empiricalMeasure3},
\begin{align}
    \mathbf{E}[ F(\Theta_{\beta,N}(f)) ] = \sum_{n=0}^{\infty} p_n \int_{\Lambda^n} F \left( \sum_{j=1}^n f(x_j) \right) G_n(X_n;X_n) \d X_n.
    \label{eq:prop:empiricalMeasure3b}
\end{align}
Since $\d \Gamma(f)$ acts as $\sum_{j=1}^n f(x_j)$ on $\mathscr{F}_n$ and $G_{\beta,N} = \sum_{n} p_n G_n$, the right-hand side of \eqref{eq:prop:empiricalMeasure3b} equals the right-hand side of \eqref{eq:prop:empiricalMeasure1}. This proves part~(a).

To prove part~(b), let $t_1,\ldots,t_k\in\mathbb R$ and apply
part~(a) with
\begin{equation}
    f(x)=\sum_{j=1}^k t_j \mathbf 1_{D_j}(x),
    \qquad
    F(x)=e^{\mathrm i x}.
    \label{eq:prop:empiricalMeasure4}
\end{equation}
In this case, the left-hand side of \eqref{eq:prop:empiricalMeasure1} equals
\begin{equation}
    \mathbf E\left[ \exp\left( \mathrm i \sum_{j=1}^k t_j M_{D_j} \right) \right].
    \label{eq:prop:empiricalMeasure5}
\end{equation}
Since the operators $\mathrm d\Gamma(\mathbf 1_{D_j})$, $j=1,\ldots,k$ commute with one another, the right-hand side of \eqref{eq:prop:empiricalMeasure1} is given by
\begin{equation}
    \Tr \left[ \prod_{j=1}^k
    e^{\mathrm i t_j \mathrm d\Gamma(\mathbf 1_{D_j})}
    G_{\beta,N} \right].
    \label{eq:prop:empiricalMeasure6}
\end{equation}
In combination, \eqref{eq:prop:empiricalMeasure1}, \eqref{eq:prop:empiricalMeasure5}, and \eqref{eq:prop:empiricalMeasure6} show that the random variables defined by the left- and the right-hand side of \eqref{eq:prop:empiricalMeasure1b} have the same characteristic function, and therefore the same law. It remains to establish part~(c).

To that end, we introduce the $k$-th factorial power of $n$ by
\begin{equation}
    n^{[k]} = \begin{cases}
        n(n-1) \ldots (n-k+1) & \text{ if } k \in \{0,\ldots,n \} \\
        0 & \text{ if } k > n,
    \end{cases}
    \label{eq:prop:empiricalMeasure7}
\end{equation}
and note that, for every bounded Borel-measurable function
$f:\Lambda^k\to\mathbb R$,
\begin{align}
    \sum_{n=k}^{\infty} p_n \sum_{j_1,\ldots,j_k}^{\neq} \mathbf{E}\left[ f(X_{j_1}^{(n)},\ldots,X_{j_k}^{(n)}) \right] &= \sum_{n=k}^{\infty} p_n \sum_{j_1,\ldots,j_k}^{\neq} \int_{\Lambda^n} f(x_{j_1},\ldots,x_{j_k}) G_n(X_n;X_n) \d X_n \nonumber \\
    &= \sum_{n=k}^{\infty} p_n n^{[k]} \int_{\Lambda^n} f(x_1,\ldots,x_k) G_n(X_n;X_n) \d X_n \nonumber \\
    &= \int_{\Lambda^k} f(x_1,\ldots,x_k) \Tr[ a_{x_1}^* \ldots a_{x_k}^* a_{x_k} \ldots a_{x_1} G_{\beta,N} ].
    \label{eq:prop:empiricalMeasure8}
\end{align}
To come to the second line, we used the symmetry of the integral kernel $G_n(X_n;X_n)$ under permutations of the coordinates $x_1,\ldots,x_n$. This proves part~(c) and completes the proof of Proposition~\ref{prop:GibbsProcess}.
\end{proof}

\section{Covariance of ideal gas on a torus} \label{app:lemmas}

The Jacobi theta function is a function of the complex variables $x, \tau$, with $\tau$ in the upper half-plane, defined as 
\begin{equation}
    \vartheta (x ; \tau ) = \sum_{n=-\infty}^\infty e^{\mathrm{i} \pi \tau n^2 + 2 \pi \mathrm{i} n x}.
    \label{eq:gaussian_sum}
\end{equation}
Here we consider only $x \in \mathbb R$ and $\tau = \mathrm{i} \lambda$, $\lambda >0$. For $x=0$ we have the following asymptotics in $\lambda$: 
\begin{equation}
  \vartheta (0 ; \mathrm{i} \lambda )\sim \begin{cases} 1 , & \lambda \to \infty, \\ 
    \lambda^{-\frac12}, & \lambda \to 0. \end{cases}
    \label{eq:gaussian_asymptotics}
\end{equation}
$\vartheta (x ; i \lambda )$ is real-valued and periodic in $x$ with period $1$. If $\lambda$ is small, $\vartheta (x ; i \lambda )$ exhibits Gaussian decay in $x$ away from its maxima at integers:
\begin{equation}
    \vartheta(x; i \lambda) \sim \lambda^{-\frac12} e^{- \frac{\pi x^2}{\lambda}}, \qquad |x| \leq \frac12.
\end{equation}
For large $\lambda$ it is essentially independent of $x$. The lemma below states these asymptotics with explicit error bounds.

\begin{lem} \label{lem:gaussian}
\begin{enumerate}
    \item For $x =0$ we have the bounds:
    \begin{subequations}
    \begin{align}
        0 \leq & \vartheta(0; \mathrm{i} \lambda) -1 \leq  e^{- \pi \lambda} (2 + \lambda^{- \frac12}), \\
        0 \leq  & \vartheta(0; \mathrm{i} \lambda) - \lambda^{-\frac12} \leq  e^{- \frac{\pi}{\lambda}} (2\lambda^{-\frac12} + 1).    
    \end{align}
    \end{subequations}
    \item For $|x| \leq \frac12$ we have:
    \begin{subequations}
    \begin{align} \label{eq:thetabounds1}
        0 \leq & \vartheta(x ; \mathrm{i} \lambda) \leq \vartheta(0 ; \mathrm{i} \lambda), \\ \label{eq:thetabounds2}
         | & \vartheta(x;\mathrm{i} \lambda)-1| \leq e^{- \pi \lambda} (2 + \lambda^{-\frac12}), \\ \label{eq:thetabounds3}
        0 \leq & \vartheta(x ; \mathrm{i} \lambda ) - \lambda^{-\frac12} e^{- \frac{\pi x^2}{\lambda}} \leq e^{- \frac{\pi}{4 \lambda}} (2 \lambda^{-\frac12}+1). 
    \end{align}
    \end{subequations}
\end{enumerate}
\end{lem}
\begin{proof}
(1) We can lower bound $\vartheta (0 ; \mathrm{i} \lambda )$ by the $n=0$ term in \eqref{eq:gaussian_sum}, $\vartheta (0 ; \mathrm{i} \lambda ) \geq 1$. Secondly,
\begin{equation}
    \vartheta (0 ; \mathrm{i} \lambda )-1 =2 \sum_{n=1}^\infty e^{- \pi \lambda n^2} \leq 2\int_0^\infty e^{- \pi \lambda x^2} \d x  = \lambda^{-\frac12},
    \label{eq:gaussian_upper}
\end{equation}
where we used monotonicity of the Gaussian function on $[0,\infty)$ to bound the sum by an integral. We~can improve this bound as follows:
\begin{equation}
    \vartheta (0 ; \mathrm{i} \lambda )-1 = 2 \sum_{n=0}^\infty e^{-\pi \lambda(n+1)^2} \leq 2e^{-\pi \lambda} \sum_{n=0}^\infty e^{- \pi\lambda n^2} \leq e^{-\pi \lambda} \left( 2 + \lambda^{-\frac12} \right).
    \label{eq:large_lambda_upper_bound}
 \end{equation}
 
Since Fourier transforms of Gaussians are Gaussian, the Poisson summation formula relates $\vartheta (0 ; i \lambda )$ to itself evaluated at a different argument:
\begin{equation}
    \vartheta (0 ; \mathrm{i} \lambda ) = \lambda^{-\frac12}\vartheta (0 ; \mathrm{i} \lambda^{-1} ).
\end{equation}
Combining this identity with lower and upper bounds derived above, we get further lower and upper bounds, useful for $\lambda \to 0$. 

(2) $|\vartheta (x ; \mathrm{i} \lambda )| \leq \vartheta (0 ; \mathrm{i} \lambda )$ follows from the triangle inequality. The bound on $|\vartheta (0 ; \mathrm{i} \lambda )-1|$ is obtained by extracting the $n=0$ term of \eqref{eq:gaussian_sum} and estimating the remainder as in \eqref{eq:large_lambda_upper_bound}. By Poisson summation,
\begin{equation}
    \vartheta (x ; \mathrm{i} \lambda ) = \lambda^{-\frac12} \sum_{n=-\infty}^\infty e^{- \frac{\pi (n+x)^2}{\lambda}} . 
\end{equation}
This series is positive and lower-bounded by the $n=0$ term. We bound the sum of other terms as in the proof of (1).
\end{proof}

The polylogarithm function $\mathrm{Li}_s(z)$, defined in \eqref{eq:polylog_def}, is a convergent power series in the open unit disc for every $s \in \mathbb C$. If $\mathrm{Re}(s) > 1$, it converges for $z=1$ and $\mathrm{Li}_s(1)=\zeta(s)$. We restrict attention to $s \in \mathbb R$ and $z \in [0,1]$. Then $\mathrm{Li}_s(z)$ is an increasing function of $z$. We will need the behavior of $\mathrm{Li}_{\frac32}(z)$ near $1$.

\begin{lem} \label{lem:Li_Holder}
Suppose that $s \in (1,2)$. Then for $z \in [0,1]$,
\begin{equation}
    0 \leq \zeta(s) - \mathrm{Li}_{s}(z) \leq (1-z)^{s-1} \left( \frac{1}{2-s} + \frac{1}{1-s} z^{1-s} \right).
\end{equation}
\end{lem}
\begin{proof}
    In the sum \eqref{eq:polylog_def} we use $1-z^m \leq m (1-z)$ for $m \leq (1-z)^{-1}$ and $1-z^m \leq 1$ for $m > (1-z)^{-1}$. Letting $M = \lfloor (1-z)^{-1} \rfloor$,
    \begin{align}
    0 & \leq \zeta(s) - \mathrm{Li}_s(z)  \leq (1-z )\sum_{m=1}^M m^{-s+1} + \sum_{m=M+1}^\infty m^{-s} \\
    & \leq (1-z) \int_0^M x^{-s+1} \d x  + \int_M^{\infty} x^{-s} \d x   = (1-z) \frac{M^{2-s}}{2-s} + \frac{M^{1-s}}{s-1} \leq (1-z)^{s-1} \left( \frac{1}{2-s} + \frac{z^{1-s}}{s-1} \right).  \nonumber 
\end{align}
\end{proof}

The following lemma discusses the  properties of the kernel of the one-body reduced density matrix of an ideal gas.

\begin{lem} \label{lem:Nid_func}
Let $\beta \in (0, 2 \pi]$, $\mu <0$, $x \in (-\tfrac12, \tfrac12)^3$. There exists $c>0$ such that:
\begin{equation}
       \sum_{p\in 2 \pi \mathbb Z^3} \frac{e^{\mathrm{i}px}}{e^{\beta (p^2-\mu)} -1} = \int_{\mathbb R^3} \frac{e^{\mathrm{i}px}}{e^{\beta (p^2 - \mu)}-1} \frac{d p}{(2 \pi)^3}+ \frac{e^{\frac{\mu}{4 \pi}}}{e^{-\beta \mu}-1} + O(\beta^{-1} e^{ - c \sqrt{- \mu}} ) .
        \label{eq:n_function}
    \end{equation}
At $x=0$, the integral over $\mathbb R^3$ can be evaluated in terms of the polylogarithm:
    \begin{equation}
   \int_{\mathbb R^3} \frac{1}{e^{\beta (p^2 - \mu)}-1} \frac{d p}{(2 \pi)^3}   =  \frac{\mathrm{Li}_{\frac32}(e^{\beta \mu})}{(4 \pi \beta)^{\frac32}}.
   \label{eq:gamma_polylog}
    \end{equation}
   Moreover, for nonzero $x \in \mathbb R^3$ the integral can be approximated as follows: 
   \begin{equation}
\left| \int_{\mathbb R^3} \frac{e^{\mathrm{i}px}}{e^{\beta (p^2 - \mu)}-1} \frac{d p}{(2 \pi)^3} - \frac{e^{- \sqrt{- \mu} |x|}}{4 \pi \beta |x|}  \right| \leq \frac{e^{- \sqrt{-\mu} |x|}}{4 \pi |x|} \left( - \mu + \frac{3 \sqrt{-\mu}}{|x|} + \frac{3}{|x|^2} \right).
          \label{eq:gamma_decay}
   \end{equation}
\end{lem}
In \eqref{eq:gamma_decay} the right hand side is a remainder in the approximation of the integral by $\frac{e^{-\sqrt{-\mu}|x|}}{4 \pi \beta |x|}$. We~note that this remainder is much smaller than the main term if $\beta \mu$ is small and $|x|$ is much larger than $\beta^{\frac12}$. 
\begin{proof}[Proof of Lemma \ref{lem:Nid_func}]
    We expand the summand in geometric series and interchange the sums:
\begin{equation}
    \sum_{p\in 2 \pi \mathbb Z^3} \frac{e^{\mathrm{i}px}}{e^{\beta (p^2-\mu)} -1} = \sum_{m=1}^\infty e^{m \beta \mu}  \sum_{p \in 2 \pi \mathbb Z^3} e^{-m \beta p^2} e^{\mathrm{i}px} = \sum_{m=1}^\infty e^{m \beta \mu} \prod_{j=1}^3 \vartheta (x_j ; 4 \pi \mathrm{i} m \beta)^3.
    \label{eq:n_geometric}
\end{equation}
We use the asymptotics of the $\vartheta$ function in Lemma \ref{lem:gaussian}. Let $M= \lfloor (4 \pi \beta)^{-1} \rfloor$. We have $M = O(\beta^{-1})$, $M^{-1} = O(\beta)$. We split the summation in \eqref{eq:n_geometric} into parts $m \leq M$ and $m > M$. We~consider the small $m$ contribution first:
\begin{equation}
    \sum_{m=1}^M e^{m \beta \mu} \prod_{j=1}^3\vartheta (x_j ; 4 \pi \mathrm{i} m \beta)^3 =  \sum_{m=1}^M (4 \pi m \beta)^{-\frac32} e^{m \beta \mu} \left( e^{-\frac{x^2}{4m \beta}} + O (e^{- \frac{1}{16 m \beta}}) \right).
    \label{eq:n_sum_small_m}
\end{equation}
An analogous manipulation for the integral in \eqref{eq:n_function} gives:
\begin{equation}
    \int_{\mathbb R^3} \frac{e^{\mathrm{i}px}}{e^{\beta (p^2 - \mu)}-1} \frac{d p}{(2 \pi)^3} = \sum_{m=1}^\infty e^{m \beta \mu} (4 \pi m \beta)^{- \frac32} e^{- \frac{x^2}{4 m \beta}}.
    \label{eq:cont_gamma_int_sum}
\end{equation}
Comparing the integral above with \eqref{eq:n_sum_small_m}, we find:
\begin{equation}
    \sum_{m=1}^M  (4 \pi m \beta)^{-\frac32} e^{m \beta \mu}  e^{-\frac{x^2}{4m \beta}}  = \int_{\mathbb R^3} \frac{e^{\mathrm{i}px}}{e^{\beta (p^2 - \mu)}-1} \frac{d p}{(2 \pi)^3} - \sum_{m=M+1}^\infty  (4 \pi m \beta)^{-\frac32} e^{m \beta \mu}  e^{-\frac{x^2}{4m \beta}}.
\end{equation}
We estimate the remainder:
\begin{equation}
    \sum_{m=M+1}^\infty (4 \pi m \beta)^{-\frac32} e^{m \beta \mu}  e^{-\frac{x^2}{4m \beta}} \leq c \, \beta^{-\frac32} e^{\frac{\mu}{4 \pi}} \int_M^\infty t^{- \frac32} \d t  = O ( \beta^{-1} e^{\frac{\mu}{4 \pi}} ).
\end{equation}
Now consider the exponential error term. Optimizing with respect to $m$ we obtain the inequality
\begin{equation}
    e^{m \beta \mu} e^{- \frac{1}{32 m \beta}} \leq e^{-c \sqrt{-\mu}}
\end{equation}
for some constant $c>0$. Therefore,
\begin{align}
    \sum_{m=1}^M (m \beta)^{-\frac32} e^{m \beta \mu} e^{- \frac{1}{16m \beta}} &\leq e^{- c \sqrt{- \mu}} \sum_{m=1}^M  (m \beta)^{- \frac32} e^{- \frac{1}{32 m \beta}}  \\
   & = \beta^{-1} e^{- c\sqrt{- \mu}} \left( \int_0^{\frac{1}{4 \pi}} x^{- \frac32} e^{-\frac{1}{32x}} \d x  + o(1) \right) = O(\beta^{-1} e^{- c\sqrt{- \mu}}), \nonumber
\end{align}
where $o(1)$ in the parenthesis refers to the limit $\beta \to 0$. We established
\begin{equation}
    \sum_{m=1}^M e^{m \beta \mu} \prod_{j=1}^3 \vartheta (x_j ; 4 \pi \mathrm{i} m \beta)^3 = \int_{\mathbb R^3} \frac{e^{\mathrm{i}px}}{e^{\beta (p^2 - \mu)}-1} \frac{\d p}{(2 \pi)^3} + O(\beta^{-1} e^{- c \sqrt{- \mu}}).
\end{equation}

Next, we consider the large $m$ contribution. Using \eqref{eq:thetabounds2}, we have
\begin{equation}
    \sum_{m=M+1}^\infty e^{m \beta \mu} \prod_{i=1}^3 \vartheta (x_j ; 4 \pi \mathrm{i} m \beta)^3 = \sum_{m=M+1}^\infty e^{m \beta \mu} \left( 1 + O ( e^{- 4 \pi^2 m \beta} ) \right).
\end{equation}
First we note
\begin{equation}
\sum_{m=M+1}^\infty e^{m \beta \mu} = \frac{e^{M \beta \mu}}{e^{- \beta \mu}-1} = \frac{e^{\frac{ \mu}{4 \pi}}}{e^{- \beta \mu}-1} +O(e^{\frac{ \mu}{4 \pi}}).
\end{equation}
Secondly,
\begin{equation}
    \sum_{m=M+1}^\infty z^m e^{- 4 \pi^2 m \beta} \leq e^{\frac{\mu}{4 \pi}} \sum_{m=M+1}^\infty  e^{- 4 \pi^2 m \beta}= O (\beta^{-1} e^{\frac{\mu}{4 \pi}}).
\end{equation}
Combining the obtained bounds gives \eqref{eq:n_function}.

Now consider the identity \eqref{eq:cont_gamma_int_sum}. Evaluating the right hand side at $x=0$ yields \eqref{eq:gamma_polylog}. To obtain \eqref{eq:gamma_decay}, we compare with the integral
\begin{equation}
\int_0^\infty (4 \pi t \beta)^{-\frac32} e^{t \beta \mu} e^{- \frac{x^2}{4 t \beta}} \d t = \frac{e^{- \sqrt{- \mu} |x|}}{4 \pi \beta |x|}.
\label{eq:BesselK12_formula}
\end{equation}
Therefore, putting $f(t) = (4 \pi t \beta)^{-\frac32} e^{t \beta \mu} e^{- \frac{x^2}{4 t \beta}}$:
\begin{align}
    \int_{\mathbb R^3} \frac{e^{\mathrm{i}px}}{e^{\beta (p^2 - \mu)}-1} \frac{d p}{(2 \pi)^3} - \frac{e^{- \sqrt{- \mu} |x|}}{4 \pi \beta |x|} &= \sum_{m=1}^\infty \int_{m-1}^m (f(m)-f(t)) \d t \\
    & = \sum_{m=1}^\infty \int_{m-1}^m \int_{t}^m f'(s) \d s \d t \nonumber \\
    &= \sum_{m=1}^\infty \int_{m-1}^m (s-m+1) f'(s) ds. \nonumber
\end{align} 
This expression can be upper-bounded and lower-bounded with almost the same calculations. We~present the proof of the upper bound. We have
\begin{equation}
    f'(s) = \left( - \tfrac32 s^{-1} + \beta \mu + \frac{x^2}{4 s^2 \beta}  \right) f(s) \leq \frac{x^2}{4 s^2 \beta} f(s),
\end{equation}
and, therefore:
\begin{equation}
    \sum_{m=1}^\infty \int_{m-1}^m (s-m+1) f'(s) \d s \leq \int_0^\infty \frac{x^2}{4 s^2 \beta} f(s) \d s = \frac{e^{- \sqrt{-\mu} |x|}}{4 \pi |x|} \left( - \mu + \frac{3 \sqrt{-\mu}}{|x|} + \frac{3}{|x|^2} \right).
\end{equation}
\end{proof}

We remark that \eqref{eq:gamma_decay} features the function
\begin{equation}
    \int_{\mathbb R^3} \frac{e^{\mathrm{i}px}}{\beta(p^2 - \mu)} \frac{d p }{(2 \pi)^3 } = \frac{e^{- \sqrt{-\mu} |x|}}{4 \pi \beta |x|},
  \label{eq:BesselK12_formula2}
 \end{equation}
 which could be obtained from the Fourier transform of $\frac{1}{e^{\beta (p^2 - \mu)}-1}$ by expanding the exponential in the denominator to first order in $\beta$. The link between \eqref{eq:BesselK12_formula2} and \eqref{eq:BesselK12_formula} is the identity
\begin{equation}
    \frac{1}{\beta (p^2 - \mu)} = \int_0^\infty e^{- t \beta (p^2 - \mu)} \d t,
    \label{eq:propagator_parameter_integral}
\end{equation}
followed by Gaussian integration over $p$. The continuous variable $t$ in \eqref{eq:propagator_parameter_integral} plays the same role as the index $m$ in the geometric series expansion \eqref{eq:n_geometric}. 

\section{Regularized Fredholm determinants} \label{app:fredholm}

In this Appendix we review properties of Fredholm determinants used in the article. We refer to \cite{Simon77} for details. Let $\mathcal H$ be a separable Hilbert space and $A$ a trace class operator on $\mathcal H$. The Fredholm determinant $\det(1-A)$ can be defined by the absolutely convergent series
\begin{equation}
    \det(1-A)  = \sum_{j=0}^\infty (-1)^j \mathrm{Tr}[\Lambda^j A],
\end{equation}
where $\Lambda^j A$ is the $j$th tensor power of $A$ restricted to anti-symmetric tensors. The Fredholm determinant is multiplicative, $\det((1-A)(1-B)) = \det(1-A) \det(1-B)$, and can be expressed in terms of the eigenvalues $(\lambda_j)$ of $A$ in the expected way:
\begin{equation}
    \det(1-A) = \prod_{j} (1- \lambda_j).
\end{equation}

If the operator norm of $A$ satisfies $\| A \| <1$, there exists a principal branch logarithm of $1-A$ defined by the holomorphic functional calculus. In this case, $\det(1-A)= \exp (\mathrm{Tr}[\log(1-A)])$ and we have the following series representation:
\begin{equation}
    \det(1-A) = \exp \left(- \sum_{j=1}^\infty \frac{1}{j} \mathrm{Tr}[A^j] \right).
\end{equation}

Now let $k \geq 1$ be an integer and let $A$ be an operator on $\mathcal H$ in the $k$th Schatten class. One shows that the operator
\begin{equation}
    (1-A) \exp \left(  \sum_{j=1}^{k-1} \frac{1}{j} A^j \right)
    \label{eq:regularized1A}
\end{equation}
is of the form $1-B$ with $B$ trace class. Note that in the case $\| A \| <1$, for which $\log(1-A)$ is given by a convergent power series, the exponential term in \eqref{eq:regularized1A} cancels the first $k-1$ terms of $\log(1-A)$. However, the definition \eqref{eq:regularized1A} does not require $\| A \| <1$ because only finitely many terms of the series are involved. The Fredholm determinant of the operator in \eqref{eq:regularized1A} is called the $k$th regularized Fredholm determinant of~$A$:
\begin{equation} \label{def:regfreddet}
    \mathrm{det}_k(1-A) \coloneq   \det \left( (1-A) \exp \left(  \sum_{j=1}^{k-1} \frac{1}{j} A^j \right) \right).
\end{equation}
If $k>1$ and $A$ is in the $(k-1)$st Schatten class, the $k$th regularized determinant is related to the $(k-1)$st by
\begin{equation}
    \mathrm{det}_k(1-A) =   \mathrm{det}_{k-1}(1-A) \exp \left(  \frac{1}{k-1} \mathrm{Tr}[A^{k-1}] \right). 
\end{equation}
As shown in \cite{Simon77}, $\det_k(1-A)$ is a Lipchitz function of $A$ uniformly on bounded subsets of the $k$th Schatten space.

\end{document}